\documentclass[1p,11pt,envcountsect,envcountsame]{elsarticle}
\makeatletter
\def\ps@pprintTitle{%
 \let\@oddhead\@empty
 \let\@evenhead\@empty
 \def\@oddfoot{}%
 \let\@evenfoot\@oddfoot}
\makeatother

\usepackage{fullpage} 
\usepackage{microtype}
\usepackage{amsmath,amssymb,bm,amsfonts,amsbsy,pstricks,bbm,mathrsfs,graphicx}
\usepackage{bigstrut,multirow,fixmath,enumerate}
\usepackage[english]{babel} 
\usepackage{amsthm}

\usepackage{pb-diagram,fp,graphicx} 

\usepackage{pb-diagram}

\newcommand{\support}{\mathit{support}}

\newcommand{\unitdecompositions}{\mathbold{ud}}
\newcommand{\path}{\mathfrak{p}}
\newcommand{\arborealdecomposition}{\mathcal{D}}
\newtheorem{lemma}{Lemma}
\newtheorem{definition}{Definition}
\newtheorem{theorem}{Theorem}
\newtheorem{corollary}{Corollary}
\newtheorem{proposition}{Proposition}

\newtheorem{observation}{Observation}
\newcommand{\lastposition}{j}

\newcommand{\normalizingprojection}{{\mathbold{\eta}}}
\newcommand{\unweightingprojection}{{\mathbold{\zeta}}}
\newcommand{\inverse}{\mathbf{inv}}
\newcommand{\width}{w}
\newcommand{\leadingsymbol}{\mathbold{ls}}
\newcommand{\emptystring}{{\lambda}}
\newcommand{\positions}{\mathit{Pos}}
\newcommand{\vertexlabel}{\Gamma_1}
\newcommand{\edgelabel}{\Gamma_2}
\newcommand{\vertexlabeling}{\rho}
\newcommand{\edgelabeling}{\xi}
\newcommand{\terms}{\mathit{Ter}}
\newcommand{\newslicealphabet}{\mathbold{\Sigma}}
\newcommand{\treezigzagnumber}{{tzn}}

\newcommand{\olivetreedecomposition}{\mathcal{T}}
\newcommand{\identitySlice}{\mathbf{I}}
\newcommand{\arity}{\mathfrak{a}}

\newcommand{\node}{\mathbold{n}}
\newcommand{\boldT}{\mathbf{T}}

\newcommand{\zigzagnumber}{{zn}}

\newcommand{\interpretedAlphabet}{\newslicealphabet(c,\vertexlabel,\edgelabel,\mathcal{X})}
\newcommand{\interpretedAlphabetMinusX}{\newslicealphabet(c,\vertexlabel,\edgelabel,\mathcal{X}\backslash\{X\})}
\newcommand{\mso}{{$\mbox{MSO}$\;}}
\newcommand{\msotwo}{{$\mbox{MSO}_2$\;}}

\newcommand{\treeAutomaton}{{\mathcal{A}}}
\newcommand{\boldS}{{\mathbf{S}}}
\newcommand{\lang}{{\mathcal{L}}} 
\newcommand{\N}{{\mathbb{N}}} 
\newcommand{\Z}{{\mathbb{Z}}} 
\newcommand{\projection}{\mathbold{\pi}}
\newcommand{\extrawidth}{ew}
\newcommand{\astate}{{\mathfrak{q}}}

\newcommand{\composedTOne}{{
  \mathrel{\vbox{\offinterlineskip\ialign{%
    \hfil##\hfil\cr
    $\scriptscriptstyle\circ$\cr
    \noalign{\kern0.1ex}
    $\boldT$\cr
}}}}}

\newcommand{\composedT}{%
  \mathrel{\vbox{\offinterlineskip\ialign{%
    \hfil##\hfil\cr
    $\scriptscriptstyle\circ$\cr
    \noalign{\kern0.1ex}
    $\boldT$\cr
}}}}

\newcommand{\composedTprime}{%
  \mathrel{\vbox{\offinterlineskip\ialign{%
    \hfil##\hfil\cr
    $\scriptscriptstyle\circ$\cr
    \noalign{\kern0.1ex}
    $\boldT'$\cr
}}}}

\newcommand{\composedTplus}{%
  \mathrel{\vbox{\offinterlineskip\ialign{%
    \hfil##\hfil\cr
    $\scriptscriptstyle\oplus$\cr
    \noalign{\kern0.1ex}
    $\boldT$\cr
}}}}

\newcommand{\graph}{{\mathcal{G}}}

\newcommand{\emptyslice}{{\bm{\varepsilon}}}

\newcommand{\automataweightingfunction}{{\mathbf{w}}} 
\newcommand{\graphweightingfunction}{{\mu}} 

\begin{document}

\makeatletter
\def\dg@dotvector(#1,#2)#3{%
   \begingroup
   \dg@XTEMP=#1\relax \dg@YTEMP=#2\relax
   \let\dg@NDOTS=\dg@XEND \let\dg@DOTDIAM=\dg@WEND
   \dg@NDOTS=\unitlength \multiply\dg@NDOTS #3\relax
   \dg@ZTEMP=\dg@YTEMP \dg@changesign\dg@YTEMP\dg@ZTEMP
   \ifnum\dg@XTEMP>\z@
      \ifnum\dg@YTEMP>\dg@XTEMP
         \multiply\dg@NDOTS\dg@YTEMP \divide\dg@NDOTS\dg@XTEMP \fi
   \else\ifnum\dg@XTEMP<\z@
      \ifnum\dg@YTEMP>-\dg@XTEMP
         \multiply\dg@NDOTS\dg@YTEMP \divide\dg@NDOTS-\dg@XTEMP \fi
   \fi\fi
   \dg@YTEMP=\dg@ZTEMP
   \divide\dg@NDOTS\dgDOTSPACING
   \ifnum\dg@NDOTS>\z@\else \dg@NDOTS=\@ne \fi
   \dg@ZTEMP=\unitlength \multiply\dg@ZTEMP #3\relax
   \divide\dg@ZTEMP\dg@NDOTS
   \ifnum\dg@XTEMP=\z@
      \dg@changesign\dg@ZTEMP\dg@YTEMP \dg@YTEMP=\dg@ZTEMP
   \else
      \dg@changesign\dg@ZTEMP\dg@XTEMP
      \multiply\dg@YTEMP\dg@ZTEMP \divide\dg@YTEMP\dg@XTEMP
      \dg@XTEMP=\dg@ZTEMP
   \fi
   \divide\dg@XTEMP\unitlength \divide\dg@YTEMP\unitlength
\ifnum\the\dg@XTEMP>0%
  \FPdiv\arang{\the\dg@YTEMP}{\the\dg@XTEMP}%
  \FParctan\arang{\arang}%
  \FPmul\arang{\arang}{57.295}%
\else
  \def\arang{90}%
\fi
   \begin{picture}(0,0)%
      \dg@DOTDIAM=\dgDOTSIZE \divide\dg@DOTDIAM\unitlength
      \multiput(0,0)(\dg@XTEMP,\dg@YTEMP){\dg@NDOTS}{%
         \smash{\rotatebox{\arang}{\rule{2pt}{.5pt}}}}%
      \multiply\dg@XTEMP\dg@NDOTS \multiply\dg@YTEMP\dg@NDOTS
      \put(\dg@XTEMP,\dg@YTEMP){\vector(#1,#2){0}}%
   \end{picture}%
   \endgroup}%
\makeatother

\begin{frontmatter}

\title{An Algorithmic Metatheorem for Directed Treewidth}
\author{\vspace{-5pt}Mateus de Oliveira Oliveira}  
\address{Institute of Mathematics, Academy of Sciences of the Czech Republic\\mateus.oliveira@math.cas.cz}

\begin{abstract}
The notion of {\em directed treewidth} was introduced by Johnson, Robertson, Seymour and Thomas
[Journal of Combinatorial Theory, Series B, Vol 82, 2001] as a first step towards an algorithmic metatheory for digraphs.
They showed that some NP-complete properties such as Hamiltonicity can be decided in polynomial time on digraphs of constant directed treewidth. 
Nevertheless, despite more than one decade of intensive research, the list of hard combinatorial problems that are known to be solvable 
in polynomial time when restricted to digraphs of constant {\em directed} treewidth has remained scarce. In this work we enrich this 
list by providing for the first time an algorithmic metatheorem connecting the monadic second order logic of graphs to directed treewidth.
We show that most of the known positive algorithmic results for digraphs of constant directed treewidth can be reformulated in 
terms of our metatheorem. Additionally, we show how to use our metatheorem to provide polynomial time algorithms for two classes 
of combinatorial problems that have not yet been studied in the context of directed width measures. More precisely, for each fixed $k,w\in \N$,
we show how to count in polynomial time on digraphs of directed treewidth $w$, the number of minimum spanning strong subgraphs that are 
the union of $k$ directed paths, and the number of maximal subgraphs that are the union of $k$ directed paths and satisfy a given minor closed 
property. To prove our metatheorem we devise two technical tools which we believe to be of independent 
interest. First, we introduce the notion of tree-zig-zag number of a digraph, a new directed width measure that is at most 
a constant times directed treewidth. 
Second, we introduce the notion of $z$-saturated tree slice language, a new formalism for the specification and manipulation of 
infinite sets of digraphs. 
\end{abstract}

\begin{keyword}
Combinatorial Slice Theory \sep Directed Treewidth \sep Tree-Zig-Zag Number \sep  Monadic Second Order Logic of Graphs 
\sep Algorithmic Metatheorems
\end{keyword}

\end{frontmatter}

\section{Introduction}
\label{section:Introduction}

Since the introduction of {\em directed  treewidth} in \cite{Reed1999,JohnsonRobertsonSeymourThomas2001}
much effort has been devoted into trying to identify algorithmically useful digraph width measures. 
Such a width measure should ideally satisfy two properties. First, it should be small on several interesting instances of digraphs. Second,
many combinatorial problems should become polynomial time tractable on digraphs of constant width. While the first property is satisfied by most of the digraph width measures introduced so far 
\cite{Barat2006,BerwangerDawarHunterKreutzerObdrzalek2012,BerwangerGradel2004,BerwangerGradelKaiserRabinovich2012,GruberHolzer2008,HunterKreutzer2008,Reed1999,Safari2005}, 
the goal of identifying large classes of problems that can be solved in polynomial time when these 
measures are bounded by a constant has proven to be extremely hard to achieve. On the positive side, Johnson, Robertson, Seymour and Thomas 
showed already in their seminal paper \cite{JohnsonRobertsonSeymourThomas2001} that certain linkage problems, such as Hamiltonicity and $k$-disjoint 
paths (for constant $k$), can be solved in polynomial time on digraphs of constant directed treewidth. 
Subsequently, It was shown in \cite{DankelmannGutinKimJung2009} that for each constant $k\in \N$, one can decide in polynomial time 
the existence of a spanning tree with at most $k$ leaves on digraphs of constant directed treewidth. 
More recently, it was shown in \cite{BerwangerDawarHunterKreutzerObdrzalek2012} that determining the winner for some classes of parity games 
can be solved in polynomial time on digraphs of constant DAG-width \cite{BerwangerDawarHunterKreutzerObdrzalek2012}. 

In this work we enrich the list of problems that can be solved in polynomial time on digraphs of constant 
directed treewidth. More precisely, we devise the first algorithmic metatheorem connecting {\em directed} treewidth 
to the monadic second order logic of graphs with edge set quantifications (\msotwo logic). 
We show that most of the positive algorithmic results obtained so far on digraphs of constant {\em directed} treewidth can be 
reformulated in terms of our metatheorem. Additionally we show how to use our metatheorem to provide polynomial time 
algorithms for a parameterized version of the minimum spanning strong subgraph problem, and for a parameterized version 
of the problem of counting subgraphs satisfying a given minor closed property. 

We note that celebrated results due to Courcelle \cite{Courcelle1990MSO} and 
Arnborg, Lagergren and Seese \cite{ArnborgLagergrenSeese1991} state that any problem expressible in 
\msotwo logic can be solved in linear time on graphs of constant {\em undirected} treewidth. Additionally, 
an equally famous result due to Courcelle, Makowsky and Rotics states that any problem expressible in 
\mso logic (without edge set quantifications) can be solved in linear time on graphs of constant clique-width \cite{CourcelleMakowskyRotics2000}. 
However, we observe that there are families of digraphs of constant {\em directed} treewidth, but simultaneously 
unbounded {\em undirected} treewidth and clique-width \cite{CourcelleMakowskyRotics2000}. 
For instance, the $n\times n$ grid, in which all horizontal edges are oriented to the right and all vertical edges are oriented upwards,  
has {\em directed} $\mbox{treewidth $0$}$, but {\em undirected} treewidth $\Theta(n)$ and clique-width $\Theta(n)$.
Thus our algorithmic metatheorem is not implied by the results in \cite{Courcelle1990MSO,ArnborgLagergrenSeese1991,CourcelleMakowskyRotics2000}. On the 
other hand, the fact that $3$-colorability  is MSO expressible 
implies that a complete analog of Courcelle's is theorem for digraphs of 
constant directed treewidth cannot be achieved unless P=NP, since $3$-colorability is already NP-complete on DAGs.

Before stating our main theorem we will introduce some notation. 
An edge-weighting function for a digraph $G=(V,E)$ is a function $\graphweightingfunction:E\rightarrow \Omega$ where $\Omega$ 
is a finite commutative semigroup of size polynomial in $|V|$. We will always assume that $\Omega$
has an identity element. We define the size of $G$ as $|G|= |V|+|E|$.  
 The weight of a subgraph $H=(V',E')$ of $G$ is defined as $\graphweightingfunction(H)=\sum_{e\in E'} \graphweightingfunction(e)$. 
We say that $H$ is the union of $k$ directed paths if there exist directed simple paths 
$\path_1,\path_2,...,\path_k$ with $\path_i=(V_i,E_i)$ for $i\in \{1,...,k\}$ such that 
$H=\path_1\cup \path_2\cup ...\cup \path_k = (\cup_{i=1}^k V_i,\cup_{i=1}^k E_i)$. 
We note  that the unions we consider are not necessarily vertex-disjoint nor edge-disjoint.

\begin{theorem}[Main Theorem]
\label{theorem:MainTheoremDirectedTreewidth}
Let $\varphi$ be an \msotwo sentence and let $k,w\in \N$. There is a computable function 
$f(\varphi,w,k)$ such that, given a weighted digraph $G=(V,E,\graphweightingfunction\!:\!E\rightarrow \Omega)$ of directed treewidth $w$, 
a positive integer $l<|V|$, and an element $\alpha \in \Omega$, one can count in time $f(\varphi,w,k)\cdot |G|^{O(k\cdot(w+1))}$ the number of subgraphs $H$ of $G$ 
simultaneously satisfying the following four properties: 
\begin{enumerate}[(i)]
	\item  $H\models \varphi$,
 	\item  $H$ is the union of $k$ directed paths, 
	\item  $H$ has $l$ vertices,
	\item\label{item:Weight}  $H$ has weight $\graphweightingfunction(H)=\alpha$. \newline
\end{enumerate}
\end{theorem}

We note that in \cite{deOliveiraOliveira2013IPEC} we proved an analog theorem for digraphs of constant {\em directed} pathwidth.
Nevertheless it can be shown that there exist families of digraphs of constant directed treewidth but unbounded {\em directed} 
pathwidth \cite{BerwangerDawarHunterKreutzerObdrzalek2012}. Therefore, Theorem \ref{theorem:MainTheoremDirectedTreewidth} is a strict generalization of the results in \cite{deOliveiraOliveira2013IPEC}. 
To prove Theorem \ref{theorem:MainTheoremDirectedTreewidth} we will introduce two new technical tools which may be of independent 
interest. The first, the tree-zig-zag number of a digraph, is a new directed width measure that is at most a constant times directed 
treewidth. The second, the notion of $z$-saturated tree slice languages, is a new framework for the manipulation of infinite families of digraphs.

\subsection{Applications}
\label{section:Applications}

The parameters $l$ and $\alpha$ in Theorem \ref{theorem:MainTheoremDirectedTreewidth} are upper bounded by
$|V|^{O(1)}$. By varying these parameters we can consider different flavours of optimization problems. 
For instance, we can choose to count the number of subgraphs of $G$ that are the union of $k$ directed paths, satisfy 
$\varphi$ and have maximal/minimal number of vertices, or maximal/minimal weight. In this section we  provide a list of 
natural combinatorial problems that can be solved in polynomial time on digraphs of constant directed treewidth 
using Theorem \ref{theorem:MainTheoremDirectedTreewidth}. In Subsection \ref{subsection:ApplicationsKnown}, we show
how to use Theorem \ref{theorem:MainTheoremDirectedTreewidth} to rederive three known  positive algorithmic results 
for digraphs of constant directed treewidth. In Subsection \ref{subsection:ApplicationsNotKnown}, we  show 
how Theorem \ref{theorem:MainTheoremDirectedTreewidth} can be used to solve in polynomial time two interesting classes 
of combinatorial problems  which have not yet been studied in the context of digraph width measures. Concerning the 
first class of problems, we 
show how to count the number of {\em minimum spanning strong subgraphs} that are the union of $k$ directed paths. Concerning the second class, 
we show how to count the number of maximal subgraphs that are the union of $k$ directed paths and satisfy some given minor closed property.

\subsubsection{First Examples}
\label{subsection:ApplicationsKnown}

In order to use Theorem \ref{theorem:MainTheoremDirectedTreewidth} to solve a counting problem 
in polynomial time, we need to exhibit an \msotwo sentence $\varphi$ specifying a suitable class 
of digraphs to be counted, and to specify values for the parameters $l$ and $\alpha$ which respectively determine 
the number of vertices and the weight of the subgraphs being counted. We observe that the 
class of digraphs specified by $\varphi$ is fixed and does not vary with the input digraph. 
The parameters $l$ and $\alpha$ on the other hand, may vary with the input. 

\paragraph{\textbf{Counting Hamiltonian Cycles}}
We set $\varphi$ to be an \msotwo sentence defining cycles, i.e., connected digraphs in which each vertex has precisely one 
incoming edge and one outgoing edge. We set $l=|V|$ since we are only interested in counting sub-cycles of $G$ that span all 
of its vertices. Finally, since any cycle is the union of $2$ directed paths, we set $k=2$.
We observe that counting Hamiltonian cycles on digraphs of constant directed treewidth can also be done 
via an adaptation of the techniques in \cite{JohnsonRobertsonSeymourThomas2001}. 

\paragraph{\textbf{Counting $\sigma$-Linkages}}

Given a sequence $\sigma=  (s_1,t_1,s_2,t_2,...,s_k,t_k)$ of $2k$ not necessarily distinct vertices, a $\sigma$-linkage
is a set of internally disjoint directed paths $\path_1,\path_2,...,\path_k$
where for each $i\in \{1,...,k\}$, the path $\path_i$ connects $s_i$ to $t_i$. To count the number of $\sigma$-linkages
on a digraph $G$ we first assign a distinct color to each vertex in the set 
$\{s_1,...,s_k,t_1,...,t_k\}$ and assume that all other vertices of $G$ are uncolored. 
Then we define an \msotwo sentence $\varphi_{\sigma}$ that is true in a digraph $H$ whenever 
it consists of the union of $k$ internally disjoint paths $\path_1,...,\path_k$ 
where for each $i$, the path $\path_i$ connects a vertex of color $c(s_i)$ to a vertex of color $c(t_i)$ in such 
a way that all internal vertices of $\path_i$ are uncolored. 
For each $l\in \{1,...,|V|\}$ we can use Theorem \ref{theorem:MainTheoremDirectedTreewidth} to 
count the number of $\sigma$-linkages of size $l$. We observe that counting $\sigma$-linkages 
can  also be done by via an adaptation of the results in \cite{JohnsonRobertsonSeymourThomas2001}.

\paragraph{\textbf{Counting Spanning-Out Trees with at most $k$-leaves}}
A spanning-out tree is a spanning tree in which all edges are directed towards the leaves.
To count the number of spanning-out trees with at most $k$-leaves we set 
$\varphi$ to be an \msotwo sentence defining trees with at most $k$-leaves. In other words, $\varphi$ defines 
connected digraphs without cycles in which at most $k$ vertices have no out-going edge. Since the tree has to span 
all vertices of $G$, we set $l=|V|$. Finally, we note that any spanning-out tree with at most $k$ leaves is the union of $k$ directed paths. 
We observe that counting spanning-out trees can also be done via an adaptation of the results in 
in \cite{DankelmannGutinKimJung2009}.

\subsubsection{New Applications}
\label{subsection:ApplicationsNotKnown}

In this section we exhibit two natural classes of counting problems that can be 
solved in polynomial time on digraphs of constant directed treewidth using Theorem 
\ref{theorem:MainTheoremDirectedTreewidth}. To the best of our knowledge these 
problems cannot be addressed in polynomial time using previously existing techniques.

\paragraph{\bf Minimum Spanning Strong Subgraph} 
The classic {\em Minimum Spanning Strong Subgraph (MSSS) problem} is defined as follows. 
Given a strongly connected digraph $G$, find a spanning strongly connected subgraph of $G$ with 
the minimum number of edges. This problem is in general NP complete since it generalizes 
the Hamiltonian cycle problem. Even though the MSSS problem has received a considerable amount of 
attention \cite{Aho1972transitive,BangHuangYeo2003Strongly,BessyThomasse2003,BangYeo2001Minimum,Gabow2004special,Vetta2001Approximating}, 
the connections between this problem and directed width measures are, to the limit of our knowledge, unexplored. Here we show 
that a parameterized version of the MSSS problem can be solved in polynomial time on digraphs 
of constant directed treewidth. A $k$-MSSS is a minimum spanning strong subgraph that is the union of $k$ directed paths. 
We note that determining the existence of a $\mbox{$k$-MSSS}$ on general digraphs is still NP-complete for each constant $k\geq 2$, 
since any Hamiltonian cycle is a $2$-MSSS. Using Theorem \ref{theorem:MainTheoremDirectedTreewidth}
we can not only determine the existence of a $k$-MSSS on digraphs of constant directed treewidth,
but also count in polynomial time the number of occurrences of such subgraphs. All we need to do is to set $l=|V|$, since the subgraphs 
we are counting are spanning, and to set $\varphi_{\mathit{str}}$ as the monadic second order sentence  
that is true in a digraph if and only if it is strongly connected. 

We observe that the techniques in \cite{JohnsonRobertsonSeymourThomas2001} cannot be directly applied 
to solve the $k$-MSSS problem in polynomial time due to the fact that the $k$ paths covering a
$k$-MSSS need not to be internally disjoint. For instance, in Figure \ref{figure:MSSS} we show a family $H_1,H_2,...$ 
of digraphs where for each $n\in \N$, $H_n$ is the union of $2$ paths. Note however that one needs 
$2n$ internally disjoint paths to cover all vertices and edges of $H_n$.

\begin{figure}[!hb] 
\centering 
\includegraphics[scale=0.30]{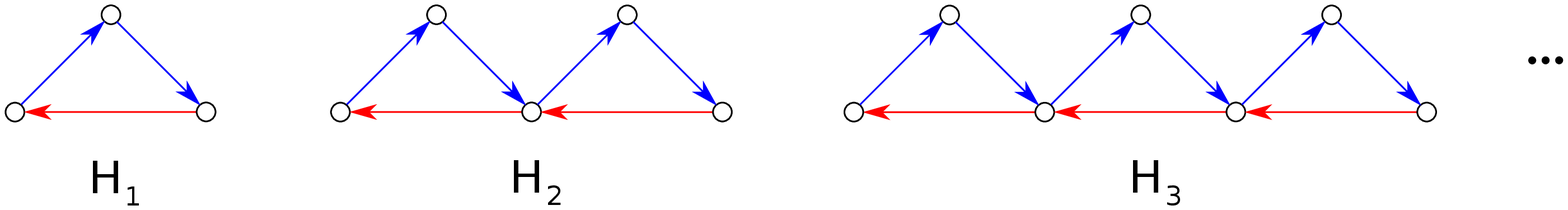} 
\vspace{-5pt}
\caption{For each $n\in \N$, the digraph $H_n$ is the union of $2$ paths. On the other hand, 
$2n$ internally disjoint paths are necessary to cover all vertices and edges of $H_n$.}
\label{figure:MSSS}
\end{figure}

\vspace{-5pt}
\paragraph{\bf Subgraphs Excluding a Minor}

An undirected graph $H$ is a minor of an undirected graph $G$ if $H$
can be obtained from a subgraph of $G$ by a sequence of edge contractions. 
A family of undirected graphs $\mathcal{F}$ is said to be minor closed if whenever 
a graph $G$ belongs to $\mathcal{F}$, any minor of $G$ is also in $\mathcal{F}$. 
Many interesting graph families are minor closed, such as,  planar graphs, outerplanar graphs, 
graphs of bounded genus, forests, series-parallel graphs, graphs of bounded undirected treewidth, etc. 
Given a minor closed family $\mathcal{F}$ and a graph $G$ it is often NP-complete to find 
a maximal subgraph of $G$ that belongs to the family $\mathcal{F}$. For instance, the following 
problems are NP-complete: finding a maximal outerplanar subgraph
\cite{Cimikowski1996Sizes,Poranen2005Heuristics}, finding a maximal planar subgraph
\cite{JungerMutzel1996,Cualinescu1998} and finding a maximal subgraph of a given genus $g$ \cite{Cualinescu1996Finding}. 

By the celebrated graph minor theorem of Robertson and Seymour \cite{RobertsonSeymour2004}, for any minor closed 
family of undirected graphs $\mathcal{F}$ there exists a finite set of undirected graphs $\hat{\mathcal{F}}$,
such that for each graph $H$, $H\in \mathcal{F}$ if and only if none of the graphs in $\hat{\mathcal{F}}$
is a minor of $H$. Thus, using the finite set $\hat{\mathcal{F}}$ one can define an \msotwo sentence
$\varphi_{\mathcal{F}}$ such that $\varphi_{\mathcal{F}}$ is true in a graph $H$ if and only if $H\in \mathcal{F}$ (see for instance \cite{CourcelleEngelfriet2012}). 
This fact implies that Theorem \ref{theorem:MainTheoremDirectedTreewidth} can be used to count in polynomial time, 
on digraphs of constant directed treewidth, the number of subgraphs that are the union of $k$ directed 
paths and whose underlying undirected graph satisfy a minor closed property. More precisely, 
if $H$ is a directed graph, let $\stackrel{\leftrightarrow}{H}$ denote the undirected graph obtained from $H$ by 
forgetting the directions of the edges in $H$. We have the following corollary of Theorem \ref{theorem:MainTheoremDirectedTreewidth}.

\begin{corollary}
\label{corollary:StructuralCounting}
Let $\mathcal{F}$ be a minor closed family of undirected graphs, $G$ be a digraph of directed treewidth $w$, and let $k\in \N$. 
Then one can count in time $f(\varphi_{\mathcal{F}},k,w)\cdot |G|^{O(k\cdot (w+1))}$ the number of (maximal) subgraphs $H$
of $G$ subject to the following restrictions: 
\begin{enumerate}
	\item $H$ is the union of $k$ directed paths.
	\item $\stackrel{\leftrightarrow}{H}$ belongs to $\mathcal{F}$.
\end{enumerate}
\end{corollary}

For instance, Corollary \ref{corollary:StructuralCounting} implies that we can count in polynomial time, on digraphs of constant directed 
treewidth, maximal planar subgraphs that are the union of $k$ directed paths, or maximal subgraphs that are the 
union of $k$ directed paths and can be embedded on a torus. In our opinion it is rather surprising that the problems addressed in 
Corollary \ref{corollary:StructuralCounting} can be solved in polynomial time, in view of the complexity of the subgraphs 
that are being counted, and in view of the fact that digraphs of constant directed treewidth may have simultaneously unbounded 
{\em undirected} treewidth and clique-width.

\subsection{Hardness Results}

We argue briefly that under the assumption that the \textbf{W} hierarchy does not collapse to \textbf{FPT}, a widely believed assumption in parameterized 
complexity theory \cite{DowneyFellows1992}, the dependence on $w$ and 
on $k$ on the exponent of the running time $f(\varphi,w,k)\cdot |G|^{O(k\cdot (w+1))}$ of Theorem \ref{theorem:MainTheoremDirectedTreewidth} cannot be removed.
Concerning the dependence on $k$, we note that the problem of determining whether there exists 
$k$ disjoint paths on DAGs from prescribed pairs of nodes is \textbf{W[1]} hard with respect to $k$ \cite{Slivkins2003}. Since any 
DAG has {\em directed} treewidth $0$, we have that the existence of $k$-disjoint paths is already \textbf{W[1]}-hard even for digraphs of directed treewidth $0$.
Concerning the dependence on $w$, we note that it can be shown \cite{LampisKaouriMitsou2011} that finding Hamiltonian paths on digraphs is \textbf{W[2]} 
hard with respect to the cycle-rank of the digraph in question. Since constant directed treewidth is more expressive than constant 
cycle rank, the hardness results in \cite{LampisKaouriMitsou2011} extends to {\em directed} treewidth. Thus the dependence on $w$ in the 
exponent of the running time $f(\varphi,w,k)\cdot |G|^{O(k\cdot (w+1))}$ cannot be removed even if $k=1$.

\subsection{Proof Techniques and Organization of the Paper}
\label{subsection:ProofTechniques}

We will prove Theorem \ref{theorem:MainTheoremDirectedTreewidth} using slice theoretic techniques. 
The notion of slice language was introduced in \cite{deOliveiraOliveira2010}
and used to solve several problems in the partial order theory of concurrency. 
Subsequently, slice languages were lifted to the context of digraphs and used to provide the first algorithmic 
metatheorem for digraphs of constant {\em directed} pathwidth \cite{deOliveiraOliveira2013IPEC}. 
In this work we extend the results in \cite{deOliveiraOliveira2013IPEC} by introducing the 
notions of {\em tree slice language} and {\em slice tree automata}. We use {\em tree} slice-languages to provide the 
first algorithmic metatheorem for digraphs of constant {\em directed} treewidth (Theorem \ref{theorem:MainTheoremDirectedTreewidth}).
More precisely, we will show that the problem of counting the number of subgraphs satisfying 
the conditions ($i$)-($iv$) of Theorem \ref{theorem:MainTheoremDirectedTreewidth} can be reduced to 
the problem of counting the number of terms accepted by a suitable deterministic slice tree-automaton.
We note that the results in this work strictly generalize the results in \cite{deOliveiraOliveira2013IPEC},
since there are families of digraphs of constant {\em directed} treewidth but unbounded {\em directed} pathwidth. 
Below we give a brief description of the main technical tools used in this paper and how they fit 
together to yield a proof of Theorem \ref{theorem:MainTheoremDirectedTreewidth}. All notions introduced 
in the following paragraphs will be re-defined more carefully along the paper.

A unit slice of arity $r$ is a digraph $\boldS$ whose vertex set is partitioned into a center $C$, 
an out-frontier $F_0$ and $r$ in-frontiers $F_1,...,F_r$ in such a way that the center $C$ has at 
most one vertex and each frontier vertex is incident with precisely one edge of $\boldS$ (Figure \ref{figure:Slices}). 
Intuitively a slice $\boldS$ 
can be glued to a slice $\boldS'$ at frontier $j$ if the out-frontier of $\boldS$ can be matched with the $j$-th in-frontier of 
$\boldS'$. A finite set $\newslicealphabet$ of slices with possibly distinct arities is called a slice alphabet. In this paper we will only 
be interested in slices of arity $0$, $1$ and $2$. A term over $\newslicealphabet$ 
is a tree-like expression $\boldT$ in which each node $p$ is labeled with 
a slice $\boldT[p]$ whose arity is equal to the number of children of $p$. We say that $\boldT$ is a unit decomposition if 
for each position $pj$ the slice $\boldT[pj]$ can be glued to the slice $\boldT[p]$ at its $j$-th frontier (Figure \ref{figure:UnitDecomposition}). 

Each unit decomposition $\boldT$ gives rise to a digraph $\composedT$ which is intuitively obtained by glueing 
each two adjacent slices in $\boldT$ along their matching frontiers. Conversely, for each digraph $G$ there is a suitable 
slice alphabet $\newslicealphabet$ and unit decomposition $\boldT$ over $\newslicealphabet$ such that $\composedT$ is 
isomorphic to $G$. We can represent infinite families of digraphs via tree-automata over slice alphabets. We say that 
such an automaton $\treeAutomaton$ is a slice tree-automaton if all terms generated by $\treeAutomaton$ are 
unit decompositions. With a slice tree-automaton $\treeAutomaton$ one can associate two types of languages. The first, 
the slice language $\lang(\treeAutomaton)$, is simply 
the set of all unit decompositions accepted by $\treeAutomaton$. The second, the graph language $\lang_{\graph}(\treeAutomaton)$ 
is the set of all digraphs represented by unit decompositions in $\lang(\treeAutomaton)$.

We say that a unit decomposition $\boldT$ has tree-zig-zag number $z$ if each {\em simple} path in the 
digraph $\composedT$ represented by $\boldT$ crosses each frontier of each slice in $\boldT$ at most $z$ times $\mbox{(Figure \ref{figure:UnitDecomposition})}.$
A slice tree-automaton $\treeAutomaton$ has tree-zig-zag number $z$ if each unit decomposition $\boldT\in \lang(\treeAutomaton)$
has tree-zig-zag number $z$. Finally, we say that a slice tree-automaton $\treeAutomaton$ is $z$-saturated over a slice alphabet $\newslicealphabet$
if the presence of a digraph $H$ in the graph language $\lang_{\graph}(\treeAutomaton)$ implies that each 
unit decomposition $\boldT$ of tree-zig-zag number $z$ representing $H$ belongs to $\lang(\treeAutomaton)$. 
The importance of the notion of saturation stems from the following fact. Given a slice tree-automaton 
$\treeAutomaton$ of tree-zig-zag number $z$ and a slice tree-automaton $\treeAutomaton'$ that 
is $z$-saturated, it is possible to show that $\lang_{\graph}(\treeAutomaton\cap \treeAutomaton') = \lang_{\graph}(\treeAutomaton)\cap \lang_{\graph}(\treeAutomaton')$. 
In other words, the set of digraphs represented by the intersection $\treeAutomaton\cap \treeAutomaton'$ is equal 
to the intersection of the sets of digraphs represented by $\treeAutomaton$ and $\treeAutomaton'$ separately. 
We note that this crucial property is not satisfied by general slice tree-automata. 
Within this framework, the proof of Theorem \ref{theorem:MainTheoremDirectedTreewidth} can be divided into the following 
steps. 

\begin{enumerate}
	\item In the first step, we will show that given a digraph $G$ of directed treewidth $w$ 
	one can construct a unit decomposition $\boldT$ of $G$ of tree-zig-zag number $z\leq 9w+18$. This construction 
	will follow from a combination of Theorem \ref{theorem:ComparisonWithOtherMeasures} with 
	Proposition \ref{proposition:OliveTreeDecompositionUnitDecomposition}.  Subsequently, we will show that using 
	$\boldT$ one can construct a slice tree-automaton $\treeAutomaton(\boldT,k\cdot z)$ of tree-zig-zag number $z$ 
	whose graph language contains all subgraphs of $G$ that are the union of $k$ directed paths. The construction of 
	$\treeAutomaton(\boldT,k\cdot z)$ will be given in the proof of Lemma \ref{lemma:SubgraphsC}.  
	\item In the second step we will show that given an \msotwo sentence $\varphi$ and an integer $k$, one can automatically 
	construct a $z$-saturated slice tree-automaton $\treeAutomaton(\varphi,k,z)$ whose graph language $\lang_{\graph}(\treeAutomaton(\varphi,k,z))$
	consists precisely of the digraphs which at the same time satisfy $\varphi$ and are the union of $k$ directed paths 
	(Theorem \ref{theorem:MonadicSliceTreeAutomataZSaturated}).
	Additionally, given a positive integer $k\in \N$ and a weight $\alpha \in \Omega$, we can use $\treeAutomaton(\varphi,k,z)$
	to construct another $z$-saturated tree-automaton $\treeAutomaton(\varphi,k,z,l,\alpha)$ whose graph language contains 
	only those digraphs generated by $\treeAutomaton(\varphi,k,z)$ which have $l$ vertices and weight $\alpha$ (Lemma \ref{lemma:AutomatonSizeWeight}).
	\item Finally, in the third step we will show that the slice language of the tree-automaton 
	$\treeAutomaton(\boldT,k\cdot z)\cap \treeAutomaton(\varphi,k,z,l,\alpha)$ has precisely one unit decomposition 
	$\boldT$ for each subgraph of $G$ that is the union of $k$ directed paths, 
	satisfy $\varphi$ and have prescribed length $l$ and weight $\alpha$. This claim will follow from Lemma \ref{lemma:IntersectionSubdecompositions} 
	using the fact that $\treeAutomaton(\boldT,k\cdot z)$ has tree-zig-zag number $z$ and that $\treeAutomaton(\varphi,k,z,l,\alpha)$ is $z$-saturated. 
	At this point, the problem of counting subgraphs of $G$ satisfying these four properties boils down to the problem of 
	counting the number of unit decompositions accepted by $\treeAutomaton(\boldT,k\cdot z)\cap \treeAutomaton(\varphi,k,z,l,\alpha)$.  
	Since the latter automaton is deterministic, this counting process can be carried in polynomial time. 
	This step will be carried in Theorem \ref{theorem:CountingSizeWeight} via an application of Theorem \ref{theorem:CountingSubgraphs}.
\end{enumerate}

The remainder of this paper is organized as follows. Next, in Section \ref{section:DirectedTreewidth} we recall 
the definition of directed treewidth \cite{JohnsonRobertsonSeymourThomas2001}. Subsequently, in 
Section \ref{section:TreeZigZagNumber}, we introduce the {\em tree-zig-zag} number of a digraph, a new directed width 
measure. 
In Section \ref{section:ComparisonDirectedTreewidth}, we show that the tree-zig-zag number of a 
digraph is at most a constant times its directed treewidth. In Section \ref{section:TreeAutomata} we recall 
some of the main definitions of tree-automata theory.  
In Section \ref{section:TreeSliceLanguage} we introduce tree slice languages and 
slice tree-automata. In Section \ref{section:zSaturationAndCounting} we introduce the notion of $z$-saturation 
and state a slice theoretic metatheorem (Theorem \ref{theorem:CountingSubgraphs}).
 In Section \ref{section:MSOandTreeSliceLanguages}
we will show that for any \msotwo sentence $\varphi$ and any $k,z\in \N$ one can construct a $z$-saturated slice automaton
$\treeAutomaton(\varphi,k,z)$ whose graph language consists of all digraphs that are the union of $k$ directed paths and satisfy $\varphi$. 
In $\mbox{Section \ref{section:ProofOfMainTheorem}}$ we will show how to restrict $\treeAutomaton(\varphi,k,z)$ into an automaton
$\treeAutomaton(\varphi,k,z,l,\alpha)$ whose graph language consists precisely of the digraphs that, at the same time, are the union of 
$k$ paths, satisfy $\varphi$, have $l$ vertices and weight $\alpha$. In the same section we prove our main theorem, 
Theorem \ref{theorem:MainTheoremDirectedTreewidth}. Finally, in Section \ref{section:Conclusion} we make 
some final remarks and discuss some future directions.

\section{Directed Treewidth}
\label{section:DirectedTreewidth}

In this section we recall the definitions of arboreal decomposition and {\em directed} treewidth. For a matter of uniformity with other notions of 
tree decompositions encountered in this paper, our notation slightly differs from the notation used in \cite{JohnsonRobertsonSeymourThomas2001}. 
Let $\{1,...,r\}^*$ denote the set of all strings over $\{1,...,r\}$ and let $\emptystring$ denote 
the empty string. A subset
$N\subseteq \{1,...,r\}^*$ is prefix closed if for every $p\in \{1,...,r\}^*$ and every $j\in \{1,...,r\}$, $pj\in N$ implies that $p\in N$.
We note that the empty string $\emptystring$ is an element of any prefix closed subset of $\{1,...,r\}^*$.
We say that $N\subseteq \{1,...,r\}^*$ is well numbered if for every $p\in \{1,...,r\}^*$ and every $j\in \{1,...,r\}$, 
the presence of $pj$ in $N$ implies that $p1,...,p(j-1)$ also belong to $N$.
An $r$-ary tree is a pair $T=(N,F)$ whose set of nodes $N$ is a finite prefix 
closed, well numbered subset of $\{1,...,r\}^*$, and whose set of arcs $F$ is defined as $F = \{(p,pj)\;|\; p,pj\in N, j \in \{1,...,r\}\}$.
Observe that by our definition, the root of an $r$-ary tree is the empty string $\emptystring$. 
A binary tree is an $r$-ary tree in which $r=2$. 
If $pj\in N$, then we say that 
$pj$ is a child of $p$, or interchangeably, that $p$ is the parent of $pj$.  A {\em leaf} is a node $p\in N$ without children. 
 If $pu\in N$ for $u\in \{1,...,r\}^*$, then we say that 
$pu$ is a descendant of $p$. For a node $p\in N$ we let $N(p) =\{pu\in N\; |\; u\in \{1,...,r\}^*\}$
denote the set of all descendants of $p$. Note that $p$ is a descendant of itself and therefore $p\in N(p)$.

Let $G=(V,E)$ be a digraph and let $Z$ and $K$ be two disjoint subsets of vertices of $G$. 
We say that $K$ is $Z$-normal if there is no directed walk in $V\backslash Z$ with first and 
last vertex in $K$ that uses a vertex of $V\backslash (Z\cup K)$. In other words, $K$ is $Z$-normal if every walk which 
starts and ends in $K$ is either wholly contained in $K$ or uses a vertex of $Z$.

An arboreal decomposition of a digraph $G=(V,E)$ is a four-tuple $\arborealdecomposition = (N,F,W,Z)$ 
where $(N,F)$ is an $r$-ary tree for some $r\in \N$,\; $W:N\rightarrow 2^{V}$ is a function that associates 
with each node $p\in N$ a {\em non-empty} set of vertices $W(p)\subseteq V$, and $Z:F\rightarrow 2^{V}$ is a function that 
associates with each arc $(p,pj)\in F$, a set of vertices $Z(p,pj)$. In the sequel, we may refer to the 
sets $W(p)$ as the {\em bags} of $\arborealdecomposition$. For a node $p\in N$ we let 
$V(p,\arborealdecomposition)  = \bigcup_{u\in N(p)} W(u)$ denote the set of all vertices of $G$ that belong 
to some bag associated with a descendant of $p$. 
The functions $W$ and $Z$ satisfy the following two properties. 

\begin{enumerate}[1)]
	\itemsep0.2em	
	\item \label{item:ArborealOne} $\{W(p)\;|\; p\in N\}$ is a partition of $V$ into non-empty sets.
	\item \label{item:ArborealTwo} For each $(p,pj)\in F$, the set $V(pj,\arborealdecomposition)$ is $Z(p,pj)$-normal. 
\end{enumerate}

Intuitively, Condition \ref{item:ArborealTwo} says that for each $(p,pj)\in F$, the set of all vertices of 
$G$ that belong to bags associated with descendants of $pj$ is $Z(p,pj)$ normal.
If $e$ is an arc in $F$ and $p$ is a node in $N$ then we write $e\sim p$ to indicate that $p$
is incident with $e$. In other words, $e\sim p$ means that either $e=(p,p')$ or $e=(p',p)$ for some $p'\in N$. 
The width $w(\arborealdecomposition)$ of the arboreal decomposition $\arborealdecomposition$ is the least 
integer $w$ such that for every node $p\in N$, $|W(p) \cup \bigcup_{e\sim p} Z(e)| \leq w+1$. 
The {\em directed treewidth} of $G$ is the 
least integer $w$ such that there is an arboreal decomposition of $G$ of width $w$. An arboreal 
decomposition $\arborealdecomposition=(N,F,W,Z)$ of a digraph $G$ is {\em good} if additionally the following condition is satisfied. 
\begin{enumerate}[3)]
	\item For each position $p\in N$, if $pi\in N$ and $pj\in N$ with $i<j$, then there is no edge in $G$ with source in $V(pj,\arborealdecomposition)$
		and target in $V(pi,\arborealdecomposition)$. 
\end{enumerate}

A haven of order $w$ in a digraph $G=(V,E)$ is a function $\beta$ that assigns to each set $Z\subseteq V$ with $|Z|<w$, 
the vertex-set of a strongly connected component of the digraph $G\backslash Z$, in such a way 
that for each two sets of vertices $Z,Z' \subseteq V$, if  $Z'\subseteq Z$ with $|Z|< w$, then $\beta(Z)\subseteq \beta(Z')$. 
It can be shown that if $G$ has a haven of order $w$ then its directed treewidth 
is at least $w-1$. Theorem $3.3$ of reference \cite{JohnsonRobertsonSeymourThomas2001} states that a digraph $G$ either has 
directed treewidth at most $3w-2$, or it has a haven of order $w$. The proof of this theorem is algorithmic. The algorithm 
either constructs a good arboreal decomposition of $G$ of width $3w-2$ or declares that $G$ has a haven of order $w$. Since 
a haven of order $w$ is a certificate that the directed treewidth of $G$ is at least $w-1$, one can be sure that 
if the directed treewidth of $G$ is at most $w-2$, a good arboreal decomposition of $G$ of width at most $3w-2$ will be found.
Equivalently, if $G$ has directed treewidth at most $w$ then one can always find an arboreal decomposition for $G$ of width 
at most $3w+4$. 

\begin{theorem}[\cite{JohnsonRobertsonSeymourThomas2001}]
\label{theorem:ConstructionGoodArborealDecomposition}
Let $G$ be a digraph of directed treewidth at most $w$. One can construct in time $|G|^{O(w)}$ a good arboreal 
decomposition of $G$ of width at most $3w+4$. 
\end{theorem}

\section{Olive-Tree Decompositions and the Tree-Zig-Zag Number of a Digraph}
\label{section:TreeZigZagNumber}

In this section we will introduce the tree-zig-zag number of a digraph, a new directed width measure. 
Next, in Section \ref{section:TreeZigZagNumber} we will show that the tree-zig-zag number of a digraph 
is at most a constant times its directed treewidth. 

\begin{definition} 
\label{definition:OliveTreeDecomposition}
An {\em olive-tree decomposition} of a digraph $G=(V,E)$ is a triple $\olivetreedecomposition = (N,F,\mathfrak{m})$ where
$(N,F)$ is a binary tree and $\mathfrak{m}:V\rightarrow N$ is an injective map from vertices of $G$ to nodes of $T$. 
\end{definition}  

The notion of olive-tree decomposition is similar to the notion of carving decomposition introduced 
by Seymour and Thomas in \cite{SeymourThomas1994}. The only difference is that in a carving decomposition, as defined in \cite{SeymourThomas1994}, the vertices
of the digraph $G$ are bijectively mapped to the leaves of the tree, while in our definition these vertices can also be mapped to 
the internal nodes of the tree, and the mapping is required to be injective, but not necessarily bijective. 
If $\olivetreedecomposition = (N,F,\mathfrak{m})$ is an olive-tree decomposition of a digraph $G=(V,E)$ 
then we let $V(p,\olivetreedecomposition)= \mathfrak{m}^{-1}(N(p))$ denote the set of all vertices of $G$ that are mapped to some
descendant of $p$.
If $V_1,V_2\subseteq V$ are two subsets of vertices of $G$ with $V_1\cap V_2 =\emptyset$, then 
we let $E(V_1,V_2)$ denote the set of all edges of $G$ with one endpoint in $V_1$ and another endpoint in $V_2$. 
The width $w(p)$ of a node $p\in N$ is defined as $w(p)=|E(V(p,\olivetreedecomposition),V\backslash V(p,\olivetreedecomposition))|$. The width $w(\olivetreedecomposition)$
of $\olivetreedecomposition$ is defined as the maximum width of a node in $N$. More precisely, 
$w(\olivetreedecomposition) = \max \{w(p)\;|\; p\in N\}$. 
We observe that an olive-tree decomposition of a digraph $G=(V,E)$ has width at most $|E|$. 
In this work we will not be interested in olive-tree decompositions of minimum width. 
Rather, we will be concerned with decompositions having small {\em tree-zig-zag number}, a digraph width measure
that will be defined below.

Let $\olivetreedecomposition= (N,F,\mathfrak{m})$ be
an olive-tree decomposition of a digraph $G=(V,E)$, $H=(V',E')$ be a subgraph of $G$,
and $\mathfrak{m}|_{V'}:V'\rightarrow N$ be the restriction of $\mathfrak{m}$ to $V'$.
We say that the 
triple $\olivetreedecomposition' = (N,F,\mathfrak{m}|_{V'})$ is the olive-tree decomposition of $H$ induced by $\olivetreedecomposition$.
A simple path in a digraph $G$ is an alternated sequence $\path=v_1e_1v_2e_2....v_{n-1}e_{n-1}v_n$ of vertices 
and edges of $G$ such that for each $i\in \{1,...,n-1\}$, the edge $e_i$ has $v_{i}$ as source and $v_{i+1}$ as target, and such that $v_i\neq v_j$ for 
each $i,j$ with $i\neq j$. 
We view $\path$ as a subgraph of $G$ by setting $\path=(V_{\path},E_{\path})$ where 
$V_{\path}=\{v_1,v_2,...,v_n\}$ and $E_{\path}=\{e_1,e_2,...,e_{n-1}\}$. We let 
$$w(\olivetreedecomposition,\path) = \max_{u\in N} |\; E_{\path} \cap E(V(u,\olivetreedecomposition), V\backslash V(u,\olivetreedecomposition))\;|$$
be the width of the path $\path$ along the olive-tree decomposition $\olivetreedecomposition$. Intuitively, 
$w(\olivetreedecomposition,\path)$  quantifies the amount of times the path $\path$ enters or leaves the set $V(u,\olivetreedecomposition)$ for 
each $u\in N$.

\begin{definition}[Tree-Zig-Zag Number]
Let $G=(V,E)$ be a digraph and $\olivetreedecomposition=(N,F,\mathfrak{m})$ be an olive-tree decomposition of $G$. The 
tree-zig-zag number of $\olivetreedecomposition$ is defined as 
$$
\treezigzagnumber(\olivetreedecomposition) =  \max\{w(\olivetreedecomposition, \path) \; | \;  \path \mbox{ is a simple path in $G$} \}. 
$$
The tree-zig-zag number of $G$ is defined as the minimum tree-zig-zag number of an olive-tree decomposition of $G$:
$$\treezigzagnumber(G) = \min\{ \treezigzagnumber(\olivetreedecomposition) \; |\; \olivetreedecomposition \mbox{ is an olive-tree decomposition of $G$}\}.$$
\end{definition}

In Equation \ref{equation:ComparisonWidthMeasures} below we compare the tree-zig-zag number of a digraph with several 
other directed width measures. In \cite{deOliveiraOliveira2013IPEC} we defined the zig-zag number $\zigzagnumber(G)$ 
of a digraph $G$ as a measure that quantifies the amount of times a directed path is allowed enter 
or leave any initial segment of a total ordering of the vertices of $G$. The tree-zig-zag number 
$\treezigzagnumber(G)$ may be regarded as an analog of $\zigzagnumber(G)$ which quantifies the amount 
of times a directed path is allowed to enter or leave any sub-tree of an olive-tree decomposition of $G$.
If $G$ is a digraph, we write $\mathit{dtw}(G)$ for its directed treewidth \cite{JohnsonRobertsonSeymourThomas2001}, 
$\mathit{Dw}(G)$ for 
its {\em $D$-width} \cite{Gruber2012}, $\mathit{dagw}(G)$ for its DAG-width \cite{BerwangerDawarHunterKreutzerObdrzalek2012}, 
$\mathit{dpw}(G)$ for its directed path-width \cite{Barat2006},  $\mathit{Kelw}(G)$ for its Kelly-width 
\cite{GanianHlinenyKneisLangerObdrzRossmanith2014}, 
$\mathit{ddp}(G)$ for its DAG-depth \cite{GanianHlinenyKneisLangerObdrzRossmanith2014}, 
$\mathit{Kw}(G)$ for its K-width \cite{GanianHlinenyKneisLangerObdrzRossmanith2014}, 
$s(G)$ for its  weak separator number \cite{Gruber2012} and $\mathit{cr}(G)$ for its cycle rank \cite{Gruber2012}. 
A dashed arrow $A\dashrightarrow B$ from measure $A$ to measure $B$ indicates that $A$ is 
at least as expressive as $B$. More precisely, there exist constants $\alpha_1,\alpha_2\in \N$ such that 
for every digraph $G$, $A(G)\leq \alpha_1\cdot B(G) + \alpha_2$. A full arrow $A\rightarrow B$ indicates that 
the measure $A$ is strictly more expressive than measure $B$. More precisely, $A$ is at least as expressive as 
$B$, and there exists an infinite class of digraphs in which $A$ is bounded by a constant, but 
$B$ is unbounded. 

\begin{equation}
\label{equation:ComparisonWidthMeasures}
\begin{diagram}
\node[3]{\mathit{zn}(G)} 
	\arrow{sse,t,1}{\mbox{\tiny{\cite{deOliveiraOliveira2013IPEC}}}} 
\\
\node[3]{\mathit{Kelw}(G)} 
	\arrow{se,t,1}{\mbox{\tiny{\cite{GanianHlinenyKneisLangerObdrzRossmanith2014}}}} 
\node[2]{\mathit{Kw}(G)} 
\\
\node{\mathbold{\treezigzagnumber(G)}} 
	\arrow{e,t,..}{\mathit{Th.}\;\ref{theorem:ComparisonWithOtherMeasures}} 
	\arrow[2]{ne,..}
\node{\mathit{dtw}(G)} 
	\arrow{ne,t}{\mbox{\tiny{\cite{HunterKreutzer2008}}}} 
	\arrow{e,t}{\mbox{\tiny{\cite{BerwangerDawarHunterKreutzerObdrzalek2012}}}}
	\arrow{se,t}{\mbox{\tiny{\cite{AmiriKaiserKreutzerRabinovichSiebertz2015}}}} 
\node{\mathit{dagw}(G)} 
	\arrow{e,t}{\mbox{\tiny{\cite{BerwangerDawarHunterKreutzerObdrzalek2012}}}}
\node{\mathit{dpw}(G)}
	\arrow{e,t}{\mbox{\tiny{\cite{Gruber2012}}}} 
\node{\mathit{cr}(G)} 
	\arrow{n,r}{\mbox{\tiny{\cite{GanianHlinenyKneisMeisterObdrzalekRossmanithSikdar2010}}}} 
	\arrow{s,r}{\mbox{\tiny{\cite{GanianHlinenyKneisMeisterObdrzalekRossmanithSikdar2010}}}} 
\\
\node{\frac{\mathit{cr}(G)}{\log |G|}}
	\arrow{e,t,..}{\mbox{\tiny{\cite{Gruber2012}}}} 
\node{s(G)} 
	\arrow{e,t,..}{\mbox{\tiny{\cite{Gruber2012}}}}  
\node{\mathit{Dw}(G)} \arrow{ne,t}{\mbox{\tiny{\cite{Safari2005,deOliveiraOliveira2013IPEC}}}} 
\node[2]{\mathit{ddp}(G)}
\end{diagram}
\end{equation}

%
%

The numbers above each arrow $A\rightarrow B$ ($A\dashrightarrow B$) in Equation \ref{equation:ComparisonWidthMeasures} 
refer to the works in which the corresponding relation between the measures $A$ and $B$ was established. 
All relations listed above can be inferred from the literature, except for the relation
$\mathbold{\treezigzagnumber}(G)\dashrightarrow \mathit{zn}(G)$, which is immediate, and the relation 
$\mathbold{\treezigzagnumber(G)} \dashrightarrow dtw(G)$, which will be formally stated in 
Theorem \ref{theorem:ComparisonWithOtherMeasures} below and proved in Section \ref{section:ComparisonDirectedTreewidth}.
The fact that DAG-width, Kelly-width and D-width are strictly more expressive than directed pathwidth follows 
from the fact that the width of the complete undirected\footnote{In this setting each undirected edge is represented 
by two directed edges in opposite directions.} binary tree on $n$ leaves is bounded with respect to these three
measures, but unbounded ($\Omega(\log n)$) with respect to directed pathwidth \cite{deOliveiraOliveira2013IPEC}.

It is worth noting that the precise statement of our main theorem (Theorem \ref{theorem:MainTheoremDirectedTreewidth})
holds if the parameter $w$ corresponds to the width of the digraph $G$ with respect to any measure reachable from 
$\mathit{dtw}(G)$ in Equation \ref{equation:ComparisonWidthMeasures}. More Precisely, 
Theorem \ref{theorem:MainTheoremDirectedTreewidth} also holds when $w$ is the Kelly width, DAG-width, D-width, directed pathwidth,
cycle rank, K-width or DAG-depth of $G$. 
{Theorem \ref{theorem:MainTheoremDirectedTreewidth}} can also be applied if the parameter $w$ is the 
tree-zig-zag number of $G$. However, in this particular case, an explicit olive-tree decomposition of $G$ of 
tree-zig-zag number $O(w)$ must be given in the input. For directed tree-width and less expressive measures, 
such an olive-tree decomposition of width $O(w)$ can be automatically constructed in time $|G|^{O(w)}$. 
This construction will be carried in Section \ref{section:ComparisonDirectedTreewidth} together with the proof of 
Theorem \ref{theorem:ComparisonWithOtherMeasures}. 
%

\begin{theorem} 
\label{theorem:ComparisonWithOtherMeasures}
Let $G$ be a digraph, $\treezigzagnumber(G)$ be its tree-zig-zag number and $dtw(G)$ be its directed treewidth. 
Then $\treezigzagnumber(G)\leq 9\cdot dtw(G) + 18$. 
\end{theorem}

\section{Tree-Zig-Zag Number vs Directed Treewidth}
\label{section:ComparisonDirectedTreewidth}

In this section we will prove Theorem \ref{theorem:ComparisonWithOtherMeasures}. 
First we will state a couple of propositions concerning $Z$-normal sets.

\begin{proposition}
\label{proposition:NormalComplement}
Let $G=(V,E)$ be a digraph and $K,Z\subseteq V$ be such that $K$ is $Z$-normal. Then for each subset $X\subseteq K$, 
$K\backslash X$ is $Z\cup X$-normal. 
\end{proposition}
\begin{proof}
The proof is by contradiction. 
Assume that there is an $X\subseteq K$ such that $K\backslash X$ is not $Z\cup X$-normal. Then there is a walk in $G\backslash (Z\cup X)$ that 
starts and ends in $K\backslash X$, but that uses a vertex from $V\backslash ((Z\cup X)\cup (K\backslash X))= V\backslash (Z\cup K)$. This contradicts the 
assumption that $K$ is $Z$-normal.  
\end{proof}

If $G=(V,E)$ is a digraph, $Z$ is a subset of $V$ and $\path = v_1e_1v_2....v_{n-1}e_{n-1}v_n = (V_{\path},E_{\path})$ is a path on $G$ then we say that 
$\path$ is internally disjoint from $Z$ if $Z\cap V_{\path} \subseteq \{v_1,v_n\}$. In other words, $\path$ is internally disjoint
from $Z$ if none of its internal vertices belongs to $Z$. The next proposition says that if $K$ is a $Z$-normal subset of $V$ and 
$\path$ is a path that is internally disjoint from $Z$, then $\path$ can enter or leave $K$ at most $2$ times.

\begin{proposition}
\label{proposition:InternallyDisjointPath}
Let $G=(V,E)$ be a digraph and $K,Z \subseteq V$ be subsets of $V$ such that $K$ is $Z$-normal. Let 
$\path = (V_{\path},E_{\path})$ be a path in $G$ that is internally disjoint from $Z$. Then $$|E_{\path} \cap E(K,V\backslash K)| \leq 2.$$
\end{proposition}
\begin{proof}
The proof is by contradiction. Assume that $|E_{\path} \cap E(K,V\backslash K)| \geq 3$. Let $e_1$, $e_2$ and $e_3$ be the first three 
edges of $\path$ that have one endpoint in $K$ and other endpoint in $V\backslash K$. Then 
$\path = \path_0 e_1 \path_1  e_2 \path_2 e_3  \path_3$ where for each $i\in \{1,2,3\}$, the source of $e_i$ is 
the last vertex of $\path_{i-1}$ and the target of $e_i$ is the first vertex of $\path_i$. Since the path $\path$ is internally disjoint 
from $Z$, we have that either $\path_1$ is entirely contained in $K$ and $\path_2$ is entirely contained in $V\backslash (K\cup Z)$, or 
$\path_1$ is entirely contained in $V\backslash (K\cup Z)$ and $\path_2$ is entirely contained in $K$. Therefore
either $e_1\path_1e_2$ or $e_2\path_2e_3$ is a path that starts and finishes at $K$ and uses a vertex of $V\backslash (K\cup Z)$. 
This contradicts the assumption that $K$ is $Z$-normal. 
\end{proof}

The next proposition says that if $\path$ is a path of $G$, then the number of edges of $\path$ crossing 
a $Z$-normal set is upper bounded by  $2\cdot |Z|+2$. 

\begin{proposition}
\label{proposition:NumberOfEdges}
Let $G=(V,E)$ be a digraph and $K,Z \subseteq V$ be subsets of $V$ such that $K$ is $Z$-normal. Then for each
path $\path = (V_{\path},E_{\path})$ in $G$, $$|E_{\path} \cap E(K,V\backslash K)| \leq 2\cdot |Z|+2.$$ 
\end{proposition}
\begin{proof}
Let $\path=(V_{\path},E_{\path})$ be a path in $G$ and assume that $V_{\path}\cap Z  = \{v_1,...,v_k\}$. We may assume without 
loss of generality that for each $i\in \{1,...,k-1\}$, $v_i$ occurs before $v_{i+1}$ in $\path$. In other words we may assume that 
$\path = \path_0 \cup \path_1 \cup ... \cup \path_k$ where $\path_0,\path_1,...,\path_k$ are internally disjoint paths in 
which for each $i\in \{1,...,k\}$, $v_i$ is the last vertex of $\path_{i-1}$ and the first vertex of $\path_i$.
We note that for each $i\in \{1,...,k\}$ the path $\path_i=(V_{\path_i},E_{\path_i})$ is internally disjoint from $Z$. Therefore, from Proposition
\ref{proposition:InternallyDisjointPath} we have that $|E_{\path_i} \cap E(K,V\backslash K)| \leq 2$. 
This implies that $|E_{\path} \cap E(K,V\backslash K)| \leq \sum_{i=0}^k |E_{\path_i}\cap E(K,V\backslash K)| \leq 2k+2 \leq 2|Z|+2$.  
\end{proof}

The next proposition says that if $G$ is a digraph and $K_1$ and $K_2$ are disjoint subsets of vertices of $G$ such that 
no edge has source in $K_2$ and target in $K_1$, then any path crossing $K_1\cup K_2$ at most $2$ times, crosses 
$K_2$ at most $3$ times. 

\begin{proposition}
\label{proposition:UnionSetsPath}
Let $G=(V,E)$ be a digraph and $K_1,K_2$ be subsets of vertices of $G$ such that 
$K_1\cap K_2 = \emptyset$ and such that there is no edge with source in $K_2$ and 
target in $K_1$. Let $\path=(V_{\path},E_{\path})$ be a path in $G$ such that 
$|E_{\path}\cap E(K_1\cup K_2, V\backslash (K_1\cup K_2))| \leq 2$. Then $|E_{\path}\cap E(K_2,V\backslash K_2)| \leq 3$.
\end{proposition}
\begin{proof}
Let $\path=(V_{\path},E_{\path})$ be a path in $G$.
If $|E_{\path}\cap E(K_1\cup K_2, V\backslash (K_1\cup K_2))|=0$  then $\path$ is either entirely contained in $K_1\cup K_2$
or entirely contained in $V\backslash (K_1\cup K_2)$ and the proposition holds trivially. Now let $|E_{\path}\cap E(K_1\cup K_2, V\backslash (K_1\cup K_2))|=1$
and let $e_1$ be the unique edge with one endpoint in $K_1\cup K_2$ and another endpoint in $V\backslash (K_1\cup K_2)$. 
Then $\path = \path_0e_1\path_1$ where the source of $e_1$ is the last vertex of $\path_0$ and the target of $e_1$ is the first vertex 
of $\path_1$. Note that for each $i\in \{0,1\}$, either $\path_i$ is entirely contained in $K_1\cup K_2$ or entirely contained 
in $V\backslash (K_1 \cup K_2)$. Since there is no edge with source in $K_2$ and target in $K_1$, we have that $\path_0$ and 
$\path_1$ can each cross $K_2$ at most one time. In other words, $|E_{\path_i}\cap E(K_2,V\backslash K_2)| \leq 1$. 
Therefore $\path_0$, $e_1$ and $\path_1$ together cross $K_2$ at most three times and the proposition holds in this case. 
Finally, let $|E_{\path}\cap E(K_1\cup K_2, V\backslash (K_1\cup K_2))|=2$. 
Let $e_1$ and $e_2$ be the only edges of $\path$ with one endpoint in $K_1\cup K_2$ and the other endpoint in $V\backslash (K_1\cup K_2)$, 
and assume that $e_1$ is visited before $e_2$. Then there are paths $\path_1,\path_2,\path_3$ such that $\path = \path_0e_1\path_1e_2\path_2$
and for $i\in \{1,2\}$, the source of $e_i$ is the last vertex of $\path_{i-1}$ and the target of $e_i$ is the source of $\path_i$. 
Note that for $i\in \{0,1,2\}$, $\path_i$ is either entirely contained in $K_1\cup K_2$ or entirely contained in $V\backslash (K_1\cup K_2)$. Note also 
that each $\path_i$ crosses $K_2$ at most one time, since there is no edge with source in $K_2$ and target in $K_1$. 
This already implies that $\path$ can cross $K_2$ at most $5$ times. We claim that with further analysis it can be shown that the number of 
crossings is at most $3$, which is optimal. The analysis is as follows. If the source of $e_1$ belongs to $V\backslash (K_1\cup K_2)$, then the 
target of $e_2$ is also in $V\backslash (K_1\cup K_2)$. In this case 
both $\path_0$ and $\path_2$ are entirely contained in $V\backslash (K_1\cup K_2)$, and therefore only $e_1,e_2$ and $\path_1$ have 
the possibility of crossing $K_2$. If the source of $e_1$ belongs 
to $K_1\cup K_2$, then the target of $e_2$ also belongs to $K_1\cup K_2$. This implies that $\path_1$ is entirely contained in 
$V\backslash (K_1\cup K_2)$. In this situation there are two sub-cases to be analysed. 
If the target of $e_2$ belongs to $K_2$ then $\path_2$ is entirely contained in $K_2$ and only $e_1,e_2$ and $\path_0$ have the 
possibility of crossing $K_2$. On the other hand, if the target of $e_2$ is in $K_1$ then $e_2$ does not cross $K_2$ and thus only 
$e_1$, $\path_0$ and $\path_2$ have the possibility to cross $K_2$. 
\end{proof}

Using Proposition \ref{proposition:UnionSetsPath} we can show that if $K_1$ and $K_2$ are disjoint subsets of vertices 
of a digraph $G$ such that  $K_1\cup K_2$ is a $Z$-normal and such that there is no edge with source in $K_2$ and target 
in $K_1$, then each path in $G$ crosses $K_2$ at most $3|Z|+3$ times.

\begin{proposition}
\label{proposition:UnionSets}
Let $G=(V,E)$ be a digraph and $K_1,K_2,Z \subseteq V$ be subsets of vertices of $G$ such that 
$K_1\cup K_2$ is $Z$-normal, $K_1\cap K_2 =\emptyset$, and such that there is no edge with source in $K_2$ and target in $K_1$. 
Then for each path $\path = (V_{\path},E_{\path})$ in $G$ we have that $$|E_{\path}\cap E(K_2,V\backslash K_2)| \leq 3\cdot |Z| + 3.$$  
\end{proposition}
\begin{proof}
Analogously to the proof of Proposition \ref{proposition:NumberOfEdges} we let $V_{\path}\cap Z = \{v_1,...,v_k\}$, and assume that
 $\path = \path_0\cup \path_1 \cup ... \cup \path_k$ where $\path_1,\path_2,...,\path_k$ are internally disjoint paths such that  
for each $i\in \{1,...,k\}$, $v_i$ is the last vertex of $\path_{i-1}$ and the first vertex of $\path_i$. We note that 
for each $i\in \{1,...,k\}$ the path $\path_i$ is internally disjoint from $Z$. Therefore, since $K_1\cup K_2$ is $Z$-normal,
by Proposition \ref{proposition:InternallyDisjointPath} we have that 
$|E_{\path_i}\cap E(K_1\cup K_2, V\backslash(K_1\cup K_2))| \leq 2$. 
Now, since there is no edge with source in $K_2$ and target in $K_1$, we can apply Proposition \ref{proposition:UnionSetsPath} 
to infer that  $|E_{\path_i}\cap E(K_2,V\backslash K_2)|\leq 3$. 
This implies that $|E_{\path} \cap E(K_2,V\backslash K_2)| \leq \sum_{i=0}^k |E_{\path_i}\cap E(K_2,V\backslash K_2)| \leq 3k+3 \leq 3|Z|+3$.  
\end{proof}

The main technical lemma of this section states that each good arboreal decomposition $\arborealdecomposition$ of width 
$w$ can be transformed into an olive-tree decomposition $\olivetreedecomposition$ of tree-zig-zag number at most $3w+6$.

\begin{lemma}
\label{lemma:FromArborealPredecompositionToOliveTreeDecomposition}
Let $G=(V,E)$ be a digraph and $\arborealdecomposition = (N,F,W,Z)$ be a good arboreal decomposition of $G$ of width 
$w$. One can construct in time $O(w\cdot |N|)$ an olive-tree decomposition $\olivetreedecomposition = (N',F',\mathfrak{m})$ of 
$G$ of tree-zig-zag number at most $3w+6$.
\end{lemma}
\begin{proof}
We start by defining the sets of nodes and arcs of the olive-tree decomposition $\olivetreedecomposition$. Intuitively, $\olivetreedecomposition$
is obtained by replacing each node $p$ of $\arborealdecomposition$, labeled with a bag $W(p)$ and having $r$ children, with a line 
$L_p \equiv a_p^0a_p^1...a_p^{|W(p)|}b_p^1b_p^2...b_p^{r}$ as depicted in Figure \ref{figure:conversionDirectedTreewidthCarving}. 

\begin{figure}[!hf] 
\centering 
\includegraphics[scale=0.25]{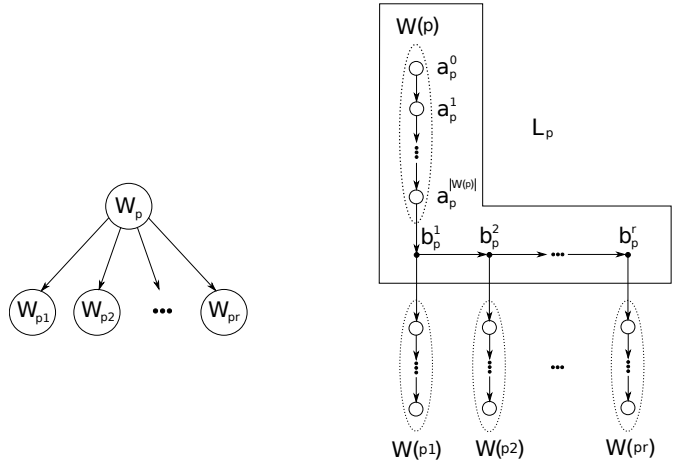}
\caption{From a good arboreal decomposition of width $w$ to an olive-tree decomposition of tree-zig-zag number at most $3w+6$. }
\label{figure:conversionDirectedTreewidthCarving}
\end{figure} 

Each vertex in $W(p)$ is mapped by $\mathfrak{m}$ to a node in $\{a_p^1,...,a_p^{|W(p)|}\}$ in such a way that no two distinct vertices 
in $W(p)$ are mapped to the same node. No vertex of $G$ is mapped to the node 
$a_p^0$ nor to the nodes $b_p^1,b_p^2,...,b_p^r$. These nodes are used to connect the line $L_p$ corresponding to the node $p$ of $\arborealdecomposition$ 
to lines corresponding to other positions. In particular for each $j\in \{1,...,r\}$, $b_p^j$ is connected to $a^0_{pj}$. In other words, $b_p^j$ is connected
to the first vertex of the line $L_{pj}$. We will show below that the olive-tree decomposition defined in this way has width 
at most $3w+6$. First however we formally define the sets $N'$ and $F'$ and the mapping function $\mathfrak{m}$.

\begin{equation}
\label{equation:setOfNodes}
N' =  \{a_p^{i}\;|\;p\in N, 0\leq i \leq |W(p)| \} \cup \{b_p^j \;|\; (p,pj)\in F \}
\end{equation}

\begin{equation}
\label{equation:setofArcs}
\begin{array}{l}
F' = \{(a_{p}^i,a_p^{i+1})\;|\;p\in N,\;0\leq i \leq |W(p)|-1 \}\; \cup\;  \{ (a_p^{|W(p)|}, b_{p}^1)\;|\; p\in N \} \; \cup   \\
\\ 
\hspace{1cm}  \{(b_p^i,b_p^{i+1})\;|\; p\in N, 1\leq i\leq r-1\}  \;\cup\; \{(b_p^{j}, a_{pj}^0)\;|\; (p,pj) \in F\} 
\end{array}
\end{equation}

Finally, the labeling function $\mathfrak{m}:V\rightarrow N'$ is chosen arbitrarily as long as it satisfies the following 
condition for each node $p\in N$.

\begin{equation}
\label{equation:MapBags}
\mathfrak{m}(W(p))= \{a_p^1,...,a_p^{|W(p)|}\}.
\end{equation}

In other words we choose $\mathfrak{m}$ in such a way that for each position $p\in N$, the vertices in $W(p)$ are 
bijectively mapped to the nodes in $\{a_p^1, ...,a_p^{|W(p)|}\}$.  
We argue that $\olivetreedecomposition$ is indeed an olive-tree decomposition of 
tree-zig-zag number at most $3w+6$.
Recall that if $G=(V,E)$ is a digraph and $\arborealdecomposition$ is an arboreal decomposition of $G$, then 
for each node $p$ of $\arborealdecomposition$, $V(p,\arborealdecomposition)$ denotes the set of 
vertices of $G$ that belong to some bag associated with a descendant of $p$ (including $p$ itself). 
Analogously, if $\olivetreedecomposition$ is an olive-tree decomposition of $G$ then for each node $u$ of $\olivetreedecomposition$, 
 $V(u,\olivetreedecomposition)$ denotes the set of vertices of $G$ that are mapped to some descendant of $u$. 
To show that $\olivetreedecomposition$ has tree-zig-zag number at most $3w+6$ we 
need to show that for each node $u\in N'$, and each path $\path$ in $G$, 
$$| E_{\path}\cap E(V(u,\olivetreedecomposition),V\backslash V(u,\olivetreedecomposition))|\leq 3w+6.$$ 
There are two cases to be considered, depending on whether $u=a_p^i$ or whether $u=b_p^j$. We analyse each of these cases below.  
\begin{enumerate}
	\item ($u=a_p^i$) We start by noting that for each node $p\in N$, $V(p,\arborealdecomposition) = V(a_p^0,\olivetreedecomposition)$.
		If $p$ is the root of $\arborealdecomposition$ then $V(p,\arborealdecomposition) = V$ and thus both $V(p,\arborealdecomposition)$ 
		and $V(a_{p}^0,\olivetreedecomposition)$ are $\emptyset$-normal.
		If $p$ is not the root of $\arborealdecomposition$ then $p$ has a parent $p'$. In this case by definition of arboreal decomposition 
		we have that $V(p,\arborealdecomposition)$ is $Z(p',p)$-normal. Thus $V(a_p^0,\olivetreedecomposition)$ is also $Z(p',p)$-normal.
		We let $X_p$ be equal to $\emptyset$ if $p$ is the root of $\arborealdecomposition$, and equal to $Z(p',p)$ if $p'$ is the parent of $p$. 
		Thus we can simply say that $V(p,\arborealdecomposition)$ is $X_p$-normal. 

		Now let $j\in \{1,...,|W(p)|\}$. Then 
		$V(a_p^j,\olivetreedecomposition) = V(a_p^0,\olivetreedecomposition) \backslash \mathfrak{m}^{-1}(\{a_p^1,...,a_p^{j-1}\})$. In other words, 
		$V(a_p^j,\olivetreedecomposition)$ is equal to $V(a_p^0,\olivetreedecomposition)$ minus the vertices of $G$ that are mapped by 
		$\mathfrak{m}$ to some node in $\{a_p^1,...,a_p^{j-1}\}$. This implies, by 
		Proposition \ref{proposition:NormalComplement} that the set $V(a_p^j,\olivetreedecomposition)$ is 
		$X_p \cup \mathfrak{m}^{-1}(\{a_p^1,...,a_p^{j-1}\})$-normal. 

		Since by the construction of $\olivetreedecomposition$, $\mathfrak{m}^{-1}(\{a_p^1,...,a_p^{j-1}\}) \subseteq W(p)$ 
		(Equation \ref{equation:MapBags}), and since $\arborealdecomposition$ has width $w$, we have that 
		$|X_p \cup \mathfrak{m}^{-1}(\{a_p^1,...,a_p^{j-1}\})| \leq w+1$. 
		Therefore, by Proposition \ref{proposition:NumberOfEdges}, for each $j\in \{0,...,|W(p)|\}$, 
		$$|E_{\path}\cap E(V(a_p^j,\olivetreedecomposition),V\backslash V(a_p^j,\olivetreedecomposition))| \leq 2(w+1)+2 \leq 3w +6.$$ 

	\item ($u=b_p^j$) Let $u=b_p^j$ for some $p\in N$ and $j\in \{1,...,r\}$. Since by the construction of $\olivetreedecomposition$,
		no vertex of $G$ is mapped to $b_p^j$, we 
		have that $V(b_{p}^j,\olivetreedecomposition) = \bigcup_{k=j}^r V(a_{pk}^0,\olivetreedecomposition)$. Additionally, since 
		$V(a_p^0,\olivetreedecomposition) = V(p,\arborealdecomposition)$ for each $p\in N$, we have that 
		$V(b_{p}^j,\olivetreedecomposition) = \bigcup_{k=j}^r V(pk,\arborealdecomposition)$.
		Note that $V(b_p^j,\olivetreedecomposition)\subseteq V(b_p^1,\olivetreedecomposition)$ for each $j\in \{1,...,r\}$. 
		Since $\arborealdecomposition$ is a good arboreal decomposition, we have that for $i,j\in \{1,...,r\}$ with $i<j$ 
		there is no edge with source in $V(pj,\arborealdecomposition)$ and target in $V(pi,\arborealdecomposition)$.
		This implies that for each $j\in \{1,...,r\}$, there is no edge with source in $V(b_{p}^j,\olivetreedecomposition)$ and 
		target in $V(b_p^1,\olivetreedecomposition)\backslash V(b_{p}^j,\olivetreedecomposition)$. Now let $X_p$ be equal 
		to $\emptyset$ if $p$ is the root of $\arborealdecomposition$, and equal to $Z(p',p)$ if $p'$ is the parent of $p$. 
		We note that $V(a_p^0,\olivetreedecomposition)=V(p,\arborealdecomposition)$ is $X_p$-normal, and that 
		$V(b_p^1,\olivetreedecomposition)= V(a_p^0,\olivetreedecomposition) \backslash W(p)$. Therefore, by Proposition \ref{proposition:NormalComplement},
		$V(b_p^1,\olivetreedecomposition)$ is $X_p\cup W(p)$-normal. 
		Now, we can apply Proposition \ref{proposition:UnionSets} with 
		$K_1 = V(b_p^1,\olivetreedecomposition)\backslash V(b_p^j,\olivetreedecomposition)$, $K_2= V(b_p^j,\olivetreedecomposition)$, 
		and $Z=X_p \cup W(p)$ to infer that  
		$$|E_{\path}\cap E(V(b_p^j,\olivetreedecomposition), V\backslash V(b_p^j,\olivetreedecomposition))|\leq 3|Z|+3 \leq  3(w+1) + 3 = 3w+6.$$
		The inequality $3|Z|+3\leq 3(w+1)+3$ follows from the fact that $|Z|=|X_p\cup W(p)|\leq w+1$,  since $\arborealdecomposition$ has width $w$.
		 
\end{enumerate}
\end{proof}

Finally we are in a position to prove Theorem \ref{theorem:ComparisonWithOtherMeasures}.  

\paragraph{\bf Proof of Theorem \ref{theorem:ComparisonWithOtherMeasures}}

By Lemma \ref{theorem:ConstructionGoodArborealDecomposition}, given a digraph 
$G$ of directed treewidth $w$ one can construct in time $|G|^{O(w)}$ a good arboreal 
decomposition $\arborealdecomposition$ of $G$ of width at most $3w+4$. By Lemma \ref{lemma:FromArborealPredecompositionToOliveTreeDecomposition} one can transform
$\arborealdecomposition$ into an olive-tree decomposition $\olivetreedecomposition$ of $G$ of tree-zig-zag number at most $3(3w+4)+6 = 9w+18$. 
$\square$

\section{Tree Automata}
\label{section:TreeAutomata}

In this section we recall some of the main concepts of tree-automata theory. For an extensive
treatment of the subject we refer the reader to the standard reference \cite{Tata2007}. As two 
non-standard applications, we consider the problem of counting the number of terms of depth 
$d$ accepted by a deterministic tree-automaton, and the problem of generating terms having a prescribed weight. 

A ranked alphabet is a finite set $\newslicealphabet = \newslicealphabet_0 \cup \newslicealphabet_1 \cup ... \cup \newslicealphabet_r$  of 
function symbols where the elements of $\newslicealphabet_i$ are function symbols of arity $i$. 
Intuitively, the arity of a function symbol specifies its number of inputs.
Constants are regarded as function symbols of arity $0$. If $f$ is a function symbol in $\newslicealphabet$
then we let $\arity(f)$ denote the arity of $f$. In other words $\arity(f)=i$ if and only if $f\in \newslicealphabet_i$.  
The set $\terms(\newslicealphabet)$ of all terms\footnote{In this work we will {\em not} be interested in terms containing variables. 
In other words, all terms considered here are ground terms.} over $\newslicealphabet$ is inductively defined as follows:

\begin{itemize}
	\item if $f\in \newslicealphabet_0$ then $f$ is a term in $\terms(\newslicealphabet)$,
	\item if $f\in \newslicealphabet_{\arity(f)}$ and $t_1,...,t_{\arity(f)}$  are terms in $\terms(\newslicealphabet)$ 
		then $f(t_1,t_2,...,t_{\arity(f)})$ is a term in $\terms(\newslicealphabet)$. 
\end{itemize}

Let $t=f(t_1,...,t_{\arity(f)})$ be a term over the ranked alphabet $\newslicealphabet$. Then we define $\leadingsymbol(t)=f$ as the leading 
symbol of $t$. We denote by $\positions(t)$ the set of positions of $t$, which is a prefix closed subset of $\{1,...,r\}^*$ used to index the subterms of $t$. 
More precisely, if  $t=f(t_1,...,t_{\arity(f)})$ then

$$\positions(t)=\{\lambda\} \cup \bigcup_{ j \in \{1,...,\arity(f)\}} \{ j p\;|\; p\in \positions(t_{ j })\}.$$

We note that if $t$ is a constant, i.e., a function symbol of arity $0$, then $\positions(t)=\{\lambda\}$.
If $t\in \terms(\newslicealphabet)$ then we let $|t|=|\positions(t)|$.   
The subterm $t|_p$ of $t$ at position $p$ is 
inductively defined as follows: $t|_{\emptystring} = t$; if $t=f(t_1,t_2,...,t_{\arity(f)})$, then for each $ j \in [\arity(f)]$ and each position 
$j p\in \positions(t)$, $t|_{ j  p} = t_{ j }|_{p}$.
If $t$ is a term and $p\in \positions(t)$ then $t[p] = \leadingsymbol(t|_p)$ denotes the leading symbol of the subterm of $t$ at position $p$. 

A tree-language over a ranked alphabet $\newslicealphabet$ is any subset $\lang\subseteq \terms(\newslicealphabet)$. In particular 
the empty set $\emptyset$ is a tree-language. A {\em bottom-up tree-automaton} over $\newslicealphabet$ is a tuple $\treeAutomaton=(Q,\newslicealphabet,Q_{F},\Delta)$
where $Q$ is a set of states, $Q_{F}\subseteq Q$ a set of final states and $\Delta=\Delta_0 \cup \Delta_1 \cup ....\cup \Delta_r$ is a transition relation 
where $\Delta_0\subseteq \Sigma_0\times Q$ and $\Delta_i\subseteq Q^i\times \newslicealphabet_i \times Q$ for each $i\in \{1,...,r\}$.
The size of $\treeAutomaton$, which is defined as $|\treeAutomaton|=|Q|+|\Delta|$, measures the number of states in $Q$ plus the number of transitions in $\Delta$.
The set $\lang(\treeAutomaton,\astate,i)$ of all terms reaching a state $\astate\in Q$ in depth at most $i$ is inductively defined as follows.

\vspace{-2pt}
\begin{equation}
\label{equation:InductiveDefinition}
\lang(\treeAutomaton,\astate,1) = \{a\in \newslicealphabet_0\; |\; (a,\astate)\in \Delta_0\}
\end{equation}
\begin{equation*}
\begin{array}{l}
\lang(\treeAutomaton,\astate,i) =  \lang(\treeAutomaton,\astate,i-1)\; \cup\; \{f(t_1,...,t_{\arity(f)})\;|\;  (\astate_1,...,\astate_{\arity(f)},f,\astate)\in \Delta_{\arity(f)}, \\
\hspace{8.5cm} t_j\in \lang(\treeAutomaton,\astate_j,i-1)\}
\end{array}
\end{equation*}

We denote by $\lang(\treeAutomaton)$ the set of all terms reaching a final state in $Q_{F}$ at any finite depth.

\begin{equation}
\lang(\treeAutomaton) =\bigcup_{{\astate\in Q_{F},i\in \N}} \lang(\treeAutomaton,\astate,i)
\end{equation}

We say that the set $\lang(\treeAutomaton)$ is the language generated by $\treeAutomaton$.
Let $\treeAutomaton=(Q,\newslicealphabet,Q_{F},\Delta)$ be a tree-automaton. We say that $\treeAutomaton$ is
{\em deterministic} if for every function symbol $f\in \newslicealphabet$ and every tuple $(\astate_1,...,\astate_{\arity(f)})\in Q^{\arity(f)}$
there exists at most one $\astate\in Q$ such that $(\astate_1,...,\astate_{\arity(f)},f,\astate)\in \Delta_{\arity(f)}$. We say that 
$\treeAutomaton$ is complete if for every function symbol $f$ and every tuple $(\astate_1,...,\astate_{\arity(f)}) \in Q^{\arity(f)}$
there exists at least one $\astate\in Q$ for which $(\astate_1,...,\astate_{\arity(f)},f,\astate)\in \Delta_{\arity(f)}$. Observe that 
from any tree-automaton $\treeAutomaton$ one can derive a complete tree-automaton $\treeAutomaton'$ generating the same 
language by adding a dead state $\astate_{\mathit{dead}}$, and creating a transition $(\astate_1,...,\astate_{\arity(f)},f,\astate_{\mathit{dead}})$ 
whenever there is no transition in $\treeAutomaton$ whose left side is $(\astate_1,...,\astate_{\arity(f)},f)$. 

If $t$ is a term in $\terms(\newslicealphabet)$, then the depth of $t$ is defined as $\max\{|p|: p\in \positions(t)\}$. In other 
words, the depth of a term $t$ is the size of the longest path from the root of $t$ to one of its leaves. We denote 
by $\mathit{depth}(t)$ the depth of $t$. 
If $\treeAutomaton$  is a tree-automaton and  $t\in \lang(\treeAutomaton)$ is a term of depth $d$, then we say that 
$\treeAutomaton$ accepts $t$ in depth $d$. 
The next lemma says that for any deterministic tree-automaton 
$\treeAutomaton$ and any $d\in \N$, one can count in polynomial time the number of terms accepted by $\treeAutomaton$ 
in depth at most $d$. 

\begin{lemma}
\label{lemma:CountingTrees}
Let $\treeAutomaton$ be a deterministic tree-automaton and let $d\in \N$. One can 
count in time $d^{O(1)}\cdot|\treeAutomaton|^{O(1)}$ the number of terms accepted by $\treeAutomaton$ in 
depth at most $d$. 
\end{lemma}
\begin{proof}
The proof follows by a standard dynamic programming argument. 
First we write a recursive formula that counts the number of terms that reach a given state $\astate$ in depth $i$:

\begin{equation}
\label{equation:CountingAcceptedTermsOne}
|\lang(\treeAutomaton,\astate,1)|= |\{(f,\astate)\;|\; f\in \newslicealphabet_0,\;\astate \in Q\}|
\end{equation}
\begin{equation}
\label{equation:CountingAcceptedTermsTwo}
|\lang(\treeAutomaton,\astate,i)|= \sum_{(\astate_1,...,\astate_{\arity(f)},f,\astate) \in \Delta} \;\; \prod_{j=1}^{\arity(f)} |\lang(\treeAutomaton,\astate_j,i-1)|
\end{equation}

Now the number of terms accepted by $\treeAutomaton$ in depth $d$ is the number of terms that reach a final state at depth 
$d$. 

\begin{equation}
\label{equation:FinalCounting}
|\lang(\treeAutomaton)| = \sum_{\astate\in Q_{F}, i\leq d} |\lang(\treeAutomaton,\astate, i)|.
\end{equation}

Thus to determine $|\lang(\treeAutomaton)|$, one can use Equation \ref{equation:CountingAcceptedTermsOne} to 
compute and store in memory the value $|\lang(\treeAutomaton,\astate,1)|$ for each $\astate\in Q$. Subsequently, using the values 
stored in memory, one can use Equation \ref{equation:CountingAcceptedTermsTwo} to compute and store in memory the values $|\lang(\treeAutomaton,\astate,2)|$ and so on. We repeat this process 
until we have computed all values $|\lang(\treeAutomaton,\astate,d)|$. At this point, we apply Equation \ref{equation:FinalCounting}
to determine the number of terms that reach a final term in depth at most $d$. Since in this work, $\arity(f)$ is
bounded by a constant (indeed in our applications $\arity(f)\leq 2$), this counting process can clearly be done in time 
$d^{O(1)}\cdot |\treeAutomaton|^{O(1)}$.

\end{proof}

\subsection{Properties of Tree-Automata}
\label{subsection:PropertiesOfTreeAutomata}

If $\lang$  is a tree language over 
a ranked alphabet $\newslicealphabet$ then the complement of $\lang$ is defined as $\overline{\lang} = \terms(\newslicealphabet)\backslash \lang$. 
A projection between ranked alphabets $\newslicealphabet$ and $\newslicealphabet'$ is any arity preserving total mapping 
$\projection:\newslicealphabet\rightarrow \newslicealphabet'$. By arity preserving we mean that if $f$ is a function symbol of arity $r$ in 
$\newslicealphabet$, then $\projection(f)$ is a function symbol of arity $r$ in $\newslicealphabet'$.
Recall that if $t$ is a term and $p\in \positions(t)$ then $t[p] = \leadingsymbol(t|_p)$ denotes the leading symbol of the subterm of $t$ rooted at position $p$. 
A projection $\projection:\newslicealphabet\rightarrow \newslicealphabet'$ can be homomorphically extended to a mapping $\projection:\terms(\newslicealphabet)\rightarrow \terms(\newslicealphabet')$ between terms
by setting $\projection(t)[p]=\projection(t[p])$ for each position $p\in \positions(t)$. 
Additionally, such mapping $\projection$ can 
be further extended to tree languages $\lang\subseteq \terms(\newslicealphabet)$ by setting $\projection(\lang) = \{\projection(t)\;|\;t\in \terms(\newslicealphabet)\}$.
Finally, given a projection $\projection:\newslicealphabet\rightarrow \newslicealphabet'$ and a tree language $\lang$ over $\newslicealphabet'$, 
the inverse homomorphic image of $\lang$ under $\projection$ is defined as $\projection^{-1}(\lang)=\{t\in \terms(\newslicealphabet)\;|\;\projection(t)\in \lang \}$,
i.e., the set of all terms over $\terms(\newslicealphabet)$ which are mapped to some term in $\lang$. 
In Lemma \ref{lemma:PropertiesOfTreeAutomata} below we list several well known closure properties 
of tree languages recognizable by tree-automata (see for instance \cite{Tata2007}).

\begin{lemma}[Properties of Tree Automata]
\label{lemma:PropertiesOfTreeAutomata}
Let $\treeAutomaton$ be an arbitrary tree-automaton over a ranked alphabet $\newslicealphabet$ and let $\treeAutomaton_1$ and $\treeAutomaton_2$ be 
deterministic complete tree-automata over $\newslicealphabet$. 
\begin{enumerate}[(i)]
	\item \label{lemma:PropertiesOfTreeAutomata:Deterministic}
		There exists a unique minimal deterministic complete tree-automaton $\mathit{det}(\treeAutomaton)$
		such that $\lang(\mathit{det}(\treeAutomaton)) = \lang(\treeAutomaton)$. Additionally, 
		$\mathit{det}(\treeAutomaton)$ can be constructed in time $2^{O(|\treeAutomaton|)}$. 
	\item \label{lemma:PropertiesOfTreeAutomata:Complement}
		One can construct in time $O(|\treeAutomaton_1|)$ a 
		deterministic complete tree-automaton $\overline{\treeAutomaton_1}$ such that 
		$\lang(\treeAutomaton_1)= \lang(\overline{\treeAutomaton_1})$.
	\item \label{lemma:PropertiesOfTreeAutomata:UnionIntersection}
		One can construct in time $O(|\treeAutomaton_1|\cdot |\treeAutomaton_2|)$ deterministic complete 
		tree-automata $\treeAutomaton_1\cup \treeAutomaton_2$ and $\treeAutomaton_1\cap \treeAutomaton_2$ such that 
		$\lang(\treeAutomaton_1\cup \treeAutomaton_2)  = \lang(\treeAutomaton_1) \cup \lang(\treeAutomaton_2)$ and 
		$\lang(\treeAutomaton_1\cap \treeAutomaton_2)  = \lang(\treeAutomaton_1) \cap \lang(\treeAutomaton_2)$. 
	\item \label{lemma:PropertiesOfTreeAutomata:TermRecognition}
		Let $t\in \terms(\newslicealphabet)$ be a term over $\newslicealphabet$. 
		Then one may determine in time $O(|\treeAutomaton|\cdot |t|)$ whether $t\in \lang(\treeAutomaton)$.
	\item \label{lemma:PropertiesOfTreeAutomata:Projection}
		Let $\projection:\newslicealphabet\rightarrow \newslicealphabet'$ be a projection. Then one can construct in 
		time $O(|\treeAutomaton|)$ a tree-automaton $\projection(\treeAutomaton)$ 
		over $\newslicealphabet'$ such that $\lang(\projection(\treeAutomaton)) = \projection(\lang(\treeAutomaton))$.
	\item \label{lemma:PropertiesOfTreeAutomata:InverseHomomorphism}
		Let $\projection:\newslicealphabet'\rightarrow \newslicealphabet$ 
		be a projection. One can construct in time $O(|\newslicealphabet'|\cdot |\treeAutomaton|)$  a tree-automaton $\projection^{-1}(\treeAutomaton)$ 
		over $\newslicealphabet'$ such that $\lang(\projection^{-1}(\treeAutomaton)) = \projection^{-1}(\lang(\treeAutomaton))$. Additionally, 
		if $\treeAutomaton$  is deterministic, then $\projection^{-1}(\treeAutomaton)$ is also deterministic.
\end{enumerate}
\end{lemma}

\subsection{Weighted Terms}
\label{subsection:WeightedTerms}

Let $\newslicealphabet$ be a ranked alphabet, $\Xi$ be a finite semigroup, and 
let $\automataweightingfunction:\newslicealphabet \rightarrow \Xi$  be a function that associates with  
each symbol $f\in\newslicealphabet$, a weight $\automataweightingfunction(f)\in \Xi$. The weight of a 
term $t\in \terms(\newslicealphabet)$ is inductively defined  as follows. 
\begin{equation}
\automataweightingfunction(t)=\left\{\begin{array}{lcl}
\automataweightingfunction(f) & & \mbox{if $t=f$ for $f\in \newslicealphabet_{0}$} \\
\automataweightingfunction(f) + \sum_{i=1}^{\arity(f)} \automataweightingfunction(t_i) & & \mbox{if $t=f(t_1,...,t_{\arity(f)})$ and $\arity(f)\geq1$} \\
\end{array}\right.
\end{equation} 

The following lemma says that given an alphabet $\newslicealphabet$, a weighting function $\automataweightingfunction:\newslicealphabet \rightarrow \Xi$, and 
a weight $a\in \Xi$, one can construct a tree-automaton $\treeAutomaton(\newslicealphabet,\automataweightingfunction,a)$ generating precisely 
the terms in $\terms(\newslicealphabet)$ with weight $a$. 

\begin{lemma}
\label{lemma:WeightedTerms}
Let $\newslicealphabet = \newslicealphabet_0 \cup ...\cup \newslicealphabet_r$ be a ranked alphabet and $\automataweightingfunction:\newslicealphabet\rightarrow \Xi$ 
be a weighting function on $\newslicealphabet$. Then for each weight $a\in \Xi$, one can construct in time $|\newslicealphabet|\cdot |\Xi|^{O(r)}$ 
a tree-automaton $\treeAutomaton(\newslicealphabet,\automataweightingfunction,a)$ such that 
$\lang(\treeAutomaton(\newslicealphabet,\automataweightingfunction,a))= \{t\in \terms(\newslicealphabet)\;|\; \automataweightingfunction(t) = a\}$. 
\end{lemma}
\begin{proof}
Let $\treeAutomaton = (Q,\newslicealphabet,Q_{F},\Delta)$ where
$$Q=\{\astate_b\;|\; b\in \Xi\} \mbox{\hspace{1cm}} Q_F=\{\astate_a\}$$
$$\Delta = \{(f,\astate_{\automataweightingfunction(f)})\;|\; f\in \newslicealphabet_0\} \cup 
\{(\astate_{b_1},...,\astate_{b_{\arity(f)}},f,\astate_b) \; |\; f\in \newslicealphabet,\; \arity(f)\geq 1, \; b=\automataweightingfunction(f) + \sum_{i=1}^{\arity(f)} b_i\}$$ 

We will show that $\treeAutomaton$ generates precisely the terms in $\terms(\newslicealphabet)$ of weight $a$. 
First, we claim that for each $b\in \Xi$ and each $i\in \N$, 
\begin{equation}
\label{equation:LevelsWeight}
\lang(\treeAutomaton,\astate_b,i) = \{t\in \terms(\newslicealphabet) \; |\; \automataweightingfunction(t) = b,\; t \mbox{ has depth }i \}.
\end{equation}

The proof of this claim follows by induction on $i$. Equation \ref{equation:LevelsWeight} is true for $i=1$, 
since in this case $\lang(\treeAutomaton,\astate_b,1)=\{f\in \newslicealphabet_0\; |\; (f, \astate_b)\in \Delta,\; \automataweightingfunction(f)=b\}$. 
Assume that Equation \ref{equation:LevelsWeight} holds for $i\in \N$. We show that it also holds for $i+1$. By Equation \ref{equation:InductiveDefinition}, 
we have that 
\begin{equation*}
\begin{array}{lcl}
\lang(\treeAutomaton,\astate_b,i+1) = \lang(\treeAutomaton,\astate_b,i) & \cup & 
\{f(t_1,...,t_{\arity(f)})\;|\; (\astate_{b_1},...,\astate_{b_{\arity(f)}},f,\astate_b)\in \Delta_{\arity(f)}',\; \\ 
& & \hspace{2.7cm}  t_j\in \lang(\treeAutomaton,\astate_{b_j},i)\;\} \\
\end{array}
\end{equation*}

By the induction hypothesis, $\automataweightingfunction(t_j)=b_j$ for each $t_j\in \lang(\treeAutomaton, \astate_{b_j},i)$. Therefore 
the weight of $f(t_1,...,t_{\arity(f)})$ is 
$\automataweightingfunction(f) + \sum_{j=1}^{\arity(f)} \automataweightingfunction(t_i) = \automataweightingfunction(f) + \sum_{j=1}^{\arity(f)} b_j$.
Since  $(\astate_{b_1},...,\astate_{b_{\arity(f)}},f,\astate_b)\in \Delta_{\arity(f)}$ if and only if $b=\automataweightingfunction(f) + \sum_{j=1}^{\arity(f)} b_j$, 
our claim is proved. Note that $\astate_a$ is the only final state in $Q_F$. Therefore the language accepted by $\treeAutomaton$ is 

\begin{equation}
\label{equation:Union}
\lang(\treeAutomaton) =  \bigcup_{i\in \N} \lang(\treeAutomaton,\astate_a, i). 
\end{equation}

Since for each $i\in \N$ the language $\lang(\treeAutomaton,\astate_a,i)$ consists of all terms of weight $a$ accepted in 
depth $i$, we have that $\lang(\treeAutomaton)$ is the set of all terms of weight $a$ accepted in any finite depth, proving 
in this way the lemma. 
\end{proof}

\section{Tree Slice Languages}
\label{section:TreeSliceLanguage}

As mentioned in the introduction, the proof of Theorem \ref{theorem:MainTheoremDirectedTreewidth} is based on the framework 
of tree slice languages. We dedicate this section to the introduction of this framework. We start by defining, in Subsection 
\ref{subsection:Slices}, the notion of {\em slice} of arity $r$. Intuitively, a slice of arity $r$ is a digraph whose vertex set is 
partitioned into a center $C$, an out-frontier $F_0$ and $r$ in-frontiers $F_1,...,F_r$. Each such a slice should be 
regarded as a function symbol of arity $r$. Within this point of view, a finite set $\newslicealphabet$ of 
slices with possibly distinct arities can be regarded as a ranked alphabet (Subsection \ref{subsection:SliceAlphabets}).
In Subsection \ref{subsection:GluabilityOfSlices} we introduce a notion of gluability for slices.
A slice $\boldS$ can be glued to a slice $\boldS'$ at frontier $j$ if the out-frontier of $\boldS$ can be matched with the $j$-th in-frontier of $\boldS'$. 
In Subsection \ref{subsection:TermsAndUnitDecompositions} we define {\em unit decompositions}, and {\em tree slice languages}. A unit 
decomposition is a term $\boldT$ over a slice alphabet $\newslicealphabet$ satisfying the property that
each two slices associated with consecutive positions of $\boldT$ can be glued along their matching frontiers. 
A tree slice language $\lang$ is a tree language over a slice alphabet 
$\newslicealphabet$ such that each term $\boldT\in \lang$ is a unit decomposition. A {\em slice tree automaton} is a 
tree automaton $\treeAutomaton$ generating a tree slice language $\lang(\treeAutomaton)$.    In Subsection \ref{subsection:DigraphsUnitDecompositions} 
we show how to associate with each unit decomposition $\boldT$, a digraph $\composedT$ which is intuitively obtained by 
gluing each two consecutive slices of $\boldT$. We can extend this association to slice languages. Namely, the graph language $\lang_{\graph}$
derived from a slice language $\lang$ is the set of all digraphs associated to unit decompositions in $\lang$. 
In Subsection \ref{subsection:Subslices} we will introduce the notion of sub-decompositions
of unit decompositions. Sub-decompositions should be regarded as a slice theoretic analog of the notion of subgraph.
A key idea of this paper is to reduce the problem of counting subgraphs of a digraph to the problem of counting 
sub-decompositions of a unit decomposition. 
In Subsection \ref{subsection:InitialSliceTreeAutomata} we show that given any slice alphabet $\newslicealphabet$, one 
can construct a slice automaton $\treeAutomaton(\newslicealphabet)$ whose slice language consists of all unit decompositions over $\newslicealphabet$. 
Finally, in Subsection \ref{subsection:SliceProjection} we introduce the notion of slice projection, which will be used in many places 
along this paper. 

Our main application for slice languages will be given in Section \ref{section:zSaturationAndCounting} where we will introduce the notion 
of {\em  $z$-saturated tree slice language}. We will use this notion to count subgraphs satisfying interesting properties on digraphs of constant tree-zig-zag number. 
Since the tree-zig-zag number of a digraph is at most a constant times its directed treewidth, we will also be able to count subgraphs 
satisfying interesting properties on digraphs of constant directed treewidth.

\subsection{\bf Slices}
\label{subsection:Slices}

A {\em slice} of arity $r\geq 0$ is a digraph $\boldS = (V,E,s,t,\vertexlabeling,\edgelabeling,[C,F_0,F_1,...,F_r])$ with vertex set 
$V = C\cup F_0 \cup ... \cup F_r$  and edge set $E$. 
The function $s:E\rightarrow V$ associates with each edge $e\in E$ a source 
vertex $e^s$, while the function $t:E\rightarrow V$ associates with each edge $e\in E$ a target vertex $e^t$. We say that $e^s$ and $e^t$
are the endpoints of $e$. The function $\vertexlabeling:C \rightarrow \vertexlabel$ labels each vertex in $C$ with an element 
from a finite set of labels $\vertexlabel$, and $\edgelabeling:  E \rightarrow \edgelabel$ labels each edge in $E$ with an element from a 
finite set of labels $\edgelabel$. We say that $C$ is the center of $\boldS$, $F_0$ is the out-frontier of $\boldS$, 
and for each $j\in \{1,...,r\}$, $F_j$ is the $j$-th in-frontier of $\boldS$. A slice is subject to the following restrictions.

\begin{enumerate}[s1)]
	\item The sets $C,F_0,...,F_r$ are pairwise disjoint. For concreteness, we assume that $C$ is either empty or 
		$C=\{1,...,n\}$ for some $n\in \N$, and that for each $j\in \{0,...,r\}$, the frontier $F_j$ is either 
		empty or $F_j=\{[j,i_{j,1}],...,[j,i_{j,c_j}]\}$ for some $c_j\in \N$, and $i_{j,1}< ... < i_{j,c_j}\in \N$. 
	\item No edge in $E$ has both endpoints in the same frontier.
	\item Each frontier vertex $v\in F_0\cup F_1 \cup ...\cup F_r$ is the endpoint of a unique edge $e$. 
\end{enumerate}

We say that $\boldS$ is a {\em unit slice} if the center $C$ has at most one vertex. In other words in a unit 
slice the center is either empty or the singleton $\{1\}$. In this work we will only be interested in 
unit slices. We say that a frontier $F_j$ is normalized if $i_{j,k}=k$ for each $k\in \{1,...,c_j\}$. A slice 
$\boldS$ is {\em normalized} if all of its frontiers are normalized.
Non-normalized slices will play an important role in Subsection \ref{subsection:Subslices} when considering 
the notion of sub-slice. 
A slice of arity $0$ 
is a slice with no in-frontier. In this case  $\boldS = (V,E,s,t,\vertexlabeling,\edgelabeling,[C,F_0])$ with $V=C \cup F_0$. 
A slice of arity $0$ should not be confused with a slice in which all in-frontiers are empty. Rather, in such a slice the in-frontiers simply do not 
exist.  
In Figure \ref{figure:Slices} we depict three examples of unit slices. 

\begin{figure}[!hf] 
\centering 
\includegraphics[scale=0.35]{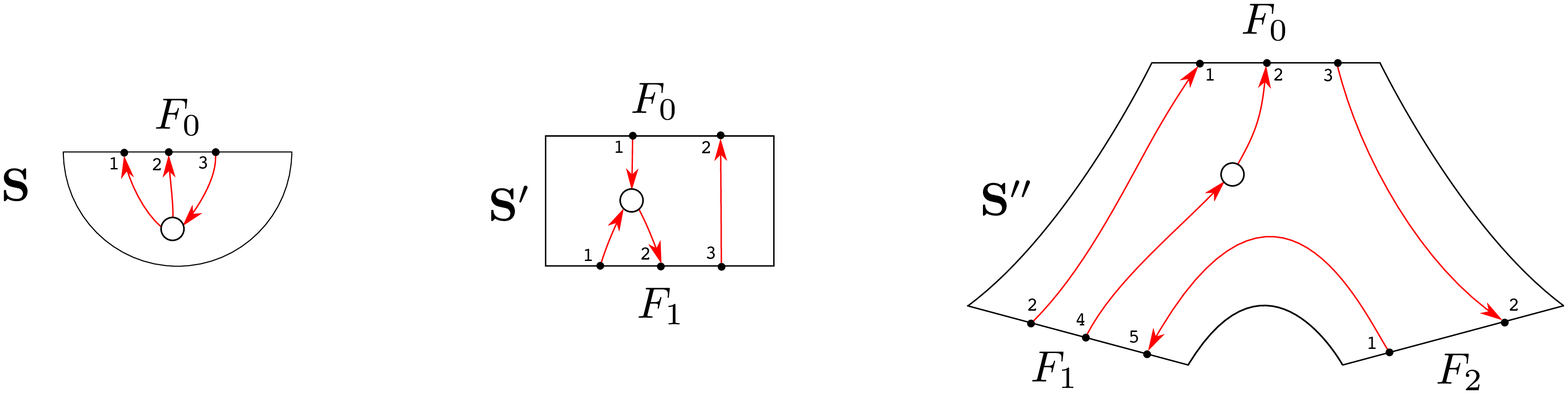}
\caption{$\boldS$ is a slice of arity $0$, $\boldS'$ is a slice of arity $1$ and $\boldS''$ is a slice of arity $2$. The out-frontier 
$F_0$ is always drawn on the top. The in-frontiers $F_1,...,F_r$ are drawn at the bottom and in increasing order from left to right.
For each frontier vertex $[j,i]$ we draw a black dot at frontier $j$ and write the number $i$ near from it. Within each frontier, the black 
dots are drawn in increasing order from left to right. The center vertex, if any, is drawn in the center of each box. The edges are drawn in red. 
The slices $\boldS$ and $\boldS'$ are normalized. The slice $\boldS''$ is not normalized because $F_1 = \{[1,2],[1,4],[1,5]\}$, instead of 
$\{[1,1],[1,2],[1,3]\}$.}
\label{figure:Slices}
\end{figure}

\subsection{\bf Slice Alphabets}
\label{subsection:SliceAlphabets}

A slice alphabet is simply a finite set $\newslicealphabet$ of slices, possibly with different arities. 
Slice alphabets will be used to define terms over slices and to provide sliced representations of digraphs. 
Let $\boldS$ be a slice with frontiers $F_j = \{[j,i_{j,1}],...,[j,i_{j,c_j}]\}$ for $j\in \{0,...,r\}$. 
The width $\width(\boldS)$ of $\boldS$ is the size of its largest frontier, i.e., $\width(\boldS) = \max_{j} \{c_j\}$. 
The {\em extra-width} $\extrawidth(\boldS)$ of $\boldS$ is the greatest number occurring in a frontier of $\boldS$. More precisely, 
$\extrawidth(\boldS) = \max_{j} \{i_{j,c_j}\}$. For instance, in Figure \ref{figure:Slices} the extra-width of the slice $\boldS''$ is $5$. 
For any $c,q,r\in \N$ with $q\geq c$, and any finite sets of labels $\vertexlabel$ and $\edgelabel$, we let $\newslicealphabet_r(c,q,\vertexlabel,\edgelabel)$ 
denote the set of all unit slices of arity $r$, width at most $c$, extra-width at most $q$, whose center vertex (if any) is labelled with an element of $\vertexlabel$, 
and whose edges are labelled with elements of $\edgelabel$. 
Now consider the set 
$$\newslicealphabet(c,q,\vertexlabel,\edgelabel) = \newslicealphabet_0(c,q,\vertexlabel,\edgelabel) \cup 
\newslicealphabet_1(c,q,\vertexlabel,\edgelabel) \cup ...\cup \newslicealphabet_r(c,q,\vertexlabel,\edgelabel).$$
We can view $\newslicealphabet(c,q,\vertexlabel,\edgelabel)$ as a ranked alphabet by regarding each 
slice in $\newslicealphabet_j(c,q,\vertexlabel,\edgelabel)$ as a function symbol of arity $j$. 
We let $\newslicealphabet_j(c,\vertexlabel,\edgelabel)$ denote the subset of $\newslicealphabet_j(c,c,\vertexlabel,\edgelabel)$ consisting only of normalized slices and 
set $\newslicealphabet(c,\vertexlabel,\edgelabel) = \newslicealphabet_0(c,\vertexlabel,\edgelabel) \cup \newslicealphabet_1(c,\vertexlabel,\edgelabel) \cup ...\cup \newslicealphabet_r(c,\vertexlabel,\edgelabel)$. In this work we are only interested in slices of arity at most $2$. Therefore, when considering 
the slice alphabets $\newslicealphabet(c,q,\vertexlabel,\edgelabel)$ and $\newslicealphabet(c,\vertexlabel,\edgelabel)$ defined above, we assume that $r=2$.

\subsection{\bf Gluability of Slices}
\label{subsection:GluabilityOfSlices}

\begin{figure}[!hf] 
\centering 
\includegraphics[scale=0.3]{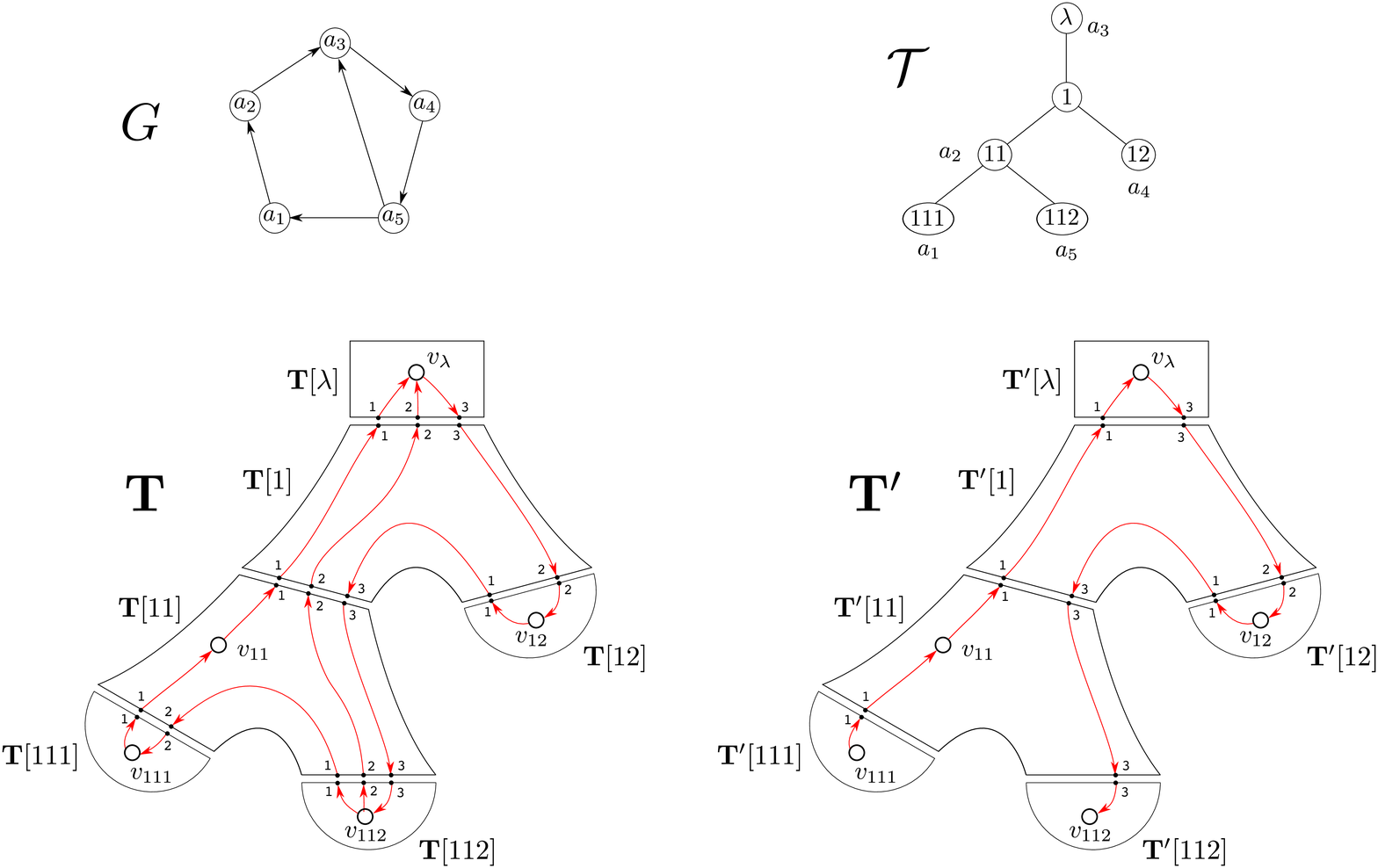}
\caption{$G$ is a digraph, $\olivetreedecomposition=(N,F,\mathfrak{m})$ is an olive-tree decomposition of $G$ and 
$\boldT$ is a unit-decomposition of $G$. Note that $\mathfrak{m}(a_1) = 111$,   $\mathfrak{m}(a_2) = 11$,  $\mathfrak{m}(a_3) = \lambda$,
  $\mathfrak{m}(a_4) = 12$  and   $\mathfrak{m}(a_1) = 112$. The unit decomposition $\boldT$ is compatible with $\olivetreedecomposition$ since
the map defined by  $a_1\rightarrow v_{111}$, $a_2\rightarrow v_{11}$, $a_3\rightarrow v_{\lambda}$, 
$a_4\rightarrow v_{12}$ and $a_5\rightarrow v_{112}$ is an isomorphism
from $G$ to $\stackrel{\circ}{\boldT}$. The unit decomposition $\boldT'$ is a sub-decomposition of $\boldT$.} 
\label{figure:UnitDecomposition}
\end{figure}

If $\boldS$ is a slice and $[j,i]$ is a vertex in the $j$-th frontier of $\boldS$, then we denote by 
$e(\boldS,j,i)$ the unique edge of $\boldS$ that has $[j,i]$ as endpoint. 
Let $\boldS = (V,E,\vertexlabeling,\edgelabeling,[C,F_0,F_1,...,F_r])$ and $\boldS' = (V',E',\vertexlabeling',\edgelabeling',[C',F_0',F_1', ...,F_r'])$ 
be two slices in $\newslicealphabet(c,q,\vertexlabel,\edgelabel)$. We say that $\boldS$ can be glued to $\boldS'$ at frontier $j$, for $1\leq j\leq r$, if the out-frontier 
of $\boldS$ can be coherently matched with the $j$-th in-frontier of $\boldS'$. Formally, $\boldS$ can be 
glued to $\boldS'$ at frontier $j$ if the following conditions are satisfied.

\begin{enumerate}
	\renewcommand{\theenumi}{g\arabic{enumi}}
	\item\label{gluingone} For each $i\in \{1,...,q\}$, $[0,i]\in F_0$ if and only if $[j,i]\in F_j'$. 
	\item\label{gluingtwo} $\edgelabeling(e(\boldS,0,i)) = \edgelabeling(e(\boldS',j,i))$. 
	\item\label{gluingthree} Either $[0,i]$ is the target of $e(\boldS,0,i)$ and $[j,i]$ is the source of $e(\boldS',j,i)$
		or $[0,i]$ is the source of $e(\boldS,0,i)$ and $[j,i]$ is the target of $e(\boldS',j,i)$.
\end{enumerate}

Intuitively, Condition \ref{gluingone} says that the vertex $[0,i]$ in the out-frontier of $\boldS$  is matched with the 
vertex $[j,i]$ in the $j$-th in-frontier of $\boldS'$. Condition \ref{gluingtwo} says that the unique edge of $\boldS$ having $[0,i]$
as endpoint has the same label as the unique edge of $\boldS'$ having $[j,i]$ as endpoint. Finally, Condition \ref{gluingthree}
says that these edges must also agree in direction. For instance, in Figure \ref{figure:UnitDecomposition}, the slice 
$\boldT[11]$ can be glued to the slice $\boldT[1]$ at frontier $1$. While $\boldT[12]$ can be 
glued to $\boldT[1]$ at frontier $2$. 

\subsection{\bf Terms over Slices, Unit Decompositions and Tree Slice Languages}
\label{subsection:TermsAndUnitDecompositions}

As observed in Subsection \ref{subsection:SliceAlphabets}, a slice alphabet $\newslicealphabet$  can be 
regarded as a ranked alphabet where each slice $\boldS\in \newslicealphabet$ of arity $r$ is 
a function symbol of arity $r$. In this paper $\newslicealphabet$ will be typically the 
slice alphabet $\newslicealphabet(c,q,\vertexlabel,\edgelabel)$ or the normalized slice alphabet $\newslicealphabet(c,\vertexlabel,\edgelabel)$, both 
defined in Subsection \ref{subsection:SliceAlphabets}.
We let $\terms(\newslicealphabet)$ denote the set of all terms 
formed with slices from $\newslicealphabet$. 
In this work however we will be only interested on terms over $\newslicealphabet$ that can give rise 
to digraphs. These terms are called {\em unit decompositions}.  

\begin{definition}[Unit Decomposition]
\label{definition:UnitDecomposition}
Let $\newslicealphabet$ be an alphabet of unit slices. A term $\boldT\in \terms(\newslicealphabet)$ is a {\em unit pre-decomposition} if for each 
two consecutive positions $p, pj\in \positions(\boldT)$, the slice $\boldT[pj]$ can be glued to the slice $\boldT[p]$ at frontier  $j$. A term 
$\boldT$ is a {\em unit decomposition} if it is a unit pre-decomposition in which the slice $\boldT[\lambda]$ at the root of $\boldT$ has empty out-frontier. 
\end{definition}

The width $w(\boldT)$ of a unit decomposition $\boldT$ is the maximum width of a slice occurring in it.
A unit decomposition is normalized if for each position $p\in \positions(\boldT)$
the slice $\boldT[p]$ is normalized. For instance, the unit decomposition $\boldT$ in Figure \ref{figure:UnitDecomposition} 
is normalized while the unit decomposition $\boldT'$ in the same figure is not. 

We let $\lang(\newslicealphabet)$ be the set of all unit decompositions in $\terms(\newslicealphabet)$. 
A tree slice language over $\newslicealphabet$ is any subset $\lang$ of $\lang(\newslicealphabet)$. 
We say that a tree slice language $\lang\subseteq \lang(\newslicealphabet)$ is normalized if all unit 
decompositions in $\lang$ are normalized.  We will see in the next subsection 
that with each unit decomposition $\boldT$ one can associate a digraph $\composedT$ which is intuitively obtained 
by gluing each two consecutive slices in $\boldT$. Thus with any slice language $\lang$ one can associate a graph 
language $\lang_{\graph}$ consisting of all digraphs that correspond to unit decompositions in $\lang$. 

Of particular importance to us are the slice languages that can be effectively represented via tree-automata 
over slice alphabets. We call these automata {\em slice tree-automata}.

\begin{definition}[Slice Tree-Automaton]
\label{definition:SliceTreeAutomaton}
Let $\newslicealphabet$ be a slice alphabet. We say that a tree-automaton $\treeAutomaton=(Q,\newslicealphabet,Q_{F},\Delta)$ over $\newslicealphabet$ 
is a {\em slice tree-automaton} if for each term $\boldT\in \lang(\treeAutomaton)$, $\boldT$ is a unit decomposition over $\newslicealphabet$. 
\end{definition}

In other words, $\treeAutomaton$ is a slice tree-automaton if $\lang(\treeAutomaton)\subseteq \lang(\newslicealphabet)$. In this case 
we say that $\lang(\treeAutomaton)$ is the slice language generated by $\treeAutomaton$. We say that a slice tree automaton $\treeAutomaton$ 
is normalized if the slice language $\lang(\treeAutomaton)$ is normalized.

\subsection{\bf Digraphs associated with Unit Decompositions}
\label{subsection:DigraphsUnitDecompositions}

Each unit decomposition $\boldT\in \lang(\newslicealphabet(c,q,\vertexlabel,\edgelabel))$  can be 
associated with a digraph $\composedT$ which is intuitively obtained by gluing together each two 
consecutive slices in $\boldT$. For instance, gluing the slices of the unit decomposition 
$\boldT$ of Figure \ref{figure:UnitDecomposition} we get the digraph $G$. To make this notion 
of gluing more precise, it will be convenient to define the notion of {\em sliced edge sequence}. 
Intuitively, each edge $e$ of the digraph $\composedT$ will be defined with basis on a sliced 
edge sequence that contains all "sliced parts" of $e$.
Below, $\support(\boldT)$ denotes the set of all positions in $\positions(\boldT)$ for 
which the slice $\boldT[p]$ has non-empty center. 

\begin{definition}[Sliced Edge Sequence]
\label{definition:SlicedEdgeSequence}
Let $\boldT$ be a unit decomposition over a slice alphabet $\newslicealphabet$. 
Let $p,p'$ be two positions in $\support(\boldT)$. 
A sliced edge sequence from $p$ to $p'$ is a sequence  
\begin{equation}
\label{equation:slicedEdgeSequence}
K\equiv(p_1,a_{1},e_{1},b_{1})(p_2,a_{2},e_{2},b_{2})...(p_n,a_{n},e_{n},b_{n})
\end{equation}
\noindent where $p_1=p$, $p_n=p'$,  and the following conditions are satisfied. 
\begin{enumerate}
	\itemsep0.35em
	\item \label{item:edgesequence1} For each $i\in \{1,...,n\}$, $e_{i}$ is an edge in $\boldT[p_i]$ with source $a_{i}$ and target $b_{i}$.
	\item \label{item:edgesequence2} $a_{1}$ is the center vertex of $\boldT[p_1]$ and $b_{n}$ is the center vertex of $\boldT[p_{n}]$. 
	\item \label{item:edgesequence3} For each $i\in \{1,...,n-1\}$, there is a $j$ such that either $p_i = p_{i+1} j$ or $p_{i+1} = p_i j$.
	\item \label{item:edgesequence4} If $p_{i} = p_{i+1} j$ then for some $k\in \{1,...,q\}$, $b_i = [0,k]$ and $a_{i+1} = [j,k]$.
	\item \label{item:edgesequence5} If $p_{i+1} = p_i j$ then for some $k\in \{1,...,q\}$, $b_i= [j,k]$ and $a_{i+1} = [0,k]$. 
\end{enumerate}
\end{definition}

We note that Conditions \ref{item:edgesequence1}-\ref{item:edgesequence5} of Definition \ref{definition:SlicedEdgeSequence} together with the fact that $\boldT$ is a unit 
decomposition ensures that the $p_i\neq p_j$ for $i\neq j$. 
To illustrate Definition \ref{definition:SlicedEdgeSequence} we note that in the unit decomposition $\boldT$ of Figure 
\ref{figure:UnitDecomposition} there is a sliced edge sequence from position $\lambda$ to position $12$, a sliced edge sequence from $12$ to $112$ and so 
on. Intuitively, each sliced edge sequence $K$ gives rise to an edge $e_K$ in the digraph $\composedT$ that is obtained by gluing 
all of its sliced parts $e_{1},...,e_{n}$. Condition \ref{item:edgesequence1} says that $e_{i}$ is the sliced part of $e_K$ 
lying at the slice $\boldT[p_i]$. Condition
\ref{item:edgesequence2} says that the source of the first sliced part of $e_K$ is the center vertex of $\boldT[p_1]$ and the target of the last 
sliced part of $e_K$ is the center vertex of $\boldT[p_n]$. Condition \ref{item:edgesequence3} says that for each $i\in \{1,...,n-1\}$, $e_{i}$ and $e_{{i+1}}$ 
lie in neighboring slices of $\boldT$. If $p_{i+1}=p_i j$ then the edge $e_i$ is intuitively directed towards the $j$-th in-frontier of 
$\boldT[p_i]$. In this case, Condition \ref{item:edgesequence4}, says that the target of $e_{i}$ lies in the $j$-th in-frontier of $\boldT[p_i]$ 
while the source of $e_{{i+1}}$ lies in the out-frontier of $\boldT[p_{i+1}]$. 
On the other hand, if $p_i=p_{i+1}j$ then the edge $e_{i}$ is intuitively directed towards the out-frontier of $\boldT[p_{i}]$. 
In this case, Condition \ref{item:edgesequence5} says that the target of $e_{i}$ lies in 
the out-frontier of $\boldT[p_i]$ and the source of $e_{{i+1}}$ lies in the $j$-th in frontier of $\boldT[p_{i+1}]$.

Let $\boldT$ be a unit decomposition and for each $p\in \positions(\boldT)$ let $\boldT[p] = (V_p,E_p,\vertexlabeling_p,\edgelabeling_p)$ be the slice of 
$\boldT$ at position $p$. The digraph $\composedT = (V,E,\vertexlabeling,\edgelabeling)$ associated with $\boldT$ is defined as follows. 
First, for each position $p\in \support(\boldT)$, we add a vertex $v_p$ to the vertex set $V$. Subsequently, for each two positions 
$p,p'\in \positions(\boldT)$ and each sliced edge sequence $K$ from $p$ to $p'$ we add an edge $e_K$ to $E$ and set its source as $e_K^s = v_p$ and its target 
as $e_K^t = v_{p'}$. Observe that multiple edges are allowed in $\composedT$ since for some pair of positions $p,p'$ there may exist 
more than one sliced edge sequence from $p$ to $p'$. For each $p\in \positions(\boldT)$, the 
vertex $v_p$ receives the same label as the center vertex of $\boldT[p]$. In other words, $\vertexlabeling(v_p) = \vertexlabeling_p(1)$ 
\footnote{We recall that if the center of a unit slice is not empty then the center is the singleton $\{1\}$.}. We note that  if 
$K$ is a sliced edge sequence as defined in Equation \ref{equation:slicedEdgeSequence} then Conditions \ref{item:edgesequence4} and \ref{item:edgesequence5} 
of Definition \ref{definition:SlicedEdgeSequence} together with Condition \ref{gluingtwo} (of Subsection \ref{subsection:GluabilityOfSlices}) 
guarantee that all edges $e_{1},e_{2},...,e_{n}$ have the same label.
Thus the label of the edge $e_K$ is set as $\edgelabeling(e_K) = \edgelabeling_{p_1}(e_{1})=...=\edgelabeling_{p_n}(e_{n})$. 

If $G=(V,E,\vertexlabeling,\edgelabeling)$ is a digraph where $\vertexlabeling:V \rightarrow \vertexlabel$ and $\edgelabeling:E\rightarrow \edgelabel$
are vertex and edge labeling functions respectively, then for each two vertices $v,v'\in V$ and each label $b\in \edgelabel$, 
we let $\overrightarrow{E}(v,v',b) = \{e\;|\;e^s = v,\;e^t = v', \edgelabeling(e) = b\}$ denote the set of all edges in $E$ which have 
$v$ as source vertex, $v'$ as target vertex and $b$ as label. 
An isomorphism from a digraph $G_1=(V_1,E_1,\vertexlabeling_1,\edgelabeling_1)$ to a digraph $G_2 = (V_2,E_2,\vertexlabeling_2,\edgelabeling_2)$
is a bijection $\phi:V_1\rightarrow V_2$ from $V_1$ to $V_2$ such that for each $v\in V_1$, 
$\rho_1(v) = \rho_2(\phi(v))$, and such that for each two vertices $v,v'\in V_1$ and each label $b\in \edgelabel$,
$|\overrightarrow{E}_1(v,v',b)| = |\overrightarrow{E}_2(\phi(v),\phi(v'),b)|$. 
A canonization function for finite digraphs is a function $[\;\cdot\;]$ satisfying two properties. First, for every digraph $G$, 
$[G]$ is a digraph isomorphic to $G$. Second, for every two digraphs $G_1$ and $G_2$, $G_1$ is isomorphic to $G_2$ if and 
only if $[G_1]=[G_2]$. We say that $[G]$ is the canonical form of $G$. In this paper we let $[\;\cdot\;]$ be an arbitrary but fixed canonization
function for finite digraphs.

We say that a term $\boldT$ is a unit decomposition of a digraph $G$ if the digraph $\composedT$ is isomorphic to $G$. 
Since with any unit decomposition $\boldT$ one can associate a digraph $\composedT$, with any tree slice language $\lang$ one 
can associate a possibly infinite family of digraphs.
If $\lang$ is a tree slice language over an alphabet $\newslicealphabet$ of unit slices, then the graph language derived from $\lang$ is the set $\lang_{\graph}$ 
of canonical forms of digraphs obtained by composing the slices of each unit decomposition in $\lang$. Formally, 
\begin{equation}
\label{equation:GraphLanguage}
\lang_{\graph} = \{[\composedT]\;|\; \boldT\in \lang\}.
\end{equation}
For convenience, in some places we may simply say that a digraph 
$H$ belongs to $\lang_{\graph}$ instead of saying that $[H]$ belongs
to $\lang_{\graph}$. If $\treeAutomaton$ is a slice tree automaton then 
we denote by $\lang_{\graph}(\treeAutomaton)$ the graph language derived 
from $\lang(\treeAutomaton)$.

\subsection{Sub-slices and Sub-Decompositions}
\label{subsection:Subslices}

In this subsection we introduce the notions of {\em sub-slice} and of {\em sub-decomposition}. Intuitively, 
the notion of sub-decomposition is a sliced version of the notion of subgraph.
Let $\boldS=(V,E,\vertexlabeling,\edgelabeling,[C,F_0,F_1,...,F_r])$ be a slice of arity $r$. We say that a slice 
$\boldS'$ is a {\em sub-slice} of $\boldS$ if $\boldS'=(V',E',\vertexlabeling',\edgelabeling',[C',F_0',F_1',...,F_r'])$ where $V'\subseteq V$, 
$E'\subseteq E$, $\vertexlabeling'=\vertexlabeling|_{V'}$, $\edgelabeling'= \edgelabeling|_{E'}$, $C'\subseteq C$ and 
$F_j'\subseteq F_j$ for each $j\in {0,1,...,r}$. In other words, a sub-slice of $\boldS$ is a subgraph of $\boldS$ that 
is also a slice. Labels of vertices and edges in a sub-slice are inherited from the original slice. 
We note that even if $\boldS$ is a normalized slice, a sub-slice $\boldS'$ of $\boldS$ may not be 
normalized. For instance, in Figure \ref{figure:UnitDecomposition}, the slice $\boldT'[1]$ is a sub-slice of $\boldT[1]$. 
Note that $\boldT'[1]$ is not normalized even though $\boldT[1]$ is. We also call attention to the fact that a sub-slice 
has always the same arity as the original slice, and that the empty slice $\emptyslice_r$ of arity $r$ is a sub-slice of any slice of 
arity $r$.

\begin{definition}[Sub-decomposition]
\label{definition:Subdecomposition}
Let $\newslicealphabet$ be a slice alphabet and let $\boldT$ and $\boldT'$ be unit decompositions in $\lang(\newslicealphabet)$. 
We say that $\boldT'$ is a sub-decomposition of $\boldT$ if the following conditions are satisfied.
\begin{enumerate}[i)]
	\item\label{subdecomposition1} $\positions(\boldT)=\positions(\boldT')$,
	\item\label{subdecomposition2} for each $p\in \positions(\boldT)$ the unit slice $\boldT'[p]$ is a sub-slice of $\boldT[p]$, 
	\item\label{subdecomposition3} for each two consecutive positions $p,p\lastposition \in \positions(\boldT)$ the slice 
		$\boldT'[p\lastposition]$ can be glued to the slice $\boldT'[p]$ at frontier $\lastposition$.
\end{enumerate}
\end{definition}

Conditions \ref{subdecomposition1}-\ref{subdecomposition3} of Definition \ref{definition:Subdecomposition} guarantee that if $\boldT'$ is a sub-decomposition of $\boldT$ then 
the digraph $\composedTprime$ is a subgraph of $\composedT$. 
We emphasize that $\composedTprime$ is an actual subgraph of $\composedT$ and not merely isomorphic to a subgraph of $\composedT$. 
Conversely, for each subgraph $H$ of $\composedT$ there is a sub-decomposition $\boldT'$ of $\boldT$ for which $\composedTprime = H$.
Again at this point we are speaking about strict equality, and not merely isomorphism.  
Thus each sub-decomposition of $\boldT$ unequivocally corresponds to a subgraph of $\composedT$. 
A crucial step towards the proof of Theorem \ref{theorem:MainTheoremDirectedTreewidth} will consist in reducing the 
problem of counting subgraphs of a digraph to the problem of counting sub-decompositions of a unit 
decomposition.

\subsection{Initial Slice Tree-Automata}
\label{subsection:InitialSliceTreeAutomata}

In this section we will show that for each slice alphabet $\newslicealphabet$ one can construct a deterministic slice tree-automaton 
$\treeAutomaton(\newslicealphabet)$ whose slice language consists of all unit decompositions that can be formed with elements from $\newslicealphabet$.
We say that $\treeAutomaton(\newslicealphabet)$ is the initial tree-automaton for $\newslicealphabet$. 
Before proceeding, we define the notion of identity slice, which will be used below in the proof of Proposition \ref{proposition:InitialAutomaton},
 and later, in the proof of Lemma \ref{lemma:SubgraphsC}.
An {\em identity slice} in $\newslicealphabet(c,q,\vertexlabel,\edgelabel)$ 
is a slice $\identitySlice=(V,E,\vertexlabeling,\edgelabeling,[C,F_0,F_1])$ of arity $1$ with empty center
($C=\emptyset$) in which all edges are "parallel". In other words for each $e\in E$, there exists a $k\in \{1,...,q\}$ such that either $e^s = [0,k]$ and 
$e^t=[1,k]$ or $e^s=[1,k]$ and $e^t=[0,k]$.  

\begin{figure}[!hf] 
\centering 
\includegraphics[scale=0.25]{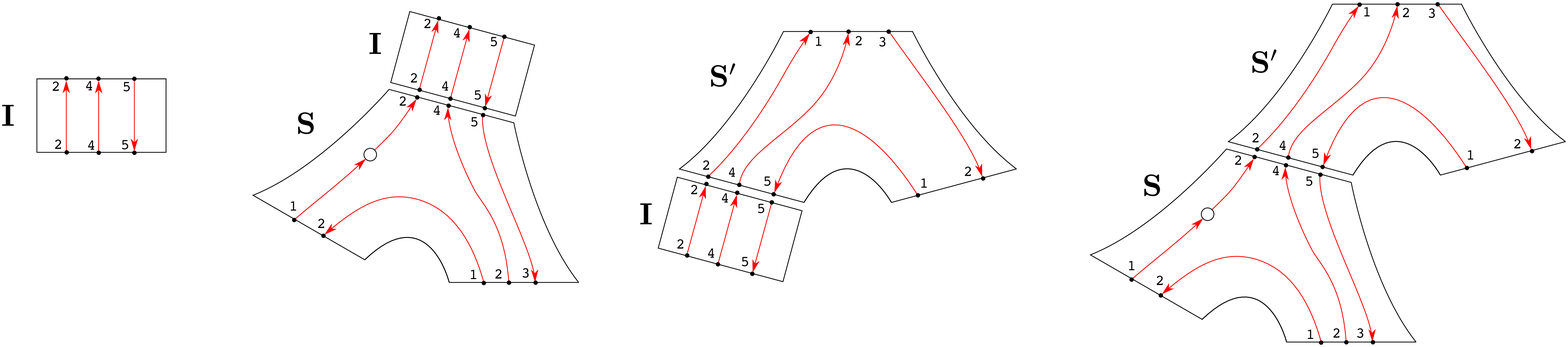}
\caption{An identity slice $\identitySlice$ and two other slices $\boldS$ and $\boldS'$. A slice $\boldS$ can be glued to $\boldS'$ at 
frontier $j$ if and only if there is a unique identity slice $\identitySlice$ such that $\boldS$ can be glued to $\identitySlice$, and
such that $\identitySlice$ can be glued to $\boldS'$ at $\mbox{frontier $j$}$. In this case $\identitySlice = \identitySlice(\boldS)$. 
}
\label{figure:IdentitySlice}
\end{figure} 

Our only interest in identity slices stems from the 
fact that for each slice $\boldS$ of arity $r$ there is a unique identity slice $\identitySlice$ for which 
$\boldS$ can be glued to $\identitySlice$. We denote this unique identity slice by $\identitySlice(\boldS)$. 
Additionally, for each $j\in \{1,...,r\}$ there exists a unique identity slice 
$\identitySlice_j$ such that $\identitySlice_j$ can be glued to $\boldS$ at frontier $j$. This implies that a
slice $\boldS$ can be glued to a slice $\boldS'$ at $\mbox{frontier $j$}$ if and only if $\identitySlice(\boldS)$ can be 
glued to $\boldS'$ at frontier $j$ (See Figure \ref{figure:IdentitySlice}). We observe that we consider $\emptyslice_1$, the empty slice 
of arity one, as an identity slice.

\begin{proposition}[Initial Slice Tree-Automaton]
\label{proposition:InitialAutomaton}
Let $\newslicealphabet$ be a slice alphabet and let $r$ be the maximum arity of a slice in $\newslicealphabet$. Then 
one can construct in time $O(|\newslicealphabet|)$ a slice tree-automaton $\treeAutomaton(\newslicealphabet)$ whose 
slice language $\lang(\treeAutomaton(\newslicealphabet))$ is  the set of all unit decompositions over $\newslicealphabet$. 
\end{proposition}
\begin{proof}
We construct the automaton $\treeAutomaton(\newslicealphabet) = (Q,\newslicealphabet,Q_F,\Delta)$ explicitly. 
First, we define the set $\identitySlice(\newslicealphabet) = \{\identitySlice(\boldS)\;|\; \boldS\in \newslicealphabet\}$ which 
consist of all identity slices $\identitySlice$ for which some slice in $\newslicealphabet$ can be glued to $\identitySlice$. The 
set of states $Q$ has one state $\astate_{\identitySlice}$ for each identity slice $\identitySlice$ in $\identitySlice(\newslicealphabet)$. 
The set of final states is the singleton $Q_F=\{\astate_{\emptyslice_1}\}$. The transition relation 
$\Delta$ has one transition $(\astate_{\identitySlice_1},...,\astate_{\identitySlice_r},\boldS,\astate_{\identitySlice(\boldS)})$ for 
each slice $\boldS$ of arity $r$ in $\newslicealphabet$, where for each $j\in \{1,...,r\}$, $\identitySlice_i$ is 
the unique identity slice that can be glued to $\boldS$ at frontier $j$. Observe that since the states $\astate_{\identitySlice_1},...,\astate_{\identitySlice_r}$ and 
$\astate_{\identitySlice(\boldS)}$ are completely determined by $\boldS$, the relation $\Delta$ has $|\newslicealphabet|$ transitions.

By the construction of the transition relation $\Delta$ we have that for each term $\boldT$ accepted by $\treeAutomaton(\newslicealphabet)$
and each two consecutive positions $p,pj$ in $\positions(\boldT)$, the slice $\boldT[pj]$ can be glued 
to the slice $\boldT[p]$ at frontier $j$. Since the unique accepting state is $\astate_{\emptyslice_1}$ we also have that $\boldT[\lambda]$
has empty out-frontier. This implies that each such term $\boldT$ is a unit decomposition. For the converse, let 
$\boldT$ be a unit pre-decomposition over $\newslicealphabet$. We will show by induction on the height of $\boldT$ that 
$\boldT$ reaches the state $\astate_{\identitySlice(\boldT[\lambda])}$. This implies in particular that if $\boldT$ is a unit 
decomposition, then $\boldT$ reaches the unique accepting state $\astate_{\emptyslice_1}$, since in this case $\boldT[\lambda]$ can 
be glued to $\emptyslice_1$. In the base case, let $\boldT$ be a unit pre-decomposition 
of height $0$. Then $\boldT$ consists of a single slice $\boldS$ of $\mbox{arity $0$}$. By definition of $\Delta$, we have that there is a 
transition $(\boldS,\astate_{\identitySlice(\boldS)})\in \Delta$ and therefore $\boldS$ reaches the state $\astate_{\identitySlice(\boldS)}$. 
Now suppose that the claim is valid for every unit pre-decomposition of height at most $h$ and let $\boldT$ be a unit pre-decomposition 
of height $h+1$. Let the slice $\boldT[\lambda]$ at the root of $\boldT$ have arity $r$. By the induction hypothesis, for each 
$i\in \{1,...,r\}$, the subterm $\boldT|_{i}$ rooted at position $i$, reaches the state $\astate_{\identitySlice_i}$ where 
$\identitySlice_i = \identitySlice(\boldT|_{i}[\lambda])$. Since by the construction of $\Delta$ the transition
 $(\astate_{\identitySlice_1},...,\astate_{\identitySlice_r},\boldT[\lambda],\astate_{\identitySlice(\boldT[\lambda])})$ belongs to $\Delta$, 
we have that $\boldT$ reaches the state $\astate_{\identitySlice(\boldT[\lambda])}$. This proves the inductive step. 
\end{proof}

\subsection{Normalizing Projection and Unweighting Projection}
\label{subsection:SliceProjection}

We say that a mapping $\projection:\newslicealphabet  \rightarrow \newslicealphabet'$
between slice alphabets is a slice projection if $\projection$ is arity preserving, gluing preserving, and empty-frontier preserving. 
By arity preserving we mean that $\arity(\boldS) = \arity(\projection(\boldS))$. By gluing preserving we mean that if $\boldS$ can be 
glued to $\boldS'$ at frontier $j$ then $\projection(\boldS)$ can be glued to $\projection(\boldS')$ at frontier $j$. And by empty-frontier 
preserving we mean that if a frontier $F_i$  is empty in $\boldS$ then the corresponding frontier in $\projection(\boldS)$ is also 
empty. Two classes of slice projections will be of particular importance to us. The normalizing projections, and the unweighting projections which 
are defined below.

\begin{figure}[!hf] 
\centering 
\includegraphics[scale=0.40]{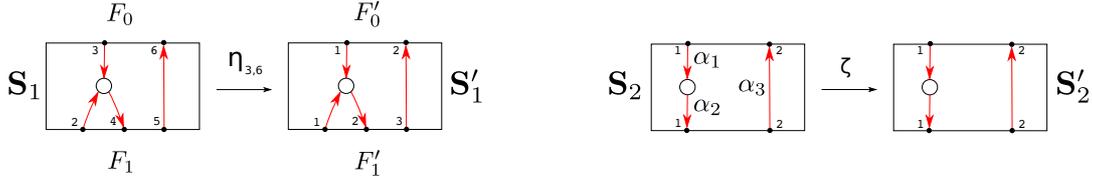}
\caption{The normalizing projection $\normalizingprojection_{c,q}$ normalizes each frontier of a slice in such a way that 
the order of the vertices in each frontier is preserved. In this example $c=3$ and $q=6$. The unnormalized slice $\boldS_1$ has frontiers
 $F_0 = \{[0,3],[0,6]\}$ and $F_1 = \{[1,2],[1,4],[1,5]\}$. After the normalization the frontiers become $F_0' = \{[0,1],[0,2]\}$
and $F_1'= \{[1,1],[1,2],[1,3]\}$. The unweighting projection $\unweightingprojection$ simply erases the weights associated to each 
edge of the slice. In this example, the weights $\alpha_1,\alpha_2,\alpha_3$ attached to the edges of the slice $\boldS_2$ are erased by $\unweightingprojection$. 
}
\label{figure:NormalizingAndUnweighting}
\end{figure} 

The normalizing projection $\normalizingprojection_{c,q}:\newslicealphabet(c,q,\vertexlabel,\edgelabel)\rightarrow \newslicealphabet(c,\vertexlabel,\edgelabel)$ 
acts on each slice $\boldS$ in $\newslicealphabet(c,q,\vertexlabel,\edgelabel)$ by renumbering the frontier vertices of $\boldS$ in 
such a way that the new resulting slice $\normalizingprojection_{c,q}(\boldS)$ is normalized and in such a way that the ordering 
of the vertices inside each frontier is preserved. More precisely, for each $j\in \{0,...,r\}$, let $F_j =\{[j,i_{j,1}],[j,i_{j,2}], ..., [j,i_{j,c_j}]\}$
where $c_j\leq c$ and $i_{j,1} < i_{j,2} < ... < i_{j,c_j} \leq q$. 
Then the slice $\normalizingprojection_{c,q}(\boldS)$ is obtained from $\boldS$ by replacing each frontier vertex $[j,i_{j,k}]$ with the vertex $[j,k]$.
After the application of the normalizing projection $\normalizingprojection_{c,q}$ the $j$-th frontier of $\boldS$ becomes 
$F_j' =\{[j,1],...,[j,c_j]\}$ (Figure \ref{figure:NormalizingAndUnweighting}). 
We note that if $\boldT$ is a unit decomposition over $\newslicealphabet(c,q,\vertexlabel,\edgelabel)$, 
then $\normalizingprojection_{c,q}(\boldT)$ is a normalized unit decomposition over $\newslicealphabet(c,\vertexlabel,\edgelabel)$ representing the same 
digraph. In other words, if $\boldT'=\normalizingprojection_{c,q}(\boldT)$ then $\composedTprime = \composedT$. 

If $\edgelabel$ is a set of edge labels and $\Omega$ is a set of edge weights, then the set $\edgelabel\times \Omega$ can be 
regarded as a new set of edge labels. 
The unweighting projection $\unweightingprojection_{\Omega}:\newslicealphabet(c,q,\vertexlabel,\edgelabel\times \Omega)\rightarrow \newslicealphabet(c,q,\vertexlabel,\edgelabel)$ 
is a function that takes a slice $\boldS\in \newslicealphabet(c,q,\vertexlabel, \edgelabel\times \Omega)$ and erases the weight coordinate 
from the label of each edge. More precisely, if $\boldS = (V,E,\vertexlabeling,\edgelabeling\times \graphweightingfunction,[C,F_0,...,F_j])$
where $\edgelabeling\times \graphweightingfunction: E\rightarrow \edgelabel \times \Omega$, then 
$\unweightingprojection_{\Omega}(\boldS) = (V,E,\vertexlabeling,\edgelabeling,[C,F_0,...,F_j])$ where $\edgelabeling:E\rightarrow \edgelabel$ is 
the projection of $\edgelabeling\times \graphweightingfunction$ to its first coordinate. Unweighting projections and normalizing projections 
will be used in Section \ref{section:ProofOfMainTheorem} to construct the slice tree-automaton $\treeAutomaton(\varphi,k,z,l,a)$ mentioned in the introduction 
(Section \ref{subsection:ProofTechniques}).

\section{$z$-Saturated Tree Slice Languages}
\label{section:zSaturationAndCounting}

In this section we will define the notion of tree-zig-zag number of a unit decomposition and the notion of  {\em  $z$-saturated} tree 
slice language. We will show that given a $z$-saturated tree slice language $\lang$ generated by a slice tree-automaton 
$\treeAutomaton$ and a unit decomposition $\boldT$ of tree-zig-zag $\mbox{number $z$,}$ we can count in polynomial time the number of subgraphs 
of $\composedT$ that are isomorphic to some digraph in $\lang_{\graph}$. 
This seemingly abstract result is a crucial step towards the proof of our main theorem (Theorem \ref{theorem:MainTheoremDirectedTreewidth}). 
The next crucial step, which will be carried in Section \ref{section:MSOandTreeSliceLanguages}, consists in 
showing that for any \msotwo logic sentence $\varphi$ and any $k,z\in \N$, one can define a $z$-saturated slice tree-automaton 
generating precisely the set of digraphs that at the same time are the union of $k$ directed paths and satisfy $\varphi$. 

\subsection{$z$-Saturation}
\label{subsection:zSaturation}

Let $\boldT$ be a unit decomposition over $\newslicealphabet(c,q,\vertexlabel,\edgelabel)$ and let $\composedT=(V,E,\vertexlabeling,\edgelabeling)$ 
be the digraph represented by $\boldT$. 
We say that $\boldT$ is compatible with an olive-tree decomposition $\olivetreedecomposition=(N,F,\mathfrak{m})$ 
of a digraph $G=(V',E',\vertexlabeling',\edgelabeling')$ if both $\boldT$ and $\olivetreedecomposition$ have the 
same tree-structure (i.e. $N=\positions(\boldT)$), and the map $\beta:V'\rightarrow V$ given by $\beta(u)=v_{\mathfrak{m}(u)}$ 
is an isomorphism from $G$ to $\composedT$. For instance, in Figure \ref{figure:UnitDecomposition}, the unit decomposition $\boldT$ is 
compatible with the olive-tree decomposition $\olivetreedecomposition$. Note that for each unit decomposition $\boldT$ 
there is a unique olive-tree decomposition $\olivetreedecomposition = (N,F,\mathfrak{m})$ of the digraph $\composedT$ such that 
$\boldT$ is compatible with $\olivetreedecomposition$. In this olive-tree decomposition, $N=\positions(\boldT)$ and $\mathfrak{m}$ is defined by setting 
$\mathfrak{m}(v_p) = p$ for each position $p\in \positions(\boldT)$.

We say that a unit decomposition $\boldT$ has {\em tree-zig-zag number} $\treezigzagnumber(\boldT) = z$ if $\boldT$ is compatible with an olive-tree decomposition 
of tree-zig-zag number $z$. Intuitively, $\boldT$ has tree-zig-zag $\mbox{number $z$}$ if each simple path of $\composedT$ crosses each frontier of each slice in $\boldT$
at most $z$ times. For instance, in Figure \ref{figure:UnitDecomposition}, the unit decomposition $\boldT$ has tree-zig-zag number $2$. Note that 
the olive-tree decomposition $\olivetreedecomposition$ in Figure \ref{figure:UnitDecomposition} that is compatible with $\boldT$ has also tree-zig-zag number 
$2$. We say that a slice language $\lang$ has tree-zig-zag number $z$ if each unit decomposition in $\lang$ 
has tree-zig-zag number $z$. 
Let $H$ be a digraph, $\newslicealphabet$ be a slice alphabet and 

$$\unitdecompositions(\newslicealphabet,H,z)=\{\boldT\in \lang(\newslicealphabet) |\composedT \simeq H, \treezigzagnumber(\boldT)\leq z\}.$$

We say that a tree slice language $\lang$ over $\newslicealphabet$ is {\em $z$-saturated} with respect to $\newslicealphabet$,
if for every digraph $H$, the fact that $[H]\in \lang_{\graph}$ implies that $\unitdecompositions(\newslicealphabet,H,z) \subseteq \lang$. In other words $\lang$ is $z$-saturated if whenever a canonical form $[H]$ belongs to the graph language $\lang_{\graph}$, all unit decompositions 
of tree-zig-zag number $z$ of $H$ belong to the slice language $\lang$. If the alphabet $\newslicealphabet$ 
is clear from the context we may say simply that $\lang$ is $z$-saturated, instead of saying that 
$\lang$ is $z$-saturated with respect to $\newslicealphabet$.
A slice tree-automaton $\treeAutomaton$ is $z$-saturated if $\lang(\treeAutomaton)$ is $z$-saturated. 
Proposition \ref{proposition:Intersection} below justifies our interest in the concept of $z$-saturation.

\begin{proposition}
\label{proposition:Intersection}
Let $\lang$ and $\lang'$ be tree slice languages over  $\newslicealphabet$ such that 
$\lang$ has tree-zig-zag $\mbox{number $z$}$ and such that $\lang'$ is $z$-saturated with respect to $\newslicealphabet$. 
Then $(\lang\cap \lang')_{\graph} = \lang_{\graph}\cap \lang'_{\graph}$. 
\end{proposition}
\begin{proof} 
The inclusion $(\lang\cap \lang')_{\graph} \subseteq \lang_{\graph}\cap \lang_{\graph}'$ holds for any two slice 
languages $\lang$ and $\lang'$ irrespectively of whether they are saturated or not. To see this, let $H$ be a digraph and 
let  $[H] \in (\lang \cap \lang')_{\graph}$. Then $H$ has a unit decomposition $\boldT$ 
in $\lang\cap \lang'$. Since $\boldT\in \lang$, $[H]\in \lang_{\graph}$ and, since $\boldT\in \lang'$, $[H]\in \lang_{\graph}'$.
Thus $(\lang \cap \lang')_{\graph} \subseteq \lang_{\graph}\cap \lang'_{\graph}$. Now we prove that if $\lang$ has 
tree-zig-zag number $z$ and $\lang'$
is $z$-saturated, the converse inclusion also holds. Let $H$ be a digraph and let $[H]\in \lang_{\graph}\cap \lang'_{\graph}$. Since 
$\lang$ has tree-zig-zag number $z$, $H$ has a unit decomposition $\boldT$ of tree-zig-zag number $z$ in $\lang$. Since $\lang'$ is $z$-saturated 
with respect to $\newslicealphabet$, each unit-decomposition of $H$ 
over $\newslicealphabet$ of tree-zig-zag number $z$ is in $\lang'$, and in particular $\boldT\in \lang'$. Therefore 
$\boldT\in \lang\cap \lang'$ and $[H]\in (\lang\cap \lang')_{\graph}$. 
\end{proof}

In other words, whenever $\lang$ has tree-zig-zag number $z$ and $\lang'$ is $z$-saturated, the intersection $\lang_{\graph}\cap \lang_{\graph}'$ 
of their graph languages is precisely the graph language of the intersection $\lang\cap \lang'$.
It is worth noting that Proposition \ref{proposition:Intersection} would not be true if none of the slice languages $\lang$ and $\lang'$ were saturated. 
For instance if $\lang=\{\boldT\}$ and $\lang'=\{\boldT'\}$ for two distinct unit decomposition $\boldT$ and $\boldT'$ of a digraph $H$ 
then $\lang_{\graph}=\lang_{\graph}' = \{ [H] \}$ but $\lang\cap \lang' = \emptyset$!

Proposition \ref{proposition:OliveTreeDecompositionUnitDecomposition} below says that any olive-tree decomposition 
$\olivetreedecomposition$ of a digraph $G$ can be efficiently converted into a unit decomposition $\boldT$ of $G$
that is compatible with $\olivetreedecomposition$. Note that there may be several unit decompositions of $G$ compatible 
with $\olivetreedecomposition$. In the proof of Proposition \ref{proposition:OliveTreeDecompositionUnitDecomposition} 
we provide an algorithm for computing one of these unit decompositions. 

\begin{proposition}
\label{proposition:OliveTreeDecompositionUnitDecomposition}
Let $\olivetreedecomposition$ be an olive-tree decomposition of a digraph $G=(V,E,\vertexlabeling,\edgelabeling)$ of width $q=\width(\olivetreedecomposition)$. 
Then one can construct in time $O(|\olivetreedecomposition|\cdot |E|)$ a normalized unit decomposition $\boldT$ over 
$\newslicealphabet(q,\vertexlabel,\edgelabel)$ compatible with $\olivetreedecomposition$.
\end{proposition}
\begin{proof}
Let $\olivetreedecomposition = (N,F,\mathfrak{m})$ be an olive-tree decomposition of $G=(V,E,\vertexlabeling,\edgelabeling)$. First we tag 
each edge $e\in G$ with a number $\tau(e)\in \{1,...,|E|\}$ in such a way that no two edges 
are tagged with the same number. We will construct a non-normalized unit decomposition $\boldT'$ over $\newslicealphabet(q,|E|,\vertexlabel,\edgelabel)$
such that $\composedTprime\simeq G$. A normalized unit decomposition $\boldT$ over $\newslicealphabet(q,\vertexlabel,\edgelabel)$ such that 
$\composedT\simeq G$ can be obtained from $\boldT'$ by an application of the normalizing projection 
$\normalizingprojection_{q,|E|}:\newslicealphabet(q,|E|,\vertexlabel,\edgelabel)\rightarrow \newslicealphabet(q,\vertexlabel,\edgelabel)$. 
To construct $\boldT'$ it is enough to specify the slice $\boldT'[p]$ for each position $p\in \positions(\boldT') = N$.
Instead of specifying each such slice $\boldT'[p]$ separately we will proceed in a more intuitive way. Namely, we will first 
define which unit slices of $\boldT'$ have a center vertex, and subsequently, for each edge $e$ in $G$ and each 
$p\in \positions(\boldT')$ we will specify which sliced part of $e$ (if any) belongs to $\boldT'[p]$. The first part 
is easy. A slice $\boldT'[p]$ has a center vertex if and only if some vertex of $G$ is mapped by $\mathfrak{m}$ to 
the position $p\in N$ in the olive-tree decomposition $\olivetreedecomposition$. For simplicity, let $v_p$ be the 
vertex of $G$ for which $\mathfrak{m}(v_p) = p$. We label the center vertex of $\boldT'[p]$ with the same label 
as the vertex $v_p$ in $G$. For each edge $e\in E$  with source $e^s = v_p$ and target $e^t=v_{p'}$
we create a sliced edge sequence $K\equiv(p_1,a_{1},e_1,b_{1})(p_2,a_{2},e_2,b_{2})...(p_n,a_{n},e_{n},b_{n})$ where $p_1=p$, $p_2=p'$,
$p_1p_2...p_n$ is the unique minimum path from $p$ to $p'$ in the tree $(N,F)$. For each $i\in \{1,...,n\}$,
the vertices $a_{i}$ and $b_{i}$ and the edge $e_{i}$ belong to the slice $\boldT'[p_i]$. The vertex $a_{1}$ is the center vertex 
of $\boldT'[p_1]$ and $b_{n}$ is the center vertex of $\boldT'[p_n]$. For each $i\in \{1,...,n-1\}$, if $p_{i+1} = p_ij$ then 
$b_{i}$ is the vertex $[j,\tau(e)]$ at the $j$-th in-frontier of $\boldT'[p_i]$ and $a_{{i+1}}$ is the vertex $[0,\tau(e)]$
at the out-frontier of $\boldT'[p_i]$. On the other hand, if $p_{i}=p_{i+1}j$ then $b_{i} = [0,\tau(e)]$ and $a_{{i+1}}=[j,\tau(e)]$.
Finally we label each edge $e_i$ of the sliced edge sequence $K$ with the same label as the edge $e$ in $G$. 
One can readily check that the sequence $K$ defined in this way is indeed a sliced edge sequence, and therefore that 
$\composedTprime = G$. As a final step, we obtain the unit decomposition $\boldT$ by an application of 
the normalizing projection $\normalizingprojection_{q,|E|}$ to $\boldT'$. In other words, $\boldT = \normalizingprojection_{q,|E|}(\boldT')$. 
\end{proof}

\subsection{Counting Subgraphs via $z$-Saturated Slice Languages}
\label{subsection:CountingSubgraphsViaZSaturatedSliceLanguages}

In this Subsection we will introduce the main technical tool of this paper. We will
show that given a $z$-saturated tree-automaton $\treeAutomaton$ representing digraphs that are the union of 
$k$ paths, and a unit decomposition $\boldT$ of tree-zig-zag number $z$, one can count in polynomial time the number of subgraphs of $\composedT$ 
that are isomorphic to some digraph in $\lang_{\graph}(\treeAutomaton)$. The proof will proceed in two steps. 
First, we will show that from a normalized unit decomposition $\boldT$ one can construct a (non-normalized) deterministic slice tree-automaton $\treeAutomaton(\boldT,k\cdot z)$  
whose slice language $\lang(\treeAutomaton(\boldT,k\cdot z))$ consists of all sub-decompositions of $\boldT$ of width at most $k\cdot z$. 
Each such sub-decomposition of $\boldT$ unequivocally identifies a subgraph of $\composedT$.
As a partial converse, each subgraph of $\composedT$ that is the union of $k$ directed paths has a representative unit decomposition in 
$\lang(\treeAutomaton(\boldT,k\cdot z))$. Note that $\lang(\treeAutomaton(\boldT,k\cdot z))$ still 
may contain unit decompositions of digraphs that are not the union of $k$ directed paths. However these undesired unit 
decompositions are irrelevant, since they will be eliminated in the next step. In our second step, we will 
show that the intersection $\treeAutomaton\cap \treeAutomaton(\boldT,k\cdot z)$ 
is a deterministic slice tree-automaton whose graph language consists precisely of the subgraphs of 
$\composedT$ that are isomorphic to some digraph in $\lang_{\graph}(\treeAutomaton)$. Note that 
$\treeAutomaton\cap \treeAutomaton(\boldT,k\cdot z)$ accepts a finite number of terms, and that the depth 
of each such accepted term is equal to the depth of $\boldT$.
At this point, the problem of counting subgraphs of $\composedT$  that are isomorphic to digraphs in $\lang_{\graph}(\treeAutomaton)$ boils down 
to counting the number of terms accepted by $\treeAutomaton \cap \treeAutomaton(\boldT,k\cdot z)$ in depth $\mathit{depth}(\boldT)$. We can count 
these terms in polynomial time using Lemma \ref{lemma:CountingTrees}. 

Lemma \ref{lemma:SubgraphsC} below says that given any unit decomposition $\boldT$ of width $q$, and any $c\leq q$, one can 
construct a slice tree-automaton whose slice language consists of all sub-decompositions of $\boldT$ of width at most 
$c$.
 
\begin{lemma}
\label{lemma:SubgraphsC}
Let $\boldT$ be a normalized unit decomposition in $\lang(\newslicealphabet(q,\vertexlabel,\edgelabel))$ and let 
$c\leq q$. Then one may construct in time $|\boldT|\cdot q^{O(c)}$ a slice tree-automaton $\treeAutomaton(\boldT,c)$ over 
$\newslicealphabet(c,q,\vertexlabel,\edgelabel)$ with $|\boldT|\cdot q^{O(c)}$ states satisfying 
the following properties.
\begin{enumerate}
	\item \label{SubgraphsC-item1} $\treeAutomaton(\boldT,c)$ is deterministic. 
	\item \label{SubgraphsC-item2} $\lang(\treeAutomaton(\boldT,c)) = \{\boldT'\in 
		\lang(\newslicealphabet(c,q,\vertexlabel,\edgelabel)) \;|\;\mbox{$\boldT'$ is a sub-decomposition of $\boldT$}\}$
\end{enumerate}
\end{lemma}
\begin{proof}
Let $\boldT$ be a unit decomposition in $\lang(\newslicealphabet(c,\vertexlabel,\edgelabel))$. We will construct a slice 
tree-automaton $\treeAutomaton = \treeAutomaton(\boldT,c) = (Q,\newslicealphabet,Q_F,\Delta)$ over $\newslicealphabet = \newslicealphabet(c,q,\vertexlabel,\edgelabel)$ whose slice language consists of all sub-decompositions of $\boldT$ of width at most $c$. The set of states $Q$ has one state $\astate_{p,\identitySlice}$ for each position 
$p\in \positions(\boldT)$ and  each identity slice $\identitySlice$ in $\newslicealphabet_1(c,q,\vertexlabel,\edgelabel)$. We note that 
since the empty slice of arity one, $\emptyslice_1$, is also an identity slice,  the state $\astate_{p,\emptyslice_1}$
belongs to $Q$ for each $p\in \positions(\boldT)$. 
The set of final states is the singleton $Q_F = \{\astate_{\lambda,\emptyslice_1}\}$. 
Now we will construct the transition relation 
$\Delta = \Delta_0 \cup \Delta_1 \cup \Delta_2$. First we recall that for each position $p\in \positions(\boldT)$, if 
$\boldT[p]$ is a slice of arity $r$ then any sub-slice $\boldS$ of $\boldT[p]$ has also arity $r$. 
Recall that if $\boldS$ is a slice, then $\identitySlice(\boldS)$ denotes the unique identity slice such that $\boldS$ can be 
glued to $\identitySlice(\boldS)$. 
For each $r\in \{0,1,2\}$, and each position $p\in \positions(\boldT)$ such that $\boldT[p]$ has arity $r$, the relation $\Delta_r$ has 
one transition $(\astate_{p1,\identitySlice_1},..., \astate_{pr,\identitySlice_r}, \boldS, \astate_{p,\identitySlice(\boldS)})$ for 
each sub-slice $\boldS$ of $\boldT[p]$ satisfying the following conditions:

\begin{enumerate}[(i)]
	\item \label{condition:one} $\boldS$ has width at most $c$, i.e., $\boldS\in \newslicealphabet(c,q,\vertexlabel,\edgelabel)$, 
	\item \label{condition:two} for each $j\in \{1,...,r\}$, $\identitySlice_j$ is the unique identity slice that can be glued to $\boldS$ at frontier $j$. 
\end{enumerate}

To see that $\treeAutomaton$ is deterministic, note that 
for each slice $\boldS$ there is a unique identity slice $\identitySlice$ such that $\boldS$ can be glued to $\identitySlice$. 
Therefore, for each tuple $(\astate_{p1,\identitySlice_1},\astate_{p2,\identitySlice_2}, ..., \astate_{pj,\identitySlice_r}, \boldS)$ there is a unique state 
$\astate_{p,\identitySlice}$ such that $(\astate_{p1,\identitySlice_1},\astate_{p2,\identitySlice_2}, ..., \astate_{pr,\identitySlice_r}, \boldS,\astate_{p,\identitySlice})$ belongs 
to $\Delta_r$. This also implies that each term $\boldT'$ accepted by $\treeAutomaton$ 
is a unit pre-decomposition, i.e., each two consecutive positions of $\boldT'$ can be glued. Additionally, the fact that $\astate_{\lambda,\emptyslice_1}$ is the unique accepting state of $\treeAutomaton$ 
implies that the slice at the root of $\boldT'$ has empty-frontier, since this slice must be glueable to $\emptyslice_1$. 
Therefore each such term $\boldT'$ is a unit decomposition. 
It remains to show that a unit decomposition $\boldT'$ is accepted by $\treeAutomaton$ if and only if $\boldT'$ is a sub-decomposition of $\boldT$ of 
width at most $c$. 

\begin{enumerate}
	\item (if direction) Let $\boldT'$ be a sub-decomposition of $\boldT$ of width at most $c$. We claim that for each position 
	$p\in \positions(\boldT') = \positions(\boldT)$ 
	the subterm $\boldT'|_p$ of $\boldT'$ rooted at $p$ reaches the state $\astate_{p,\identitySlice(\boldT'[p])}$. This claim implies in particular 
	that the whole term $\boldT' = \boldT'|_{\lambda}$ reaches the unique accepting state $\astate_{\lambda,\emptyslice_1}$. 
	The proof is by induction on the height of the position $p$. 
	In the base case, $p$ is a leaf of the set $\positions(\boldT)$.
	In this case, the slice $\boldT[p]$ has arity zero, and thus the sub-slice $\boldT'[p]$ has also arity zero. By the construction 
	of the transition relation $\Delta_0$ given above, the transition $(\boldT'[p], \astate_{p,\identitySlice(\boldT'[p])})$ belongs to 
	$\Delta_0$ and thus $\boldT'|_p$ reaches the state $\astate_{p,\identitySlice(\boldT'[p])}$. 
	Now assume by induction that the claim is valid for every position $p'$ of height $h$. 
	Let  $p$ be a position in $\positions(\boldT)$  of height $h+1$ with children are  $p1,...,pr$ 
	for some $r\in \{1,2\}$. By the induction hypothesis, for each $i\in \{1,...,r\}$ the term $\boldT'|_{pi}$ reaches the 
	state $\astate_{pi,\identitySlice(\boldT'[pi])}$. By the definition of the transition relation $\Delta_r$, we have that the transition 
	$(\astate_{p1,\identitySlice(\boldT'[p1])}, ...,\astate_{pr,\identitySlice(\boldT'[pr])},\boldT'[p], \astate_{p,\identitySlice(\boldT'[p])})$ belongs to $\Delta_r$,
	 and thus $\boldT'|_p$ reaches the state $\astate_{p,\identitySlice(\boldT'[p])}$. This proves our claim.

	\item (only if direction) For the converse, let $\boldT'$ be a unit decomposition accepted by $\treeAutomaton$.
	We will prove that $\boldT'$ is a sub-decomposition of $\boldT$ by showing that 
	$\positions(\boldT) = \positions(\boldT')$ and that for each position $p\in \positions(\boldT')$, $\boldT'[p]$ is a sub-slice of 
	$\boldT[p]$. 
	We claim that for each $p\in \positions(\boldT')$, the subterm $\boldT'|_p$ of $\boldT'$ rooted at $p$ reaches the state 
	$\astate_{p,\identitySlice(\boldT'[p])}$. By the construction of the transition relation $\Delta$ this claim implies both 
	that $\boldT'[p]$ is a sub-slice of $\boldT[p]$ for each $p\in \positions(\boldT')$ and that $\positions(\boldT')=\positions(\boldT)$, as 
	desired.
	The proof of this claim is by induction on the depth of $p$. In the base case, $p=\lambda$. In this 
	case, $\boldT'|_{\lambda} = \boldT'$ reaches the unique accepting state $\astate_{\lambda,\emptyslice_1} = \astate_{\lambda,\identitySlice(\boldT'[\lambda])}$. 
	Now assume that for every position $p$ of depth at most $d$, the term $\boldT'|_{p}$ reaches the state $\astate_{p,\identitySlice(\boldT'[p])}$.
	We will show that the claim holds for every position $p$ of depth  $d+1$.  
	Let $p\in \positions(\boldT')$ be a position of depth $d$. By the induction hypothesis, $\boldT'|_{p}$ reaches the state $\astate_{p,\identitySlice(\boldT'[p])}$. 
	Let $\boldT'[p]$ have arity $r$ for some $r\in \{1,2\}$. Since $\boldT'|_{p}$ reaches $\astate_{p,\identitySlice(\boldT'[p])}$, there exist states $\mathfrak{q}_{1},...,\mathfrak{q}_r$ such that 
	the transition $(\mathfrak{q}_1,...,\mathfrak{q}_r,\boldT'[p],\astate_{p,\identitySlice(\boldT'[p])})$ belongs to $\Delta$ and $\boldT'|_{pj}$ reaches 
	$\mathfrak{q}_j$ for each $j\in \{1,...,r\}$.
	By the definition of $\Delta$, for each $j\in \{1,...,r\}$, $\mathfrak{q}_{j} = \astate_{pj,\identitySlice_j}$ where $\identitySlice_{j}$ is the 
	unique identity slice that can be glued to $\boldT'[p]$ at frontier $j$. 
	Since $\boldT'$ is a unit decomposition, $\boldT'[pj]$ can be glued to $\boldT'[p]$ at frontier $j$. Therefore $\identitySlice_j = \identitySlice(\boldT'[pj])$.  
	Thus $\mathfrak{q}_j = \astate_{pj,\identitySlice(\boldT'[pj])}$  and  $\boldT'|_{pj}$ reaches $\astate_{pj,\identitySlice(\boldT'[pj])}$.
	This proves our inductive step. 
\end{enumerate} 
\end{proof}

Proposition \ref{proposition:UnionZigZag} below establishes a relation between the minimum number of paths necessary to 
cover all edges and vertices of a digraph $H$, and the width of a unit decomposition of $H$ of tree-zig-zag number 
$z$. Intuitively, if $\boldT$ is a unit decomposition of tree-zig-zag number $z$ of a digraph $H$, then each directed simple path $\path$ in $H$
crosses each frontier of a slice in $\boldT$ at most $z$ times. Therefore, if $H$ is the union of $k$ directed simple paths $\path_1,...,\path_k$, 
then all such paths together cross each frontier of each slice of $\boldT$ at most $k\cdot z$ times.

\begin{proposition}
\label{proposition:UnionZigZag}
Let $H$ be a digraph that is the union of $k$-paths and $\newslicealphabet$ be a slice alphabet. 
Then any unit decomposition $\boldT\in \lang(\newslicealphabet)$ of $H$ of tree-zig-zag number $z$ has width at most $k\cdot z$. 
\end{proposition}
\begin{proof}
Let $H = (V,E)$ be a digraph that is the union of $k$ directed paths  $\path_1, ..., \path_k$ where for each $i\in \{1,...,k\}$,  $\path_i = (V_{\path_i},E_{\path_i})$. 
Let $\boldT\in \lang(\newslicealphabet)$ be a unit decomposition of tree-zig-zag number $z$ of a digraph $H$. Then $\boldT$ is compatible 
with an olive-tree decomposition $\olivetreedecomposition = (N,F,\mathfrak{m})$ of tree-zig-zag number $z$. Additionally, the width 
of $\boldT$ is equal to the width of $\olivetreedecomposition$. Since $\olivetreedecomposition$ has 
tree-zig-zag number $z$, for each position $p\in N$ and each $i\in \{1,...,k\}$
we have that $$|E(V(p,\olivetreedecomposition), V\backslash V(p,\olivetreedecomposition))\;\cap \; E_{\path_i}| \leq z.$$ 
This implies that $$|E(V(p,\olivetreedecomposition), V\backslash V(p,\olivetreedecomposition))\; \cap\; \bigcup_{i=1}^k E_{\path_i}\;| \leq k\cdot z.$$ 
But since $E = \cup_{i=1}^k E_{\path_i}$, we have that  $|E(V(p,\olivetreedecomposition),V\backslash V(p,\olivetreedecomposition))| \leq k\cdot z$. 
Thus $\olivetreedecomposition$ has width at most $k\cdot z$, implying in this way that $\boldT$ has also width 
at most $k\cdot z$. 
\end{proof}

Next we state the main lemma of this section. Intuitively Lemma \ref{lemma:IntersectionSubdecompositions} below says that if $\boldT$ is a unit decomposition 
of tree-zig-zag number $z$ of a digraph $G$, and if $\treeAutomaton$ is a $z$-saturated tree-automaton representing only digraphs that are 
the union of $k$ directed paths, then the slice language $\lang(\treeAutomaton(\boldT,k\cdot z)\cap \treeAutomaton)$ has precisely one unit decomposition 
for each subgraph of $\composedT$ that is isomorphic to a digraph in $\lang_{\graph}(\treeAutomaton)$. In this sense, the problem of counting 
the number of subgraphs of $\composedT$ that are isomorphic to a digraph in $\lang_{\graph}(\treeAutomaton)$ boils down to counting the number 
of unit-decompositions in $\lang(\treeAutomaton(\boldT,k\cdot z) \cap \treeAutomaton)$. This counting step will be detailed in Theorem 
\ref{theorem:CountingSubgraphs}. 

\begin{lemma}
\label{lemma:IntersectionSubdecompositions}
Let $\boldT$ be a unit decomposition of tree-zig-zag number $z$ over $\newslicealphabet(q,\vertexlabel,\edgelabel)$. Let $\treeAutomaton$ be a deterministic 
$z$-saturated slice automaton over $\newslicealphabet(k\cdot z,q,\vertexlabel,\edgelabel)$ such that each digraph in $\lang_{\graph}(\treeAutomaton)$
is the union of $k$ directed paths. 
\begin{enumerate}
	\item \label{lemma:IntersectionSubdecompositions-One} The tree-automaton $\treeAutomaton(\boldT,k\cdot z) \cap \treeAutomaton$ is deterministic and all its 
		accepted unit decompositions have depth at most $\mathit{depth}(\boldT)$. 
	\item \label{lemma:IntersectionSubdecompositions-Two} $H$ is a subgraph of $\composedT$  such that $[H] \in \lang_{\graph}(\treeAutomaton)$ if and only if 
		there exists a unit decomposition $\boldT' \in \lang(\treeAutomaton(\boldT,k\cdot z) \cap \treeAutomaton)$ such that $\composedTprime=H$. 
\end{enumerate}
\end{lemma}
\begin{proof}
Item \ref{lemma:IntersectionSubdecompositions-One} is straightforward. The automaton $\treeAutomaton(\boldT,k\cdot z) \cap \treeAutomaton$ is 
deterministic because both $\treeAutomaton(\boldT,k\cdot z)$ and $\treeAutomaton$ are deterministic 
(Lemma \ref{lemma:PropertiesOfTreeAutomata}.\ref{lemma:PropertiesOfTreeAutomata:UnionIntersection}). Since by construction the construction 
of $\treeAutomaton(\boldT,k\cdot z)$, all unit decompositions
accepted by $\treeAutomaton(\boldT,k\cdot z)$ have depth $\mathit{depth}(\boldT)$, we have that all unit decompositions accepted by 
$\treeAutomaton(\boldT,k\cdot z)\cap \treeAutomaton$ also have depth $\mathit{depth}(\boldT)$. 
Now we proceed to prove item \ref{lemma:IntersectionSubdecompositions-Two}. 
\begin{enumerate}[(a)]
	\item (if direction) Let $\boldT'$ be a unit decomposition in $\lang(\treeAutomaton(\boldT,k\cdot z)\cap \treeAutomaton)$ such 
		that $\composedTprime = H$. Since $\boldT'\in \lang(\treeAutomaton(\boldT,k\cdot z))$, by Lemma \ref{lemma:SubgraphsC},
		$\boldT'$ is a  sub-decomposition of $\boldT$. Therefore $H$ is a subgraph of $\composedT$. Additionally,
		since $\boldT'\in \lang(\treeAutomaton)$, $[H]\in \lang_{\graph}(\treeAutomaton)$. 
	\item (only if direction) Let $H$ be a subgraph of $\composedT$ such that $[H]\in \lang_{\graph}(\treeAutomaton)$. 
		Since $H$ is a subgraph of $\composedT$, there is a sub-decomposition $\boldT'$ of $\boldT$ such that $\composedTprime = H$.
		We will show that $\boldT'$ belongs to $\lang(\treeAutomaton(\boldT,k\cdot z)\cap \treeAutomaton)$.  
		Since $\boldT$ has tree-zig-zag number $z$, $\boldT'$ has tree-zig-zag number at most $z$. Now, since 
		$H\in \lang_{\graph}(\treeAutomaton)$, we have that  $H$ is  the union of $k$ directed paths. By Proposition 
		\ref{proposition:UnionZigZag}, each unit decomposition of $H$ of tree-zig-zag number at most $z$ has width at most 
		$k\cdot z$. Thus $\boldT'$ has width at most $k\cdot z$. Finally, since $\treeAutomaton$ is $z$-saturated with 
		respect to $\newslicealphabet(k\cdot z,q,\vertexlabel,\edgelabel)$ we have that $\boldT'$ belongs 
		to $\lang(\treeAutomaton)$. Thus $\boldT'\in \lang(\treeAutomaton(\boldT,k\cdot z) \cap \treeAutomaton)$. 
\end{enumerate}
\vspace{-3pt}
\end{proof}

The next Theorem is the main application for Lemma \ref{lemma:IntersectionSubdecompositions}. Intuitively, given a unit decomposition $\boldT$
of tree-zig-zag number $z$, and a $z$-saturated tree-automaton $\treeAutomaton$ representing only digraphs that are the union of $k$ directed paths, where 
$z$ and $k$ are constants, one can count in polynomial time the number 
of subgraphs of $\composedT$ that are isomorphic to some digraph in $\lang(\treeAutomaton)$. The idea is that Lemma \ref{lemma:IntersectionSubdecompositions} 
allow us to reduce this counting problem to the problem of counting the number of accepted unit decompositions in the 
tree-automaton $\treeAutomaton(\boldT,k\cdot z) \cap \treeAutomaton$.

\begin{theorem}[Slice Theoretic Metatheorem]
\label{theorem:CountingSubgraphs}
Let $\boldT$ be a unit decomposition over $\newslicealphabet(q,\vertexlabel,\edgelabel)$ and let $\treeAutomaton$ be a 
deterministic $z$-saturated slice tree-automaton over $\newslicealphabet(k\cdot z,q,\vertexlabel,\edgelabel)$ satisfying the 
property that each digraph in $\lang_{\graph}(\treeAutomaton)$ is the union of $k$ directed paths. Then one can count in 
time $|\boldT|^{O(1)}\cdot q^{O(k\cdot z)}\cdot |\treeAutomaton|^{O(1)}$ the number of subgraphs of $\composedT$ that are isomorphic 
to a digraph in $\lang_{\graph}(\treeAutomaton)$. 
\end{theorem}
\begin{proof}
First, we construct in time $|\boldT|\cdot q^{O(k\cdot z)}$ the tree-automaton $\treeAutomaton(\boldT,k\cdot z)$ of 
Lemma \ref{lemma:SubgraphsC} whose slice language consists of the set of all sub-decompositions of $\boldT$ of 
width at most $k\cdot z$. Subsequently, using Lemma \ref{lemma:PropertiesOfTreeAutomata}.\ref{lemma:PropertiesOfTreeAutomata:UnionIntersection}, 
we construct the tree-automaton $\treeAutomaton(\boldT,k\cdot z) \cap \treeAutomaton$ in time $|\boldT| \cdot q^{O(k\cdot z)} \cdot |\treeAutomaton|$. 
By Lemma \ref{lemma:IntersectionSubdecompositions} a subgraph of $\composedT$ is isomorphic to some digraph in 
$\lang_{\graph}(\treeAutomaton)$ if and only if there exists a sub-decomposition $\boldT'$ of $\boldT$ in $\lang(\treeAutomaton(\boldT,k\cdot z)\cap \treeAutomaton)$
for which $\composedTprime = H$. Therefore, counting the subgraphs of $\composedT$ that are isomorphic to some 
digraph in $\lang_{\graph}(\treeAutomaton)$ amounts to counting the number of unit decompositions of depth $\mathit{depth}(\boldT)$ in 
$\lang(\treeAutomaton(\boldT,k\cdot z)\cap \treeAutomaton)$. 
In other words, this problem is equivalent to the problem of counting the number of terms accepted by $\treeAutomaton\cap \treeAutomaton(\boldT,k\cdot z)$
in depth $\mathit{depth}(\boldT)$. Since $\mathit{depth}(\boldT)\leq |\boldT|$, by Lemma \ref{lemma:CountingTrees}, this counting process can be done in time 
$|\boldT|^{O(1)} \cdot q^{O(k\cdot z)} \cdot |\treeAutomaton|^{O(1)}$. 
\end{proof}

Next, in Section \ref{section:MSOandTreeSliceLanguages} we will show that for each \msotwo sentence  
$\varphi$ and each $k,z\in \N$, one can construct a $z$-saturated tree-automaton $\treeAutomaton(\varphi,k,z)$ 
representing the set of all digraphs that at the same time are the union of $k$ directed paths and satisfy $\varphi$. Subsequently, 
in Section \ref{section:ProofOfMainTheorem} we will show how to restrict $\treeAutomaton(\varphi,k,z)$ into a 
$z$-saturated slice tree-automaton $\treeAutomaton(\varphi,k,z,l,\alpha)$ representing only the digraphs in $\lang_{\graph}(\treeAutomaton(\varphi,k,z))$
which have a prescribed number $l$ of vertices and a prescribed weight $a\in \Omega$. The proof of Theorem \ref{theorem:MainTheoremDirectedTreewidth}
will follow by plugging the tree-automaton $\treeAutomaton(\varphi,k,z,l,\alpha)$ into Theorem \ref{theorem:CountingSubgraphs}.

\section{\msotwo Logic and Tree Slice Languages}
\label{section:MSOandTreeSliceLanguages}

Defining interesting families of digraphs via $z$-saturated tree-automata is a difficult task. 
The difficulty relies on the fact that to construct a $z$-saturated tree-automaton we have to make sure 
that for each digraph $H$ in the graph language $\lang_{\graph}(\treeAutomaton)$, 
all unit decompositions of $H$ with tree-zig-zag number at most $z$ are in the 
slice language $\lang(\treeAutomaton)$. In this section we will introduce a suitable way of circumventing
this difficulty by using the monadic second order logic of graphs with edge set quantifications, or \msotwo logic 
for short. This logic, which extends first order logic by incorporating quantification over sets of vertices and over sets of edges, 
is able to express a large variety of natural graph properties \cite{CourcelleEngelfriet2012}. 
We will show that for any \msotwo sentence $\varphi$ and any $k,z\in \N$ we can automatically construct a $z$-saturated slice tree-automaton $\treeAutomaton(\varphi,k,z)$ whose graph language consists of 
all digraphs that at the same time satisfy $\varphi$ and are the union of $k$ directed paths 
(Theorem \ref{theorem:MonadicSliceTreeAutomataZSaturated}). 

Let $\vertexlabel$ be a set of vertex labels and $\edgelabel$ be a set of edge labels.
A $(\vertexlabel,\edgelabel)$-labeled digraph is a relational structure $G=(V,E,s,t,\vertexlabeling,\edgelabeling)$  
comprising a set of vertices $V$, a set of edges $E$, source and target relations $s,t \subseteq E\times V$,
a vertex-labeling relation $\vertexlabeling \subseteq V\times \vertexlabel$ and an edge-labeling relation $\edgelabeling \subseteq E\times \edgelabel$.
The language of \msotwo logic for $(\vertexlabel,\edgelabel)$-labeled digraphs includes the connectives $\vee,\wedge,\neg$, variables for vertices, 
edges, sets of vertices and sets of edges, the quantifier $\exists$ that can be applied to these variables, and the following predicates:

\begin{enumerate}
	\item $x\in X$ where $x$ is a vertex variable and $X$ a vertex set variable, 
	\item $y\in Y$ where $y$ is an edge variable and $Y$ an edge set variable,
	\item Equality, $=$, of variables representing vertices, edges, sets of vertices and sets of edges.
	\item $s(y,x)$ where $y$ is an edge variable, $x$ a vertex variable, and the interpretation is that $x$ is the source of $y$.
	\item $t(y,x)$ where $y$ is an edge variable, $x$ a vertex variable, and the interpretation is that $x$ is the target  $y$.
	\item For each $a\in \vertexlabel$, a predicate $\vertexlabeling(x,a)$ where $x$ is a vertex variable, and the interpretation is 
		that $x$ is a vertex labeled with $a$.
	\item For each $b\in \edgelabel$, a predicate $\edgelabeling(x,b)$ where $x$ is an edge variable, and the interpretation is 
		that $x$ is an edge labeled with $b$.
\end{enumerate}

Let $\mathcal{X}$ be a set of free first order variables and second order variables. An interpretation of $\mathcal{X}$ 
in $G$ is a function $M:\mathcal{X}\rightarrow (V\cup E) \cup (2^{V}\cup 2^{E})$ that associates with 
each vertex variable $x\in \mathcal{X}$, a vertex in $V$, with each edge variable $y\in \mathcal{X}$, an edge in $E$, 
with each vertex set variable $X\in \mathcal{X}$ a set of vertices and with each edge set variable $Y\in \mathcal{X}$,
a set of edges.  The semantics of a formula $\varphi$ with free variables 
$\mathcal{X}$ being true under interpretation $M$ is the standard one. An \msotwo {\em sentence} is an \msotwo formula without free 
variables. 
If $G=(V,E,s,t,\vertexlabeling,\edgelabeling)$ is a $\mbox{$(\vertexlabel,\edgelabel)$-labeled}$ digraph and $\varphi$
is an \msotwo sentence then we write $G\models \varphi$ to indicate that $G$ satisfies $\varphi$.  
Let $\newslicealphabet(c,\vertexlabel,\edgelabel)= \newslicealphabet_0(c,\vertexlabel,\edgelabel) \cup \newslicealphabet_1(c,\vertexlabel,\edgelabel)
\cup \newslicealphabet_2(c,\vertexlabel,\edgelabel)$ be the ranked slice alphabet defined in Section \ref{subsection:SliceAlphabets} (with $r=2$).
Lemma \ref{lemma:MonadicSliceTreeAutomata} below, which will be proved in $\mbox{Section \ref{subsection:ProofLemmaMonadicSliceTreeAutomata},}$
establishes a connection between \msotwo logic and slice tree-automata.

\begin{lemma} 
\label{lemma:MonadicSliceTreeAutomata} For every \msotwo sentence $\varphi$ over $(\vertexlabel,\edgelabel)$-labeled digraphs 
and every $c\in \N$, one can construct a normalized deterministic slice tree-automaton $\treeAutomaton(\varphi,c)$ over
$\newslicealphabet(c,\vertexlabel,\edgelabel)$ generating the following tree slice language:
\begin{equation}
\lang(\treeAutomaton(\varphi,c)) = \{\boldT \in \lang(\newslicealphabet(c,\vertexlabel,\edgelabel))\;|\; \composedT \models \varphi\}.
\end{equation}
\end{lemma} 

In other words, Lemma \ref{lemma:MonadicSliceTreeAutomata} says that given an \msotwo sentence  $\varphi$ and a number $c\in N$ we 
can construct a slice tree-automaton whose tree slice language consists of all unit decompositions of width at most 
$c$ representing a digraph that satisfies $\varphi$. Another way of interpreting $\mbox{Lemma \ref{lemma:MonadicSliceTreeAutomata}}$
is as a "sliced" version of Courcelle's celebrated model checking theorem \cite{Courcelle1990MSO}. Recall that Courcelle's theorem 
states that graphs of constant {\em undirected} treewidth can be model checked in linear time against \msotwo sentences. 
In analogy with Courcelle's theorem, Lemma \ref{lemma:MonadicSliceTreeAutomata} says that digraphs admitting unit decompositions of constant width can be model checked in linear time against \msotwo sentences. 
In order to verify whether a digraph $G$ admitting a unit decomposition of width at most $c$ satisfies an \msotwo sentence  $\varphi$,
 all one needs to do is to find a normalized unit decomposition $\boldT$ of $G$ of width at most $c$,
and then check in linear time if $\boldT$ is accepted by $\treeAutomaton(\varphi,c)$. 

In this work however we will not be interested in model checking properties on digraphs admitting unit decompositions of 
constant width. Instead we will use Lemma \ref{lemma:MonadicSliceTreeAutomata} to construct $z$-saturated slice tree-automata
representing families of digraphs that are the union of $k$ directed paths and satisfy a given prescribed \msotwo property. These automata,
 which will be constructed in the proof of Theorem \ref{theorem:MonadicSliceTreeAutomataZSaturated} below, 
can be coupled to Theorem \ref{theorem:CountingSubgraphs} to provide a way of counting subgraphs satisfying interesting 
properties on digraphs of constant tree-zig-zag number, and hence, on digraphs of constant {\em directed} treewidth. 
At this point our approach differs substantially from Courcelle's theorem \cite{Courcelle1990MSO} as well as from 
the approaches in \cite{ArnborgLagergrenSeese1991,CourcelleMakowskyRotics2000} in the sense 
that, as mentioned in the introduction, digraphs of constant {\em directed} treewidth may have simultaneously unbounded {\em undirected} 
treewidth and unbounded clique-width. 

\begin{theorem} 
\label{theorem:MonadicSliceTreeAutomataZSaturated} For every \msotwo sentence $\varphi$ and every $k,z\in \N$, 
one can effectively construct a normalized deterministic $z$-saturated slice tree-automaton $\treeAutomaton(\varphi,k,z)$ over 
the slice alphabet $\newslicealphabet(k\cdot z,\vertexlabel,\edgelabel)$ representing the following graph language. 
\begin{equation}
\lang_{\graph}(\treeAutomaton(\varphi,k,z)) = \{[H]\;|\; H \models \varphi,\; H\mbox{ is the union of $k$ directed paths}\}. 
\end{equation}
\end{theorem} 
\begin{proof}
Let $\gamma(k)$ be the \msotwo sentence  that is true in a digraph $H$ if and only if $H$ is the union of $k$ directed paths.
Using Lemma \ref{lemma:MonadicSliceTreeAutomata} we construct a normalized deterministic tree-automaton 
$\treeAutomaton(k\cdot z, \varphi\wedge \gamma(k))$ generating the set of all unit decompositions $\boldT$ over 
$\newslicealphabet(k\cdot z,\vertexlabel,\edgelabel)$ for which the digraph $\composedTprime$  satisfies $\varphi\wedge \gamma(k)$. 
In other words if $[H]\in \lang_{\graph}(\treeAutomaton(k\cdot z,\varphi\wedge \gamma(k)))$
then $H$ satisfies $\varphi$ and is the union of $k$ directed paths. 
For the converse, suppose that $H$ is a digraph that is the union of $k$ directed paths and satisfies $\varphi$.
Then by Proposition \ref{proposition:UnionZigZag}, each unit decomposition $\boldT$ of $H$ of tree-zig-zag number at most $z$ has width 
at most $k\cdot z$. Therefore, $\boldT \in \lang(\treeAutomaton(\varphi,k,z))$. This implies not only that $[H]\in \lang_{\graph}(\treeAutomaton(\varphi,k,z))$
but also that $\lang(\treeAutomaton(\varphi,k,z))$ is $z$-saturated. 
\end{proof}

\subsection{Proof of Lemma \ref{lemma:MonadicSliceTreeAutomata}}
\label{subsection:ProofLemmaMonadicSliceTreeAutomata}

To prove Lemma \ref{lemma:MonadicSliceTreeAutomata} we need to translate each \msotwo sentence $\varphi$
expressing a property of $(\vertexlabel,\edgelabel)$-labeled digraphs into an \msotwo sentence $\psi$ expressing a property 
of unit decompositions over $\newslicealphabet(c,\vertexlabel,\edgelabel)$ in such a way that
for each unit decomposition $\boldT$ over $\newslicealphabet(c,\vertexlabel,\edgelabel)$, $\boldT$ satisfies $\psi$ if and 
only if the digraph $\composedT$ represented by $\boldT$ satisfy $\varphi$. 
With this goal in mind we need to define a new \msotwo vocabulary which is suitable for 
expressing properties of unit decompositions. 
The language of \msotwo logic for unit decompositions 
over $\newslicealphabet(c,\vertexlabel,\edgelabel)$ has the connectives $\vee,\wedge,\neg$, vertex variables, edge variables, 
vertex set variables and edge set variables, the quantifier $\exists$ that can be applied to these variables, 
and the following predicates:

\begin{enumerate}
	\item $x\in X$ where $x$ is a vertex variable and $X$ a vertex set variable, 
	\item $y\in Y$ where $y$ is an edge variable and $Y$ an edge set variable,
	\item Equality, $=$, of variables representing vertices, edges, sets of vertices and sets of edges.
	\item $\hat{s}(y,x)$ where $y$ is an edge variable, $x$ a vertex variable, and the interpretation is that for some position $p\in \positions(\boldT)$, 
		$x$ is a vertex of $\boldT[p]$, $y$ is an edge of $\boldT[p]$ and $x$ is the source of $y$.
	\item $\hat{t}(y,x)$ where $y$ is an edge variable, $x$ a vertex variable, and the interpretation is that for some position $p\in \positions(\boldT)$, 
		$x$ is a vertex of $\boldT[p]$, $y$ is an edge of $\boldT[p]$ and $x$ is the target of $y$.
	\item For each $a\in \vertexlabel$, a predicate $\hat{\vertexlabeling}(x,a)$ where $x$ is a vertex variable, and the interpretation is 
		that for some $p\in \positions(\boldT)$, $x$ is a vertex of $\boldT[p]$ labeled with $a$.
	\item For each $b\in \edgelabel$, a predicate $\hat{\edgelabeling}(y,b)$ where $y$ is an edge variable, and the interpretation is 
		that for some $p\in \positions(\boldT)$, $y$ is an edge of $\boldT[p]$ labeled with $b$.
	\item For each $j\in \{0,1,2\}$ and each $i\in \{1,...,c\}$, the predicate $F_{j,i}(x)$ where $x$ is a vertex variable and 
		the interpretation is that for some position $p\in \positions(\boldT)$, $x$ is the frontier vertex $[j,i]$ of the slice $\boldT[p]$. 
	\item The predicate  $C(x)$ where $x$ is a vertex variable and the interpretation is that for some position $p\in \positions(\boldT)$, $x$ is the unique
		center vertex of the slice $\boldT[p]$. 
	\item The predicate $\mathit{Neighbors}(x_1,x_2)$ where $x_1,x_2$ are vertex variables and the interpretation is that for some position $p\in \positions(\boldT)$
		and some $j\in \{1,2\}$, $x_1$ is a vertex of $\boldT[p]$ and $x_2$ a vertex of $\boldT[pj]$.
\end{enumerate}

Recall from Section \ref{subsection:DigraphsUnitDecompositions} that if $\boldT$ is a unit decomposition then the digraph 
$\composedT$ has an edge $e_K$ with source $e^{s}_K=v_{p}$ and target $e^t_K=v_{p'}$ if and only if there exists a sliced edge 
sequence 

$$K \equiv  (p_1,a_{1},e_{1},b_{1})(p_2,a_{2},e_{2},b_{2})...(p_n,a_{n},e_{n},b_{n})$$

from $p_1=p$ to $p_n=p'$. 
We note that each edge of each slice occurring in $\boldT$ belongs 
to a unique sliced edge sequence. In particular, each sliced edge sequence $K$ is unequivocally determined by its first edge $e_{1}$. 
 Using conditions \ref{item:edgesequence1}-\ref{item:edgesequence5} of Definition \ref{definition:SlicedEdgeSequence}, 
it is straightforward to write an \msotwo formula $\theta(u,y,v)$ in the vocabulary of unit decompositions with free vertex variables $u,v$ and 
free edge variable $y$, which is true in a unit decomposition $\boldT$ if and only if there exist positions $p$ and $p'$ in $\boldT$ 
such that $u$ is the center vertex of $\boldT[p]$, $y$ is an edge in $\boldT[p]$ with source $u$, $v$ is the center vertex of $\boldT[p']$, and there exists 
a sliced edge sequence from $p$ to $p'$ whose first edge is $y$. Using the formula $\theta(u,y,v)$ we can map formulas in the 
vocabulary of $(\vertexlabel,\edgelabel)$-labeled digraphs to formulas in the vocabulary of unit decompositions over $\newslicealphabet(c,\vertexlabel,\edgelabel)$, 
as done below in Proposition \ref{proposition:TranslationFormula}.

\begin{proposition}
\label{proposition:TranslationFormula}
Let $\varphi$ be an \msotwo formula in the vocabulary of $(\vertexlabel,\edgelabel)$-labelled graphs. There 
is a formula $\psi$ in the vocabulary of unit decompositions over $\newslicealphabet(c,\vertexlabel,\edgelabel)$ such 
that for each unit decomposition $\boldT$ over $\newslicealphabet(c,\vertexlabel,\edgelabel)$, 
$\boldT \models \psi$ if and only if $\composedT \models \varphi$. 
\end{proposition}
\begin{proof}
As mentioned above, using the predicates $C(x)$, $F_{j,i}(x)$, $\mathit{Neighbors}(x_1,x_2)$, 
$\hat{s}(y,x)$ and $\hat{t}(y,x)$ we can define a formula $\theta(u,y,v)$ that is true in a unit decomposition 
$\boldT$ if and only if there is a sliced edge sequence with first vertex $u$, first edge $y$ and 
last vertex $v$. 
The translation from $\varphi$ to $\psi$ proceeds as follows. We replace each occurrence of the predicate $\vertexlabeling(x,a)$ in $\varphi$ 
with the predicate $\hat{\vertexlabeling}(x,a)$, each occurrence of $\edgelabeling(x,a)$
with $\hat{\edgelabeling}(x,a)$, each occurrence of $s(y,x)$ with $(\exists v)\theta(x,y,v)$ where $v$ is a new variable 
not occurring in $\varphi$, and each occurrence of $t(y,x')$ with $(\exists u) \theta(u,y,x')$, where $u$ is a new variable 
not occurring in $\varphi$. Now it is straightforward to prove by induction of the structure of $\varphi$ that 
for each given unit decomposition $\boldT\in \newslicealphabet(c,\vertexlabel,\edgelabel)$, $\boldT\models \psi$ if 
and only if $\composedT\models \varphi$.  
\end{proof}

In the last step of the proof of Lemma \ref{lemma:MonadicSliceTreeAutomata} we will show that for each \msotwo sentence
$\psi$ in the vocabulary of unit decompositions over $\newslicealphabet(c,\vertexlabel,\edgelabel)$, it is possible to 
construct a slice tree-automaton $\treeAutomaton = \treeAutomaton(\psi,\newslicealphabet(c,\vertexlabel,\edgelabel))$ such that 
$\boldT \in \lang(\treeAutomaton)$ if and only if $\boldT$ satisfies $\psi$.

Let $\mathcal{X}$ be a set of variables, and $\boldS \in \newslicealphabet(c,\vertexlabel,\edgelabel)$ be a unit slice with $r$ vertices and $r'$ edges 
(including the frontier vertices).  We represent an interpretation of 
$\mathcal{X}$ in $\boldS$ as a $|\mathcal{X}|\times (r+r')$ boolean matrix $I$ whose rows are 
indexed by the variables in $\mathcal{X}$ and the columns are indexed by the 
vertices and edges of $\boldS$. Intuitively, if $x$ is a vertex (edge) variable 
and $u$ is a vertex (edge) in $\boldS$ then we set $I_{x,u}=1$ if and only if $u$ is 
assigned to $x$. On the other hand, if $X$ is a vertex (edge) set variable 
then $I_{X,u}=1$ if and only if $u$ belongs to the set of vertices (edges) assigned to $X$.
If $\boldT$ is a unit decomposition in $\lang(\newslicealphabet(c,\vertexlabel,\edgelabel))$
then an interpretation of $\mathcal{X}$ in $\boldT$ is a function $\mathcal{I}$ that 
associates with each position $p\in \positions(\boldT)$, an interpretation $\mathcal{I}(p)$ of $\mathcal{X}$ in the slice $\boldT[p]$. 
We define the $\mathcal{X}$-interpreted extension of $\newslicealphabet(c,\vertexlabel,\edgelabel)$ as the following set.

\begin{equation*} 
\interpretedAlphabet=\bigcup_{\boldS\in \newslicealphabet(c,\vertexlabel,\edgelabel)} \boldS^{\mathcal{X}}
\end{equation*}

where for each slice $\boldS\in \newslicealphabet(c,\vertexlabel,\edgelabel)$,

\begin{equation}
 \boldS^{\mathcal{X}} = \{ (\boldS,I)\;|\; I \mbox{ is an interpretation of } \mathcal{X} \mbox{ in } \boldS\}. 
\end{equation}

If $\boldT$ is a unit decomposition in $\lang(\newslicealphabet(c,\vertexlabel,\edgelabel))$ and $\mathcal{I}$ is an 
interpretation of $\mathcal{X}$ in $\boldT$ then we write $\boldT^{\mathcal{I}}$ to denote the term in 
$\lang(\interpretedAlphabet)$ in which $\boldT^{\mathcal{I}}[p] = (\boldT[p],\mathcal{I}(p))$ for each position $p\in \positions(\boldT)$.
We say that $\boldT^{\mathcal{I}}$ is an interpreted term. 

Now we are in a position to prove Lemma \ref{lemma:MonadicSliceTreeAutomata}. For each \msotwo formula $\psi$
in the vocabulary of unit decompositions over $\newslicealphabet(c,\vertexlabel,\edgelabel)$ 
with free variables $\mathcal{X}$ we will construct a tree-automaton $\treeAutomaton(\psi,\interpretedAlphabet)$ 
over $\interpretedAlphabet$ whose slice language  
$\lang(\treeAutomaton(\psi,\interpretedAlphabet))$ consists of all interpreted terms $\boldT^{\mathcal{I}}\in \lang(\interpretedAlphabet)$ 
for which $\boldT \models \psi(\mathcal{X})$ with interpretation $\mathcal{I}$. The tree-automaton 
$\treeAutomaton(\psi,\newslicealphabet(c,\vertexlabel,\edgelabel,\mathcal{X}))$ is constructed inductively with 
respect to the structure of the formula $\psi$.

\paragraph{Base Case} In the base case the formula $\psi$ is one of the predicates 
$x\in X$, $x_1=x_2$, 
$C(x)$, $F_{j,i}(x)$ for $j\in \{0,1,2\}$, $\mathit{Neighbors}(x_1,x_2)$, $\hat{s}(y,x)$, $\hat{t}(y,x)$, $\hat{\vertexlabeling}(x,a)$ or 
$\hat{\edgelabeling}(y,b)$. Below, we describe the behavior of the tree-automaton $\treeAutomaton = \treeAutomaton(\psi,\interpretedAlphabet)$
when $\psi$ is each of these predicates. If $x$ is a vertex (edge) variable in $\psi$, then we say that an interpreted term 
$\boldT^{\mathcal{I}}$ passes the singleton test with respect to $x$ if there exists a unique position $p\in \positions(\boldT)$ 
and a unique vertex (edge) $u$ in $\boldT[p]$ such that $\mathcal{I}(p)_{x,u}=1$. We note that this condition can be easily 
checked by a tree-automaton over $\newslicealphabet(c,\vertexlabel,\edgelabel,\mathcal{X})$. Intuitively, $\boldT^{\mathcal{I}}$ passes the singleton 
test with respect to $x$ if precisely one vertex (edge) of some slice of $\boldT$ is assigned to $x$.

\begin{enumerate}
	\item If $\psi \equiv (x_1 = x_2)$ where $x_1$ and $x_2$ are vertex (edge) variables, 
		then $\treeAutomaton$ accepts $\boldT^{\mathcal{I}}$ if and only if $\boldT^{\mathcal{I}}$ passes the singleton test with 
		respect to both $x_1$ and $x_2$,  and for each position $p\in \positions(\boldT)$, and each vertex (edge) 
		$u\in \boldT[p]$, $\mathcal{I}(p)_{x_1,u} = \mathcal{I}(p)_{x_2,u}$.
 	\item If $\psi \equiv (X_1 = X_2)$ where $X_1$ and $X_2$ are vertex (edge) set variables, 
		then $\treeAutomaton$ accepts $\boldT^{\mathcal{I}}$ if and only if $\mathcal{I}(p)_{x_1,u} = \mathcal{I}(p)_{x_2,u}$ 
		for each position $p\in \positions(\boldT)$, and each vertex (edge) $u\in \boldT[p]$. 
	\item If $\psi \equiv x\in X$ where $x$ is a vertex (edge) variable and $X$ is a vertex (edge) set 
		variable then $\treeAutomaton$ accepts $\boldT^{\mathcal{I}}$ if and only if  $\boldT^{\mathcal{I}}$ passes the singleton test with 
		respect to $x$ and for each position $p\in \positions(\boldT)$, and each vertex (edge) $u\in \boldT[p]$, 
		$\mathcal{I}(p)_{x,u}=1$ implies that $\mathcal{I}(p)_{X,u}=1$.
	\item If $\psi \equiv \hat{s}(y,x)$ (resp. $\psi = \hat{t}(y,x)$) then $\treeAutomaton$ accepts $\boldT^{\mathcal{I}}$ if and only if $\boldT^{\mathcal{I}}$ passes the 
		singleton test with respect to both $y$ and $x$, and there exists a position $p\in \positions(\boldT)$, a vertex $v\in \boldT[p]$ and 
		an edge $e$ in $\boldT[p]$ such that $v$ is the source (target) of $e$ and $\mathcal{I}(p)_{x,v}=1$ and $\mathcal{I}(p)_{y,e}=1$. 
	\item  If $\psi \equiv \hat{\vertexlabeling}(x,a)$  then $\treeAutomaton$ accepts $\boldT^{\mathcal{I}}$ if and only if $\boldT^{\mathcal{I}}$ passes 
		the singleton test with respect to $x$ and there is a position $p\in \positions(\boldT)$ and a vertex $v\in \boldT[p]$ such that 
		$\mathcal{I}(p)_{x,v}=1$ and $v$ is labeled with $a$. 
	\item  If $\psi \equiv \hat{\edgelabeling}(y,b)$  then $\treeAutomaton$ accepts $\boldT^{\mathcal{I}}$ if and only if $\boldT^{\mathcal{I}}$ passes 
		the singleton test with respect to $y$ and there is a position $p\in \positions(\boldT)$ and an edge $e\in \boldT[p]$ such that 
		$\mathcal{I}(p)_{y,e}=1$ and $e$ is labeled with $b$. 
	\item If $\psi \equiv C(x)$ (resp. $\psi\equiv F_{i,j}(x)$) then $\treeAutomaton$ accepts $\boldT^{\mathcal{I}}$ if and only if $\boldT^{\mathcal{I}}$ 
		passes the singleton test with 	respect to $x$ and there exists a position $p\in \positions(\boldT)$ and a vertex $v\in \boldT[p]$ such that 
		$\mathcal{I}(p)_{x,v}=1$ and $v$ is the center vertex of $\boldT[p]$ (resp. $v$ is the vertex $[j,i]$ at the $j$-th frontier of $\boldT[p]$). 
	\item If $\psi \equiv \mathit{Neighbors}(x_1,x_2)$ then $\treeAutomaton$ accepts $\boldT^{\mathcal{I}}$ if and only if $\boldT^{\mathcal{I}}$ passes 
		the singleton test with respect to both $x_1$ and $x_2$, and there exists a position $p\in \positions(\boldT)$, a number $j\in \{1,2\}$,
		a vertex $v\in \boldT[p]$ and a vertex $v'\in \boldT[pj]$ such that $\mathcal{I}(p)_{x_1,v}=1$ and $\mathcal{I}(pj)_{x_2,v'} = 1$.
\end{enumerate}

\paragraph{Disjunction, conjunction and negation}
The three boolean operations $\vee,\wedge,\neg$ are handled using the fact that tree-automata are effectively 
closed under union, intersection and complement (Lemma \ref{lemma:PropertiesOfTreeAutomata}).
Below we let $\treeAutomaton(\interpretedAlphabet)$ be the slice tree-automaton generating the tree slice language $\lang(\interpretedAlphabet)$, 
i.e., the set of all unit decompositions over $\interpretedAlphabet$. 

\begin{equation}
\begin{array}{c}
\treeAutomaton(\psi \vee \psi',\interpretedAlphabet) = \treeAutomaton(\psi ,\interpretedAlphabet) \cup \treeAutomaton(\psi',\interpretedAlphabet) \\
\\
\treeAutomaton(\psi \wedge \psi',\interpretedAlphabet) = \treeAutomaton(\psi ,\interpretedAlphabet) \cap \treeAutomaton(\psi',\interpretedAlphabet) \\
\\
\treeAutomaton(\neg \psi,\interpretedAlphabet) = \overline{\treeAutomaton(\psi ,\interpretedAlphabet)} \cap 
\treeAutomaton(\interpretedAlphabet)\\
\end{array}
\end{equation}

Observe that in the definition of $\treeAutomaton(\neg \psi,\interpretedAlphabet)$, the 
intersection with the tree-automaton $\treeAutomaton(\interpretedAlphabet)$ guarantees 
that the language generated by $\treeAutomaton(\neg \psi, \interpretedAlphabet)$ has no
term that is not an unit decomposition.

\paragraph{Existential Quantification}
To eliminate existential quantifiers we proceed as follows: For each variable $X$, 
define the slice projection $\mathit{Proj}_{X}:\interpretedAlphabet\rightarrow \interpretedAlphabetMinusX$  that 
sends each interpreted slice $(\boldS,I)\in \interpretedAlphabet$ to the interpreted slice $(\boldS,I\backslash X)$ in
the slice alphabet $\interpretedAlphabetMinusX$
where $I\backslash X$ denotes the matrix $I$ with the row corresponding to the variable $X$ deleted. Subsequently, 
we extend $\mathit{Proj}_{X}$ homomorphically to terms by setting $\mathit{Proj}_X(\boldT)[p] = \mathit{Proj}_X(\boldT[p])$ 
to each position $p$ in $\positions(\boldT)$. Finally, we extend $\mathit{Proj}_X$ to tree slice languages by applying 
the projection to each term of the language.
Then we set $$\treeAutomaton(\exists X \psi(\mathcal{X}), \interpretedAlphabetMinusX) = \mathit{Proj}_{X}(\treeAutomaton(\psi(\mathcal{X}),\interpretedAlphabet)).$$

We note that if $\psi$ is a sentence, i.e., a formula without free variables, then by the end of this 
inductive process all variables occurring in $\psi$ will have been projected. In this way, the slice language 
$\lang(\treeAutomaton(\psi,\newslicealphabet(c,\vertexlabel,\edgelabel)))$
will consist precisely of the unit decompositions $\boldT$ over $\newslicealphabet(c,\vertexlabel,\edgelabel)$ 
for which $\boldT \models \psi$. 

To finalize the proof of Lemma \ref{lemma:MonadicSliceTreeAutomata}, let $\varphi$ be a sentence in the vocabulary of $(\vertexlabel,\edgelabel)$-labeled digraphs.
We apply Proposition \ref{proposition:TranslationFormula} to translate $\varphi$ into a sentence $\psi$ in the vocabulary of unit decompositions 
over $\newslicealphabet(c,\vertexlabel,\edgelabel)$ such that $\boldT\models \psi$ 
if and only if $\composedT \models \varphi$. By setting $\treeAutomaton(\varphi,c) = \treeAutomaton(\psi,\newslicealphabet(c,\vertexlabel,\edgelabel))$
we have that $\treeAutomaton(\varphi,c)$ accepts $\boldT$ if and only if $\boldT\models \psi$ if and only if $\composedT\models \varphi$. This concludes 
the proof of Lemma \ref{lemma:MonadicSliceTreeAutomata}. 
$\square$

\section{Proof of Theorem \ref{theorem:MainTheoremDirectedTreewidth}}
\label{section:ProofOfMainTheorem}

In this section we will prove Theorem \ref{theorem:MainTheoremDirectedTreewidth}, which states that given an 
\msotwo sentence $\varphi$, a digraph $G$ of directed treewidth $w$ and a number $k\in \N$ one can count in polynomial 
time the number of subgraphs of $G$ that are the union of $k$ directed paths, satisfy $\varphi$, and have prescribed
size $l$ and weight $\alpha$. First, in Subsection \ref{subsection:PrescribedSize} we will show how to construct slice tree-automata 
representing digraphs of a prescribed size. Subsequently, in Subsection \ref{subsection:PrescribedWeight} we will 
show how to construct slice tree-automata representing digraphs of a prescribed weight. In Subsection \ref{subsection:InverseHomomorphicImage}
we will define a suitable  notion of inverse homomorphic image for slice languages. Using the results in these three subsections 
in conjunction with the tree-automaton $\treeAutomaton(\varphi,k,z)$ of $\mbox{Theorem \ref{theorem:MonadicSliceTreeAutomataZSaturated},}$ 
we will show, in Subsection \ref{subsection:Restriction}, how to construct a $z$-saturated slice tree-automaton $\treeAutomaton(\varphi,k,z,l,\alpha)$ representing all digraphs 
that are the union of $k$ directed paths, satisfy $\varphi$, and have prescribed size $l$ and $\mbox{weight $\alpha$}$.
The proof of Theorem \ref{theorem:MainTheoremDirectedTreewidth}, which will be detailed in Subsection \ref{subsection:TheProof},
will follow by plugging $\treeAutomaton(\varphi,k,z,l,\alpha)$ into Theorem \ref{theorem:CountingSubgraphs}.

\subsection{Generating Digraphs of a Prescribed Size}
\label{subsection:PrescribedSize}

In this subsection we will show that given an arbitrary slice alphabet $\newslicealphabet$ of unit slices, and a number $l$, one can construct 
a deterministic tree-automaton generating precisely the unit decompositions in $\lang(\newslicealphabet)$
that give rise to digraphs with $l$ vertices. Let $\Z_m$ denote the integers with addition modulo $m$. Consider
the following weighting function $\automataweightingfunction_{\Z_m}:\newslicealphabet \rightarrow \Z_{m}$:

\begin{equation}
\label{equation:WeightingCounting}
\automataweightingfunction_{\Z_m}(\boldS)= 
\left\{\begin{array}{l}
0 \mbox{ if $\boldS$ has empty center,}\\
1 \mbox{ if $\boldS$ has a center vertex.}\\
\end{array}
\right.
\end{equation}

Recall from Section \ref{subsection:WeightedTerms} that given a term $\boldT\in \terms(\newslicealphabet)$, the 
weight of $\boldT$ is defined as 

\begin{equation}
\label{equation:WeightTerm}
\automataweightingfunction_{\Z_m}(\boldT) = \sum_{p\in \positions(\boldT)} \automataweightingfunction_{\Z_m}(\boldT[p]).
\end{equation}

In other words, the weight of $\boldT$ is simply the sum of the weights of all slices occurring in $\boldT$. 
One can readily check that $\boldT$ has weight 
$\automataweightingfunction_{\Z_m}(\boldT) = l$ if and only there are $l\;(\mathrm{mod}\; m)$ slices in $\boldT$ with non-empty center. 
In particular, if $\boldT$ is a unit decomposition in $\lang(\newslicealphabet)$, 
then $\automataweightingfunction_{\Z_m}(\boldT)$ is the number of vertices in the digraph $\composedT$ represented by $\boldT$, modulo $m$.  

\begin{observation}
\label{observation:CountingVertices}
Let $\boldT$ be a unit decomposition in $\lang(\newslicealphabet)$. 
Then $$\automataweightingfunction_{\Z_m}(\boldT)= |\!\composedT\!|\; (\mathrm{mod}\;m).$$ 
\end{observation}

Recall from Lemma \ref{lemma:WeightedTerms} that if $\newslicealphabet$ is a slice alphabet,
 $\automataweightingfunction:\newslicealphabet \rightarrow \Xi$ is a weighting function on $\newslicealphabet$ 
and $a\in \Xi$, then the automaton $\treeAutomaton(\newslicealphabet,\automataweightingfunction,a)$ 
generates precisely the set of terms $\boldT\in \terms(\newslicealphabet)$ whose weight is $\automataweightingfunction(\boldT)=a$.
By setting $\Xi = \Z_m$, $\automataweightingfunction = \automataweightingfunction_{\Z_{m}}$ and $a=l$, the tree-automaton 
$\treeAutomaton(\newslicealphabet,\automataweightingfunction_{\Z_m},l)$ generates the set of all terms over 
$\newslicealphabet$ which have $l\;(\mathrm{mod}\;m)$ slices with non-empty center. Let $\treeAutomaton(\newslicealphabet)$ be the 
slice tree-automaton generating the set of all unit decomposition over  $\newslicealphabet$. 
Then for each $l\in \{0,...,m-1\}$, 

\begin{equation*}
\label{equation:CountingVerticesB}
\lang(\treeAutomaton(\newslicealphabet,\automataweightingfunction_{\Z_m},l) \cap \treeAutomaton(\newslicealphabet))
=
\{\boldT\in \lang(\newslicealphabet)\;|\; |\!\composedT\!| = l\; (\mathrm{mod}\; m)\}.
\end{equation*} 

In other words $\treeAutomaton(\newslicealphabet,\automataweightingfunction_{\Z_m},l) \cap \treeAutomaton(\newslicealphabet)$ 
generates all unit decompositions over $\newslicealphabet$ whose corresponding digraph has $l\;(\mathrm{mod}\; m)$ vertices. 

\subsection{Generating Digraphs of a Prescribed Weight}
\label{subsection:PrescribedWeight}

Let $G = (V,E,\vertexlabeling,\edgelabeling)$ be a $(\vertexlabel,\edgelabel)$-labeled digraph and $\graphweightingfunction:E\rightarrow \Omega$
be a function that weights the edges in $E$ with elements from a finite semigroup $\Omega$.  We say that the pair 
$(G,\mu)$ is a weighted digraph. Alternatively, we can view $(G,\mu)$ as the digraph $(V,E,\vertexlabeling,\edgelabeling\times \graphweightingfunction)$
where $\edgelabeling\times \graphweightingfunction:E\rightarrow \edgelabel\times \Omega$ is a function that labels each edge 
$e\in E$, with the element  $[\edgelabeling\times \graphweightingfunction](e) = (\edgelabeling(e),\graphweightingfunction(e))$.
In this way we consider that unit decompositions of weighted digraphs are formed with elements of the slice alphabet $\newslicealphabet(c,q,\vertexlabel,\edgelabel\times \Omega)$. 
We insist in having two label sets $\edgelabel$ and $\Omega$ because while we consider that the set $\edgelabel$ is fixed, 
the set $\Omega$ may vary with the input digraph.

For a unit decomposition $\boldT$ over $\newslicealphabet(c,q,\vertexlabel,\edgelabel\times \Omega)$ we 
let $\graphweightingfunction(\composedT)$ be the sum of the weights of all edges in the digraph $\composedT$.  
Let $\boldS$ be a slice in $\newslicealphabet(c,q,\vertexlabel,\edgelabel\times \Omega)$ and let $E$ be the edge set 
of $\boldS$. We denote by $E_{\mathit{out}}$ the set of all edges that have an endpoint 
in the out-frontier of $\boldS$ and the other endpoint in the center of $\boldS$. We denote by $E_{\mathit{in}}$
the set of edges whose endpoints lie in distinct in-frontiers of $\boldS$. 
Let  $\automataweightingfunction_{\Omega}:\newslicealphabet(c,q,\vertexlabel,\edgelabel\times \Omega)\rightarrow \Omega$ be 
a weighting function on $\newslicealphabet(c,q,\vertexlabel,\edgelabel\times \Omega)$ that associates 
with each slice $\boldS\in \newslicealphabet(c,q,\vertexlabel, \edgelabel\times \Omega)$ the value 
$$\automataweightingfunction_{\Omega}(\boldS) = \sum_{e\in E_{\mathit{out}}}\graphweightingfunction(e) - \sum_{e\in E_{\mathit{in}}}\graphweightingfunction(e).$$ 
Note that edges that have only one endpoint at an in-frontier of $\boldS$ do not have their weights counted neither positively, 
nor negatively. The weight of a unit decomposition $\boldT$ over the slice alphabet $\newslicealphabet(c,q,\vertexlabel,\edgelabel\times \Omega)$ is defined as

\begin{equation}
\label{equation:Weight}
\automataweightingfunction_{\Omega}(\boldT) = \sum_{p\in \positions(\boldT)} \automataweightingfunction_{\Omega}(\boldT[p]).
\end{equation}

The next proposition says that the weight $\automataweightingfunction_{\Omega}(\boldT)$ of 
$\boldT$ is equal to the weight $\graphweightingfunction(\composedT)$ of the digraph $\composedT$ represented by $\boldT$.

\begin{proposition}
\label{proposition:CountingWeight}
Let $\boldT$ be a unit decomposition in $\lang(\newslicealphabet(c,q,\vertexlabel,\edgelabel \times \Omega))$. 
Then 
\begin{equation}
\label{equation:EqualityWeight}
\automataweightingfunction_{\Omega}(\boldT) = \graphweightingfunction(\composedT).
\end{equation}
\end{proposition}
\begin{proof}
Recall that each edge $e_K$ in the digraph $\composedT$ represented by $\boldT$ is specified by a sliced edge sequence 
$K\equiv (p_1,a_{1},e_{1},b_{1})(p_2,a_{2},e_{2},b_{2})...(p_n,a_{n},e_{n},b_{n})$, where $e_{i}$ is
the sliced part of $e_K$ lying at slice $\boldT[p_i]$. Recall that by definition, for each $i\in \{1,...,n\}$, 
the weight of $e_{i}$ in $\boldT[p_i]$ is equal to the weight of $e_K$ in $\composedT$. We claim that the overall contribution of the weights of 
the edges in $K$ to the sum in Equation \ref{equation:Weight} is equal to the weight of $e_K$. This claim implies 
Equation \ref{equation:EqualityWeight}. There are three cases to be considered. If $p_1$ is a descendant of $p_n$, then $e_{1}$ is the only 
sliced part of $e_K$ whose weight contributes positively to the sum in Equation \ref{equation:Weight}. The 
weights of all other sliced parts $e_{2},...,e_{n}$ are not counted at all. This happens because,
in this case, $e_1$ is the only edge of $K$ that has a center vertex and an out-frontier vertex as endpoints. 
All other edges of $K$ have one endpoint in some in-frontier and another endpoint in the center or out-frontier, and 
for this reason their weights are not counted.  
Analogously, if $p_n$ is a descendant of $p_1$, then $e_{n}$ is the only 
sliced part of $e_K$ that has its weight contributed positively to the sum in $\mbox{Equation \ref{equation:Weight}.}$ 
The weights of all other sliced parts $e_{1},...e_{{n-1}}$ are not counted at all. Finally, if neither $p_1$ is a
descendant of $p_n$, nor $p_n$ is a descendant of $p_1$, then both sliced parts $e_{1}$ and $e_{n}$ have 
their weights contributed positively to the sum in Equation \ref{equation:Weight}.  Nevertheless, in this case there exists some $k$, with $1<k<n$ such that 
$e_{k}$ has both of its endpoints in distinct in-frontiers of $\boldT[p_k]$.
Indeed, $p_k$ is the position farthest away from the root with the property that both $p_1$ and $p_n$ are descendants of $p_k$. 
Therefore we have that the weight of $e_k$ is counted negatively. 
The weights of all other edges $e_2...e_{k-1}$ and $e_{k+1}...e_{n-1}$ are not counted at all, since each of these edges 
have one endpoint in some in-frontier and another endpoint at an out-frontier. This proves our claim.
\end{proof}

In view of Proposition \ref{proposition:CountingWeight}, if $\treeAutomaton(\newslicealphabet,\automataweightingfunction_{\Omega},\alpha)$ is the tree-automaton of Lemma 
\ref{lemma:WeightedTerms} in which 
$\automataweightingfunction = \automataweightingfunction_{\Omega}$ and $a=\alpha$, then the slice language 
of $\treeAutomaton(\newslicealphabet,\automataweightingfunction_{\Omega},\alpha) \cap \treeAutomaton(\newslicealphabet)$ 
is the set of all unit decompositions over $\newslicealphabet$ which represent a digraph of weight $\alpha$.
More precisely, 

\begin{equation}
\label{equation:CountingSizeWeightB}
\lang(\treeAutomaton(\newslicealphabet,\automataweightingfunction_{\Omega},\alpha) \cap \treeAutomaton(\newslicealphabet)) 
=  \{\boldT\in \lang(\newslicealphabet)\;|\; \graphweightingfunction(\composedT) = \alpha \}.
\end{equation}

\subsection{Inverse Homomorphic Image of Slice Languages}
\label{subsection:InverseHomomorphicImage}

Let $\projection:\newslicealphabet \rightarrow \newslicealphabet'$ be a slice projection such as defined in 
Section \ref{subsection:SliceProjection}. If $\lang$ is a slice language over $\newslicealphabet'$ then the inverse homomorphic 
image $\projection^{-1}(\lang)$, as defined in Section \ref{subsection:PropertiesOfTreeAutomata}, is not necessarily 
a slice language since for a unit decomposition $\boldT\in \lang$, the inverse set $\projection^{-1}(\boldT)$ consisting of all terms whose image is $\boldT$ 
may have some terms over $\newslicealphabet$ that are not unit decompositions. To fix this we intersect $\projection^{-1}(\boldT)$
with the slice language $\lang(\newslicealphabet)$ of all unit decompositions over $\newslicealphabet$. More precisely, if 
$\boldT$ is a unit decomposition over $\newslicealphabet'$ then we define $\inverse(\projection,\boldT) = \projection^{-1}(\boldT)\cap \lang(\newslicealphabet)$. 
Going further, if $\lang$ is a slice language over $\newslicealphabet'$ then  

\begin{equation}
\label{equation:InverseImageUnweighting}
\inverse(\projection,\lang) = \bigcup_{\boldT\in \lang} \inverse(\projection,\boldT)  = \projection^{-1}(\lang) \cap \lang(\newslicealphabet). 
\end{equation}

For instance, if $\normalizingprojection_{c,q}:\newslicealphabet(c,q,\vertexlabel,\edgelabel)\rightarrow \newslicealphabet(c,\vertexlabel,\edgelabel)$
is a normalizing projection and $\boldT$ is a unit decomposition over $\newslicealphabet(c,\vertexlabel,\edgelabel)$, then 
$\inverse(\normalizingprojection_{c,q}, \boldT)$ consists of all unit decompositions that are obtained from $\boldT$ by renumbering the 
vertices on each frontier of each slice of $\boldT$ with numbers from $\{1,...,q\}$ in such a way that the order 
in each frontier is preserved. Therefore $\inverse(\normalizingprojection_{c,q},\lang)$ is the maximal unnormalized slice language
whose image under $\normalizingprojection_{c,q}$ is $\lang$. Note that if $\lang$ is a $z$-saturated slice language over 
$\newslicealphabet(c,\vertexlabel,\edgelabel)$ then $\inverse(\normalizingprojection,\lang)$ is a $z$-saturated slice language 
over $\newslicealphabet(c,q,\vertexlabel,\edgelabel)$. 

Analogously, if $\unweightingprojection_{\Omega}:\newslicealphabet(c,q,\vertexlabel,\edgelabel\times \Omega)\rightarrow \newslicealphabet(c,q,\vertexlabel,\edgelabel)$
is an unweighting projection, and $\boldT$ is a unit decomposition over $\newslicealphabet(c,q,\vertexlabel,\edgelabel)$, then 
$\inverse(\unweightingprojection_{\Omega},\boldT)$ consists of all unit decompositions over $\newslicealphabet(c,q,\vertexlabel,\edgelabel\times \Omega)$ that 
are obtained from $\boldT$ by weighting the edges of each slice in $\boldT$ with elements from $\Omega$ in such a way that gluability 
of slices is preserved. Thus $\inverse(\unweightingprojection_{\Omega},\lang)$ is a slice language consisting of all weighted versions of 
unit decompositions in $\lang$. We note that if $\lang$ is a $z$-saturated slice language over $\newslicealphabet(c,q,\vertexlabel,\edgelabel)$ then 
$\inverse(\unweightingprojection_{\Omega},\lang)$ is a $z$-saturated slice language over $\newslicealphabet(c,q,\vertexlabel,\edgelabel\times \Omega)$. 

\subsection{Restricting $\treeAutomaton(\varphi,k,z)$}
\label{subsection:Restriction}

In Theorem \ref{theorem:MonadicSliceTreeAutomataZSaturated} we showed that given any \msotwo sentence $\varphi$ in the 
vocabulary of $\mbox{$(\vertexlabel,\edgelabel)$-labeled}$ digraphs, and any $z,k\in \N$ one can construct a normalized 
$z$-saturated slice tree-automaton $\treeAutomaton(\varphi,k,z)$ over the slice alphabet
$\newslicealphabet(k\cdot z,\vertexlabel,\edgelabel)$ whose graph language $\lang_{\graph}(\treeAutomaton(\varphi,k,z))$
consists precisely of the digraphs that are the union of $k$ directed paths and satisfy $\varphi$. 
In this section we show how to construct a $z$-saturated tree-automaton $\treeAutomaton(\varphi,k,z,l,\alpha)$
over the slice alphabet $\newslicealphabet(c,q,\vertexlabel,\edgelabel\times \Omega)$ whose graph $\lang_{\graph}(\treeAutomaton(\varphi,k,z,l,\alpha))$
contains only the digraphs in $\lang_{\graph}(\treeAutomaton(\varphi,k,z))$ that have a prescribed size $l$  and prescribed weight 
$\alpha\in \Omega$. If $\varphi$ is an \msotwo sentence in the vocabulary of $(\vertexlabel,\edgelabel)$-labeled digraphs, $G=(V,E)$ is a 
$(\vertexlabel,\edgelabel)$-labeled digraph and $\mu:E\rightarrow \Omega$ is a weighting function, then we say that 
the weighted digraph $(G,\mu)$ satisfies $\varphi$ if $G$ satisfies $\varphi$. In other words, a weighted digraph satisfies an \msotwo
sentence if its unweighted version does. 

\begin{lemma}
\label{lemma:AutomatonSizeWeight}
Let $\varphi$ be an \msotwo sentence over $(\vertexlabel,\edgelabel)$-labeled digraphs, $q,k,z,l,m\in \N$ be positive integers with $l<m$, $q\geq k\cdot z$,
and let $\alpha\in \Omega$. For some computable function $g$, one can construct in time 
$g(\varphi,k,z,|\vertexlabel|,|\edgelabel|)\cdot q^{O(k\cdot z)} \cdot |\Omega|^{O(k\cdot z)}\cdot m^{O(1)}$
a $z$-saturated  slice tree-automaton $\treeAutomaton(\varphi,k,z,l,\alpha)$ over $\newslicealphabet(k\cdot z,q,\vertexlabel,\edgelabel\times \Omega)$ 
such that 
\begin{equation*}
\label{equation:AutomatonPrescribedLenghtWeight}
\begin{array}{lcr}
\lang(\treeAutomaton(\varphi,k,z,l,\alpha))  =  \{\boldT\in \newslicealphabet(k\cdot z,q,\vertexlabel,\edgelabel\times \Omega) \; |\;
 \mbox{ $\composedT\models \varphi$, $\composedT$ is the union of $k$ directed paths,} \\ 
\hspace{8.3cm} \mbox{ $\composedT$ has $l\;(\mathrm{mod}\;m)$ vertices, $\composedT$ has weight $\alpha$}\}. 
\end{array}
\end{equation*}
\end{lemma}
\begin{proof}
Let $\normalizingprojection_{k\cdot z,q}:\newslicealphabet(k\cdot z,q,\vertexlabel,\edgelabel)\rightarrow \newslicealphabet(k\cdot z,\vertexlabel,\edgelabel)$
be a normalizing projection and let 
$\unweightingprojection_{\Omega}:\newslicealphabet(k\cdot z,q,\vertexlabel,\edgelabel\times \Omega) \rightarrow \newslicealphabet(k\cdot z,q,\vertexlabel,\edgelabel)$ 
be an unweighting projection. By the discussion in Section \ref{subsection:InverseHomomorphicImage}, the slice tree-automaton 
\begin{equation}
\label{equation:Inverses}
\inverse(\unweightingprojection_{\Omega},\inverse(\normalizingprojection_{k\cdot z,q}, \treeAutomaton(\varphi,k,z)))
\end{equation}
is a deterministic $z$-saturated tree-automaton over $\newslicealphabet(k\cdot z, q,\vertexlabel,\edgelabel\times \Omega)$ whose graph 
language consists of all weighted versions of digraphs in $\treeAutomaton(\varphi,k,z)$. 
We will restrict the tree-automaton in Equation \ref{equation:Inverses} so that it represents only digraphs with weight $\alpha$ and 
$l\;\mathrm{mod}\;m$ vertices. For simplicity of notation, let  
$\newslicealphabet = \newslicealphabet(k\cdot z ,q,\vertexlabel,\edgelabel\times \Omega)$. Recall that the deterministic tree-automaton 
$\treeAutomaton(\newslicealphabet,\automataweightingfunction_{\Z_m},l) \cap \treeAutomaton(\newslicealphabet)$ constructed in 
Section \ref{subsection:PrescribedSize}  generates precisely the set of unit decompositions over $\newslicealphabet$ that give 
rise to a digraph with $l\;\mathrm{mod}\;m$ vertices. Recall also that the deterministic tree-automaton 
$\treeAutomaton(\newslicealphabet,\automataweightingfunction_{\Omega},\alpha) \cap \treeAutomaton(\newslicealphabet)$  
constructed in Section \ref{subsection:PrescribedWeight} generates precisely the unit decompositions over $\newslicealphabet$
which give rise to digraphs of weight $\alpha$. Therefore, the slice tree-automaton
\begin{equation*}
\label{equation:treeAutomaton}
\treeAutomaton(\varphi,k,z,l,\alpha) = \inverse(\unweightingprojection_{\Omega},\inverse(\normalizingprojection,\treeAutomaton(\varphi,k,z))) \cap 
\treeAutomaton(\newslicealphabet,\automataweightingfunction_{\Z_m},l) \cap \treeAutomaton(\newslicealphabet,\automataweightingfunction_{\Omega},\alpha) \cap 
\treeAutomaton(\newslicealphabet)
\end{equation*}
is a $z$-saturated, deterministic slice tree-automaton over the alphabet $\newslicealphabet$ whose graph language 
$\lang_{\graph}(\treeAutomaton(\varphi,k,z,l,\alpha))$ consists of all digraphs that are the union of $k$ directed paths, 
satisfy $\varphi$, have $l\;(\mathrm{mod}\;m)$ vertices and weight $\alpha$. 

To finalize the proof, we need to estimate the size of $\treeAutomaton(\varphi,k,z,l,\alpha)$. 
By Lemma \ref{lemma:PropertiesOfTreeAutomata}.\ref{lemma:PropertiesOfTreeAutomata:InverseHomomorphism} 
the automaton in Equation \ref{equation:Inverses} can be constructed in time $|\treeAutomaton(\varphi,k,z)|\cdot \newslicealphabet^{O(1)}$.
By Proposition \ref{proposition:InitialAutomaton}, the automaton $\treeAutomaton(\newslicealphabet)$ can be constructed in time 
$O(|\newslicealphabet|)$. By Lemma \ref{lemma:WeightedTerms} the automaton
$\treeAutomaton(\newslicealphabet,\automataweightingfunction_{\Z_m},l)$ can be constructed in time $|\newslicealphabet|\cdot |\Z_m|^{O(1)}$
and the automaton $\treeAutomaton(\newslicealphabet,\automataweightingfunction_{\Omega},\alpha)$ 
can be constructed in time $|\newslicealphabet|\cdot |\Omega|^{O(1)}$.
Therefore, given $\treeAutomaton(\varphi,k,z)$, the tree automaton
$\treeAutomaton(\varphi,k,z,l,\alpha)$ can be constructed in time $|\treeAutomaton(\varphi,k,z)|\cdot |\newslicealphabet|^{O(1)} \cdot |\Omega|^{O(k\cdot z)} \cdot m^{O(1)}$.
Note that the size of the alphabet $\newslicealphabet(k\cdot z,q,\vertexlabel,\vertexlabel\times \Omega)$ is bounded by
$2^{O(k\cdot z \log k\cdot z)}\cdot |\vertexlabel| \cdot |\edgelabel|^{O(k\cdot z)}\cdot |\Omega|^{O(k\cdot z)} \cdot q^{O(k\cdot z)}$.  
Thus, the tree-automaton $\treeAutomaton(\varphi,k,z,l,\alpha)$ can be constructed in time 

$$g(\varphi,k,z,|\vertexlabel|,|\edgelabel|) \cdot q^{O(k\cdot z)} \cdot |\Omega|^{O(k\cdot z)} \cdot m^{O(1)},$$

where $g(\varphi,k,z,|\vertexlabel|,\edgelabel)$ is the time necessary to construct $\treeAutomaton(\varphi,k,z)$ times $2^{O(k\cdot z \log k\cdot z)}$. 
\end{proof}

\subsection{Proof of Theorem \ref{theorem:MainTheoremDirectedTreewidth}}
\label{subsection:TheProof}

The proof of Theorem \ref{theorem:MainTheoremDirectedTreewidth} will follow as 
a corollary of the following theorem, whose proof is obtained by plugging 
the automaton $\treeAutomaton(\varphi,k,z,l,\alpha)$, constructed in Lemma \ref{lemma:AutomatonSizeWeight},
into Theorem \ref{theorem:CountingSubgraphs}. 

\begin{theorem}
\label{theorem:CountingSizeWeight}
Let $\boldT\in \lang(\newslicealphabet(q,\vertexlabel,\edgelabel\times \Omega))$ be a normalized unit decomposition of width $q$ and tree-zig-zag number $z$. 
Let $\varphi$ be an \msotwo sentence in the vocabulary of $(\vertexlabel,\edgelabel)$-labeled digraphs. 
Then for each $k,l\in \N$ and each $\alpha\in \Omega$ one can count in time $f(\varphi,k,z)\cdot \boldT^{O(1)}\cdot q^{O(k\cdot z)}\cdot |\Omega|^{O(k\cdot z)}$
the number of subgraphs $H$ of $\composedT$ simultaneously satisfying the following four properties: 
\begin{enumerate}
	\item \label{CountingSizeWeight-One} $H\models \varphi$,
	\item \label{CountingSizeWeight-Two} $H$ is the union of $k$ directed paths, 
	\item \label{CountingSizeWeight-Three} $H$ has $l$ vertices,
	\item \label{CountingSizeWeight-Four} $H$ has weight $\graphweightingfunction(H) = \alpha$. 
\end{enumerate}
\end{theorem}
\begin{proof}
The proof follows by a combination of Theorem \ref{theorem:CountingSubgraphs} with 
Lemma \ref{lemma:AutomatonSizeWeight}. First, Theorem \ref{theorem:CountingSubgraphs} says that given a $z$-saturated slice tree automaton $\treeAutomaton$ 
over $\newslicealphabet(k\cdot z,q,\vertexlabel,\edgelabel\times \Omega)$, we can count in time $|\boldT|^{O(k\cdot z)}\cdot |\treeAutomaton|^{O(1)}$
the number of subgraphs of $\composedT$ which are isomorphic to some digraph in $\lang_{\graph}(\treeAutomaton)$.
Second by Lemma \ref{lemma:AutomatonSizeWeight}, we can construct a $z$-saturated slice tree-automaton $\treeAutomaton(\varphi,k,z,l,\alpha)$ 
such that a digraph $H$ belongs to the graph language $\lang_{\graph}(\treeAutomaton(\varphi,k,z,l,\alpha))$ if and only if 
$H$ satisfies $\varphi$, is the union of $k$ directed paths, has $l\;(\mathrm{mod}\;m)$ vertices, and weight $\alpha$.
Since any subgraph of $\composedT$ has at most $|\boldT|$ vertices, if we set $m=|\boldT|+1$ and 
$\treeAutomaton = \treeAutomaton(\varphi,k,z,l,\alpha)$, then Theorem \ref{theorem:CountingSubgraphs} 
provides us with an algorithm for counting all the subgraphs of $\composedT$ that 
satisfy Conditions \ref{CountingSizeWeight-One}-\ref{CountingSizeWeight-Four} of the present theorem. 
Since $|\treeAutomaton(\varphi,k,z,l,\alpha)| \leq g(\varphi,k,z,|\vertexlabel|,|\edgelabel|) \cdot q^{O(k\cdot z)}\cdot |\Omega|^{O(k\cdot z)}\cdot m^{O(1)}$, 
and since the label sets $\vertexlabel$ and $\edgelabel$ are fixed, the algorithm runs in time 
$$f(\varphi,k,z)\cdot \boldT^{O(1)}\cdot q^{O(k\cdot z)}\cdot |\Omega|^{O(k\cdot z)},$$
\noindent where $f(\varphi,k,z) = g(\varphi,k,z,|\vertexlabel|,|\edgelabel|)^{O(1)}$ and $|\vertexlabel|$ and $|\edgelabel|$ are treated as constants. 
\end{proof}

Finally, the proof of our main theorem (Theorem \ref{theorem:MainTheoremDirectedTreewidth})  follows as an 
application of Theorem \ref{theorem:CountingSizeWeight}.  

\paragraph{\textbf{Proof of Theorem \ref{theorem:MainTheoremDirectedTreewidth}}}
Let $G=(V,E,\vertexlabeling,\edgelabeling\times \graphweightingfunction)$ be a $(\vertexlabel,\edgelabel\times \Omega)$-labeled digraph of {\em directed}
treewidth $w$. By Theorem  \ref{theorem:ConstructionGoodArborealDecomposition}, 
we can construct in time $|G|^{O(w)}$ a good arboreal decomposition $\arborealdecomposition$ of $G$ of width $O(w)$. 
By Theorem \ref{theorem:ComparisonWithOtherMeasures}, from $\arborealdecomposition$ we can construct an olive-tree decomposition 
$\olivetreedecomposition$ of tree-zig-zag number $z$ for some $z\leq 9w+18$.
Using Proposition \ref{proposition:OliveTreeDecompositionUnitDecomposition} we can use $\olivetreedecomposition$ 
to construct a normalized unit decomposition $\boldT$ over $\newslicealphabet(q,\vertexlabel,\edgelabel\times \Omega)$ such that $\boldT$
has  tree-zig-zag number $z$ and $\composedT=G$. Therefore given an \msotwo sentence $\varphi$, and positive integers $k,z\in \N$, we  can 
apply Theorem \ref{theorem:CountingSizeWeight} to count in time $f(\varphi,k,z)\cdot |\boldT|^{O(1)}\cdot q^{O(k\cdot z)}\cdot |\Omega|^{O(k\cdot z)}$,
the number of subgraphs of $\composedT$ that are the union of $k$ directed paths, satisfy $\varphi$, have $l$ vertices and weight 
$\alpha$. Since $|\boldT|\leq |G|^{O(1)}$, $q\leq |E|$, and by assumption $|\Omega|\leq |G|^{O(1)}$, we have that the total running time of 
the algorithm is $f(\varphi,k,z)\cdot |G|^{O(k\cdot z)}$. Since $z\leq 9w+18$, the running time of the algorithm stated in terms 
of directed treewidth is $f(\varphi,k,z)\cdot |G|^{O(k(w+1))}$.  Here we write $w+1$ in the exponent, to emphasize that the treewidth of 
$G$ can be $0$. 
$\square$

\section{Conclusion}
\label{section:Conclusion}

In this work we devised the first algorithmic metatheorem for digraphs of constant directed treewidth. 
We showed that most of the previously known positive algorithmic results for this class of digraphs 
can be re-stated in terms of our metatheorem. Additionally, we showed 
how to use our metatheorem to provide polynomial time algorithms for
two classes of counting problems whose polynomial-time solvability is not 
implied by previously existing techniques. Namely, for each fixed $k$, we showed how to count in polynomial time on digraphs of 
constant directed treewidth, the number of minimum spanning  strong subgraphs that are the union of $k$ directed paths, 
and the number of subgraphs that are the union of $k$ directed paths and satisfy a given minor closed property.

To prove our main theorem we introduced two new theoretical tools which in our opinion are of 
independent interest. The first, the tree-zig-zag number of a digraph, is a new directed width 
measure that is at most a constant times its directed treewidth. Concerning this measure, 
we leave open the problem of determining whether there exist families of digraphs of constant 
tree-zig-zag number but unbounded directed treewidth, or whether the directed treewidth of a digraph 
is always bounded by a function of its tree-zig-zag number.  The second theoretical tool we have introduced 
is the notion of $z$-saturated tree-automata. 
By Theorem \ref{theorem:CountingSubgraphs}, given a digraph $G$ of constant directed treewidth, and 
a $z$-saturated tree-automaton $\treeAutomaton$ generating only digraphs that are the union of $k$ directed paths, 
one can count the number of subgraphs of $G$ that are isomorphic to some digraph in $\lang_{\graph}(\treeAutomaton)$. It would be 
interesting to study ways of constructing $z$-saturated tree-automata without the help of \msotwo logic. 
Such a construction would open the possibility of using Theorem \ref{theorem:CountingSubgraphs} to solve counting problems, 
on digraphs of constant directed treewidth, that may not be approachable via Theorem \ref{theorem:MainTheoremDirectedTreewidth}.

\section{Acknowledgements.}
The author is currently supported by the European Research Council, ERC grant agreement 339691, within the context 
of the project Feasibility, Logic and Randomness (FEALORA).

\bibliographystyle{abbrv}

\end{document}